\documentclass[reprint,aps,twocolumn,superscriptaddress,showpacs,notitlepage,nofootinbib,floatfix,10pt]{revtex4-2}

\usepackage{amsmath,amssymb,amsfonts,amstext,amsthm}
\usepackage{graphicx,breakurl,color,setspace}
\usepackage{mathtools}
\usepackage{dsfont}
\usepackage[T1]{fontenc}
\usepackage[latin9]{inputenc}
\usepackage[normalem]{ulem}
\usepackage[unicode=true,pdfusetitle,pdfborder={0 0 1}]{hyperref}
\usepackage{xcolor}
\usepackage{ragged2e}
\hypersetup{
    colorlinks,
    linkcolor=[RGB]{243 68 110},
    citecolor=[RGB]{0 160 63},
    urlcolor=[RGB]{57 115 239}
}

\usepackage{upgreek}
\usepackage{bm}
\usepackage{mathrsfs}
\usepackage[linesnumbered, ruled, vlined]{algorithm2e}
\SetArgSty{textnormal}
\SetKw{Return}{Return}
\RestyleAlgo{ruled}
\theoremstyle{definition}
\newtheorem{theorem}{Theorem}

\newtheorem{definition}[theorem]{Definition}

\newtheorem{lemma}[theorem]{Lemma}

\newtheorem{corollary}[theorem]{Corollary}
\theoremstyle{remark}

\renewcommand{\thefigure}{\arabic{figure}}

\begin{document}

\title{Hierarchical Quantum Logical Processor with Amortized Long-Range Connectivity}
\author{Zi-Han Chen}

\email{czh007@mail.ustc.edu.cn}

\affiliation{
Hefei National Research Center for Physical Sciences at the Microscale and Department of Modern Physics, New Cornerstone Science Laboratory, University of Science and Technology of China, Hefei 230026, China}
\affiliation{
Shanghai Research Center for Quantum Science and CAS Center for Excellence in Quantum Information and Quantum Physics, University of Science and Technology of China, Shanghai 201315, China}
\affiliation{
Hefei National Laboratory, University of Science and Technology of China, Hefei 230088, China}

\author{Ming-Cheng Chen}

\affiliation{
Hefei National Research Center for Physical Sciences at the Microscale and Department of Modern Physics, New Cornerstone Science Laboratory, University of Science and Technology of China, Hefei 230026, China}
\affiliation{
Shanghai Research Center for Quantum Science and CAS Center for Excellence in Quantum Information and Quantum Physics, University of Science and Technology of China, Shanghai 201315, China}
\affiliation{
Hefei National Laboratory, University of Science and Technology of China, Hefei 230088, China}

\author{Chao-Yang Lu}

\affiliation{
Hefei National Research Center for Physical Sciences at the Microscale and Department of Modern Physics, New Cornerstone Science Laboratory, University of Science and Technology of China, Hefei 230026, China}
\affiliation{
Shanghai Research Center for Quantum Science and CAS Center for Excellence in Quantum Information and Quantum Physics, University of Science and Technology of China, Shanghai 201315, China}
\affiliation{
Hefei National Laboratory, University of Science and Technology of China, Hefei 230088, China}

\author{Jian-Wei Pan}

\affiliation{
Hefei National Research Center for Physical Sciences at the Microscale and Department of Modern Physics, New Cornerstone Science Laboratory, University of Science and Technology of China, Hefei 230026, China}
\affiliation{
Shanghai Research Center for Quantum Science and CAS Center for Excellence in Quantum Information and Quantum Physics, University of Science and Technology of China, Shanghai 201315, China}
\affiliation{
Hefei National Laboratory, University of Science and Technology of China, Hefei 230088, China}

\date{\today}

\begin{abstract}
    High qubit overhead is a key bottleneck for fault-tolerant quantum computation.  Quantum low-density parity-check (qLDPC) codes offer high encoding efficiency but typically require non-local connectivity in every syndrome extraction cycle, incurring additional physical errors and implementation complexity. We introduce a  hierarchical logical processor (HLP) architecture that implements a high-rate quantum CSS code with distance-$d_{0}$ rotated surface codes (RSC), allowing long-range connectivity to be used only once every $\Theta(d_{0})$ rounds of surface-code syndrome extraction. HLPs introduce elongated RSC patches called shuttle buses. Using transversal hybrid-unit CNOT gates, a single shuttle bus can simultaneously couple to multiple standard RSC patches. This capability enables efficient level-1 syndrome extraction with suppressed level-1 error correlations and supports highly parallel logical Pauli measurements. We perform circuit-level simulations of several concrete HLP constructions and benchmark both logical memory and logical Pauli measurement performance. At a physical error rate of $10^{-3}$, a hierarchical logical processor can achieve $3$\textendash{}$4$ times higher qubit efficiency than the rotated surface code and $20$\textendash{}$30$ times shorter logical error-correction cycle times than the yoked surface code~\cite{gidney_yoked_2025}.
\end{abstract}

\maketitle		

\noindent
Utility-scale quantum algorithms require quantum operations at extremely low error rates that are orders of magnitude lower than state-of-the-art error rates on physical qubits~\cite{gidney_how_2025,kivlichan_improved_2020,campbell_early_2022}. Fault-tolerant quantum computation (FTQC) aims to bridge this gap between the desired error rate and physical error rates by encoding logical quantum information into quantum error correction (QEC) codes. The rotated surface code (RSC) is a leading candidate for experimental QEC, since the RSC has simple stabilizer patterns and a high threshold~\cite{dennis_topological_2002,horsman_rotated_surgery_2012}. However, a marked shortcoming of the RSC is its low qubit efficiency\textemdash{}a single patch of distance-$d$ RSC encodes only one logical qubit with $d^2$ physical qubits. Recent constructions of quantum low-density parity-check (qLDPC) codes~\cite{gottesman_fault_tolerant_2014,breuckmann_quantum_2021,bravyi_high_threshold_2024,tillich_quantum_2013,xu_constant_overhead_2024,panteleev_quantum_2022,cain_shors_2026} offer significantly higher encoding efficiency than the RSC. However, their direct implementations generally require non-local couplings in each round of syndrome extraction (SE), with the required coupling range depending on the specific code construction.

Reconfigurable neutral atom arrays, a rapidly evolving quantum platform supporting long-range connectivity~\cite{bluvstein_quantum_2022,bluvstein_logical_2024,reichardt_fault_tolerant_2025}, are promising for realizing these qLDPC codes. Nevertheless, long-range gates, implemented via atom shuttling, can introduce additional physical errors and implementation complexity~\cite{bluvstein_quantum_2022,bluvstein_logical_2024}. Thus, an important design objective is to reduce qubit overhead while limiting the use of long-range connectivity. For instance, measuring only a fraction of stabilizer generators in each SE round is an effective approach to mitigate the cost of long-range connectivity for implementing qLDPC codes~\cite{berthusen_partial_2024,berthusen_toward_2025,berthusen_adaptive_2025}. A few recent protocols propose using native $\mathrm{iSWAP}$ gates on superconducting platforms~\cite{mcewen_relaxing_2023,eickbusch_demonstration_2025,chen_efficient_2025} to extract syndromes and perform routing simultaneously, thereby compiling the required non-local connectivity of certain qLDPC codes into SE circuits with only nearest-neighbor coupling~\cite{geher_directional_2025,nixon_vine_2026,gu_nearest_neighbour_2026}. Alternatively, code concatenation\textemdash{}concatenating a level-1 QEC code on top of a level-0 RSC~\cite{gidney_yoked_2025,pattison_hierarchical_2025}\textemdash{}allows a long time window for measuring non-local level-1 stabilizers, thereby amortizing the demand on non-local connectivity. In particular, a recently proposed quantum memory architecture, yoked surface codes~\cite{gidney_yoked_2025}, implements specific high-rate level-1 codes on the RSC with only nearest-neighbor connectivity and achieves approximately three times the qubit efficiency of the RSC. These recent developments demonstrate promising routes toward high-rate quantum memories under restricted connectivity, while motivating the exploration of logical operations under similar hardware constraints.  

In this work, we propose a new FTQC architecture, the hierarchical logical processor (HLP), which implements a high-rate quantum CSS code concatenated with the RSC. HLPs combine high encoding efficiency with infrequent use of long-range gates for level-1 SE, while supporting parallel logical operations. HLPs rely on hybrid-unit CNOT gates that couple an elongated RSC patch (a shuttle bus) to multiple RSC patches (cores) simultaneously. Cores and shuttle buses function as level-1 data and ancilla qubits, respectively. Using hybrid-unit CNOT gates, we construct a unified class of readout gadgets for measuring long level-1 Pauli operators. The readout-gadget framework supports both memory and computation\textemdash{}readout gadgets measure level-1 stabilizers for level-1 SE, while sequences of the same readout gadgets implement fault-tolerant logical Pauli measurements. By appropriately spacing the hybrid-unit CNOT gates, long-range connectivity is required only once every $\Theta(d_0)$ level-0 SE rounds. 

We develop an approximate error-reduction framework to prove the fault tolerance of the constructions above. Under a phenomenological noise model, we prove that appropriate temporal separation of readout gadgets suppresses correlations induced by hybrid-unit CNOT gates, yielding a hookless level-1 error model with effectively local stochastic errors. This result does not require the level-1 code to be LDPC. We further show that logical measurement sequences for logical $X$ and $Z$ operators achieve full level-1 distance and can be densely interleaved to perform multiple (mutually commuting) logical measurements in parallel. We perform circuit-level simulations of concrete HLP constructions to support these theoretical predictions. At a physical error rate of $10^{-3}$, an HLP based on the $[[256,194,4]]$ code achieves $3$\textendash{}$4$ times higher qubit efficiency than the RSC and $20$\textendash{}$30$ times shorter logical error-correction cycle times than the corresponding yoked surface code architecture~\cite{gidney_yoked_2025}. We further extend the readout framework to general logical Pauli measurements and joint logical measurements between an HLP and external RSC patches. With relaxed long-range-connectivity requirements, beyond-RSC encoding efficiency, and parallel logical measurement capabilities, HLPs provide a promising building block for large-scale FTQC architectures.

\begin{figure}
    \centering
    \includegraphics[width=\columnwidth]{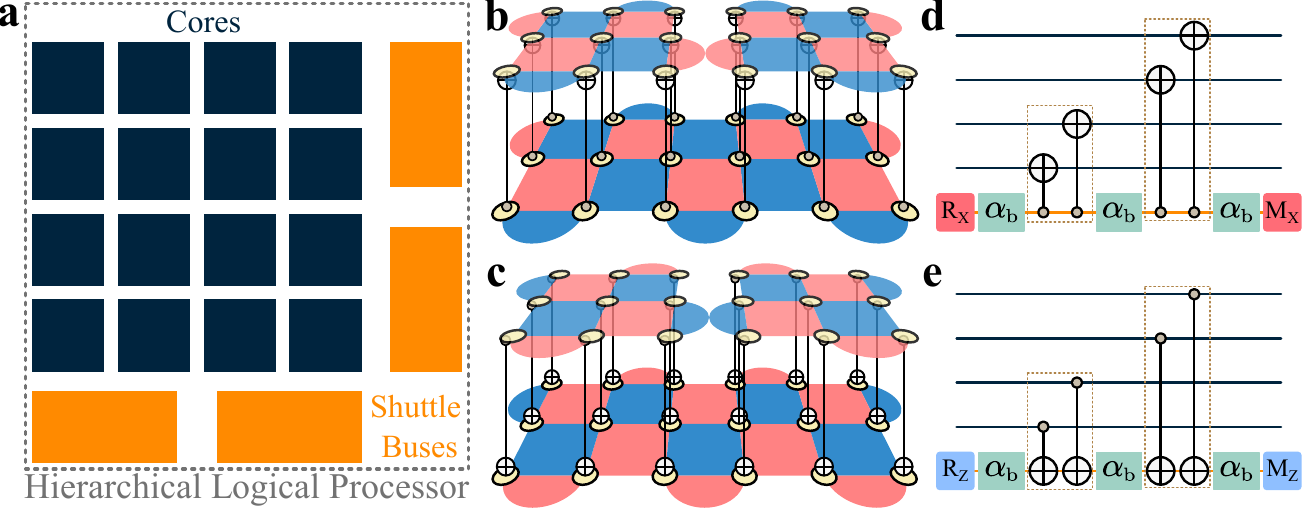}
    \caption{\textbf{Architecture of a hierarchical logical processor.} (\textbf{a}) Basic working units\textemdash{}cores and shuttle buses\textemdash{}for a hierarchical logical processor. (\textbf{b}) A bus-core CNOT gate between an $X$ bus and two cores with $d_{0}=3$ and $d_{1}=2$. (\textbf{c}) A core-bus CNOT gate between two cores and a $Z$ bus. (\textbf{d}\textendash{}\textbf{e}) Level-1 circuits representing an $X$-basis and a $Z$-basis readout gadget, respectively. Logical qubits on cores and on shuttle buses are shown as dark and orange lines, respectively. Level-1 CNOT gates in every dashed box are implemented by a bus-core or core-bus CNOT gate. Each teal block labeled by $\alpha_{\mathsf{b}}$ represents at least $\alpha_{\mathsf{b}}d_{0}$ level-0 SE rounds.}
    \label{fig: processor architecture}
\end{figure}

\textit{Architectural design}\textemdash{}Consider an HLP implementing a distance-$d_{1}$ level-1 CSS code $\mathcal{C}_1$ concatenated with the distance-$d_{0}$ RSC (the level-0 code). The HLP is supported on two types of RSC-based working units: cores and shuttle buses (Fig.~\ref{fig: processor architecture}(\textbf{a})). Each core is a distance-$d_{0}$ RSC patch, serving as a level-1 data qubit for $\mathcal{C}_1$; each shuttle bus is an elongated RSC patch of width $d_{0}$ and length $d_{0}d_{1}$, serving as a level-1 ancilla qubit. A shuttle bus whose logical $X$ (or $Z$) operator runs along the long edge is called an $X$ (or $Z$) bus. HLPs rely on a class of hybrid-unit CNOT gates that transversally couple one shuttle bus to up to $d_{1}$ cores. More specifically, an $X$ (or $Z$) bus is the control (or target) of a bus-core (or core-bus) CNOT gate (Fig.~\ref{fig: processor architecture}(\textbf{b}\textendash{}\textbf{c})). Hybrid-unit CNOT gates provide a common primitive for measuring high-weight level-1 Pauli operators. An $X$-basis (or $Z$-basis) readout gadget measures a level-1 Pauli $X$ (or $Z$) operator supported on cores with an $X$ (or $Z$) bus and bus-core (or core-bus) CNOT gates (Fig.~\ref{fig: processor architecture}(\textbf{d}\textendash{}\textbf{e})). Level-1 SE is performed by applying these gadgets to measure level-1 stabilizers. We will later use the same gadgets to perform logical Pauli measurements.

Because each shuttle bus can interact with multiple cores, and each core can in turn interact with multiple shuttle buses, errors can propagate through hybrid-unit CNOT gates and induce correlated level-1 errors. We suppress these correlations by inserting level-0 SE rounds between hybrid-unit CNOT gates. On each shuttle bus, consecutive hybrid-unit CNOT gates have at least $\alpha_{\mathsf{b}}d_{0}$ intervening level-0 SE rounds; the same number of rounds is required between shuttle-bus initialization and the first hybrid-unit CNOT gate, and between the last hybrid-unit CNOT gate and shuttle-bus measurement. For simplicity, we assume perfect time boundaries for HLPs. In this way, cores also have perfect time boundaries, whereas shuttle buses are initialized and measured transversally. 

We define the separation between two hybrid-unit CNOT gates as the minimum number of intervening level-0 SE rounds on any working unit acted on by both gates. The separation between two readout gadgets is then the minimum separation over all pairs of hybrid-unit CNOT gates, one from each gadget. During level-1 SE, readout gadgets of the same basis are required to have a separation of at least $\alpha_{\mathsf{c}}d_{0}$. We refer to $\alpha_{\mathsf{b}}$ and $\alpha_{\mathsf{c}}$ as compilation parameters. The former suppresses error propagation from shuttle buses; the latter suppresses error propagation from cores. We establish an error-reduction theorem showing that, despite error propagation induced by hybrid-unit CNOT gates, appropriate compilation parameters yield a hookless level-1 error model with effectively local stochastic errors. In this reduced error model, each level-1 error acts on a single working unit, and errors on shuttle buses do not propagate to cores but only affect the corresponding measurement outcomes.

\textit{Theorem~1 (Approximate error reduction).} Consider a circuit composed of level-0 SE rounds on cores and shuttle buses and $X$- and $Z$-basis readout gadgets. Set compilation parameters to $\alpha_{\mathsf{b}}=d_1$ and $\alpha_{\mathsf{c}}=1$, and assume perfect time boundaries for cores. Then, under the phenomenological depolarizing noise model with a sufficiently small error rate $p$, except for a rare event of probability $\sim p^{d_{0}d_{1}/2}$, the induced level-1 $X$ (or $Z$) errors are local stochastic with an effective level-1 error rate $p_1\sim p^{d_{0}/2}$; that is, for every level-1 $X$ (or $Z$) error configuration $\mathsf{f}$, the probability of inducing $\mathsf{f}$ is upper bounded by $\sim p_1^{|\mathsf{f}|}$. The formal statement and proof are provided in supplementary materials section~\ref{sec: approximate level 1 error reduction}. 

Notably, the theorem does not require the level-1 code to be LDPC. Even high-weight level-1 stabilizer measurements yield local stochastic level-1 errors under the stated separation conditions. With the stronger requirement that every pair of readout gadgets, regardless of basis, is separated by at least $\alpha_{\mathsf{c}}d_{0}$ level-0 SE rounds, long-range CNOT gates are required only once every $\Theta(d_{0})$ level-0 SE rounds for each working unit. Larger compilation parameters further relax the demand on long-range connectivity at the cost of longer level-1 SE cycle times. The theorem gives sufficient fault-tolerance conditions, but smaller compilation parameters may suffice in practice. Below, circuit-level simulations test this possibility and benchmark concrete constructions of HLP memories.

\begin{figure*}
    \centering
    \includegraphics[width=0.9\textwidth]{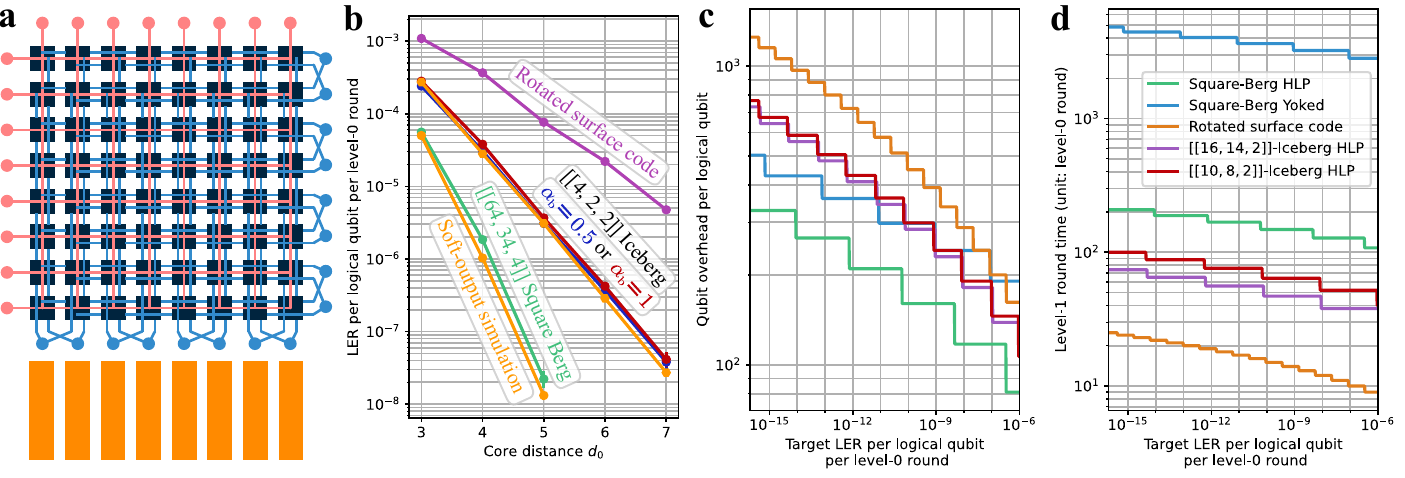}
    \caption{\textbf{Memory performance of hierarchical logical processors (HLPs).} (\textbf{a}) Level-1 layout of an HLP based on the $[[64,34,4]]$ Square Berg code~\cite{gidney_yoked_2025}. Cores and shuttle buses are represented by dark squares and orange rectangles, respectively. Each level-1 $X$ (or $Z$) stabilizer is shown as a red (or blue) circle and is connected to all level-1 data qubits on its support by red (or blue) lines. (\textbf{b}) Logical error rates from circuit-level and soft-output simulations of HLPs based on the $[[4,2,2]]$ Iceberg and $[[64,34,4]]$ Square Berg codes. We set $\alpha_{\mathsf{c}}=1$. For circuit-level simulations, $\alpha_{\mathsf{b}}=0.5$ or $1$ for the Iceberg-code HLP and $\alpha_{\mathsf{b}}=1$ for the Square Berg-code HLP; all soft-output simulations use $\alpha_{\mathsf{b}}=1$. (\textbf{c}\textendash{}\textbf{d}) Extrapolated physical-qubit overhead per logical qubit and level-1 SE cycle time, respectively, for RSCs, HLPs, and the yoked surface code~\cite{gidney_yoked_2025}. The HLPs based on the $[[10,8,2]]$ and $[[16,14,2]]$ Iceberg codes use one and two concurrent shuttle buses, respectively. The Square Berg code used here is the $[[256,194,4]]$ code. For the cycle-time comparison in (\textbf{d}), the RSC baseline is set to $d$ level-0 SE rounds for a distance-$d$ RSC. Simulations assume a physical error rate of $10^{-3}$ and perfect time boundaries for HLPs.}
    \label{fig: hlp performance}
\end{figure*}

\textit{Memory performance}\textemdash{}We benchmark HLP memory performance using circuit-level simulations of HLPs based on the $[[4,2,2]]$ Iceberg code and the $[[64,34,4]]$ Square Berg code~\cite{gidney_yoked_2025} (Fig.~\ref{fig: hlp performance}(\textbf{a}\textendash{}\textbf{b})). We fix $\alpha_{\mathsf{c}}=1$ and use $\alpha_{\mathsf{b}}=0.5$ or $1$ for the Iceberg-code HLP and $\alpha_{\mathsf{b}}=1$ for the Square Berg-code HLP. In all cases, the logical error rate decreases with increasing core distance, exhibiting scaling close to the full concatenated code distance $d_{0}d_{1}$ (Fig.~\ref{fig: hlp performance}(\textbf{b})). These results suggest that compilation parameters substantially smaller than the sufficient value $\alpha_{\mathsf{b}}=d_{1}$ in Theorem~1 can effectively suppress errors in HLPs, thereby shortening the level-1 SE cycle.

To access larger HLPs, we develop a soft-output simulator following~\cite{gidney_yoked_2025,meister_efficient_2024} that directly samples and decodes effective level-1 errors. Its results closely agree with circuit-level simulations for the HLPs above (Fig.~\ref{fig: hlp performance}(\textbf{b})), allowing us to extrapolate the memory performance of larger HLP constructions. At a physical error rate of $10^{-3}$, an HLP based on the $[[256,194,4]]$ Square Berg code achieves $3$\textendash{}$4$ times higher qubit efficiency than the RSC over target logical error rates from $10^{-10}$ to $10^{-15}$ (Fig.~\ref{fig: hlp performance}(\textbf{c})). Compared with the yoked surface code~\cite{gidney_yoked_2025} based on the same level-1 code, the HLP reduces the space overhead by $100$\textendash{}$200$ physical qubits per logical qubit and shortens the level-1 SE cycle time by a factor of $20$\textendash{}$30$ (Fig.~\ref{fig: hlp performance}(\textbf{c}\textendash{}\textbf{d})). Simulation and extrapolation methods are detailed in supplementary materials section~\ref{sec: details on numerical simulation}. The level-1 SE cycle time of the $[[256,194,4]]$ HLP remains approximately ten times that of the RSC baseline (Fig.~\ref{fig: hlp performance}(\textbf{d})). This longer error-correction timescale motivates exploiting the parallel logical operations enabled by the HLP, which we consider next.

\begin{figure*}
    \centering
    \includegraphics[width=0.9\textwidth]{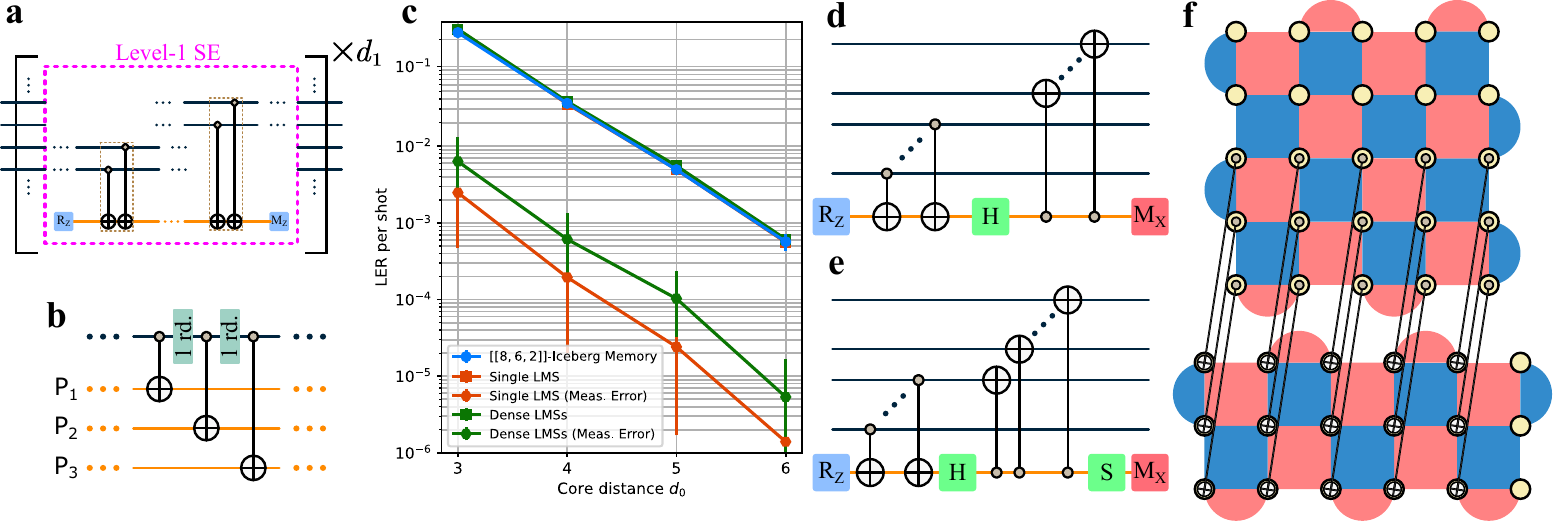}
    \caption{\textbf{Parallel logical Pauli measurements on a hierarchical logical processor (HLP).} (\textbf{a}) A $Z$-basis logical measurement sequence (LMS) comprising $d_{1}$ logical readout gadgets, one in each of $d_{1}$ consecutive level-1 syndrome extraction (SE) rounds. (\textbf{b}) Dense packing of logical readout gadgets from three LMSs measuring logical Pauli operators $\mathsf{P}_{1}$, $\mathsf{P}_{2}$, and $\mathsf{P}_{3}$, respectively. Gadgets from different LMSs can be separated by a single level-0 SE round; each teal strip denotes one such round. (\textbf{c}) Circuit-level benchmarking of logical $Z$ measurements for an HLP based on the $[[8,6,2]]$ Iceberg code. Squares show total logical error rates, including memory and logical measurement errors; hexagons show logical measurement error rates alone. The memory baseline contains five level-1 SE rounds with perfect time boundaries. (\textbf{d}) $H$-transformed readout gadget for a logical Pauli operator $\mathsf{P}_{xz}$ with $x\cdot z=0$. (\textbf{e}) $HS$-transformed readout gadget for a logical Pauli operator $\mathsf{P}_{xz}$ with $x\cdot z=1$. (\textbf{f}) Extended core-bus CNOT gate between an external rotated surface code patch and a $Z$ bus.} 
    \label{fig: logical pauli measurement}
\end{figure*}

\textit{Parallel logical Pauli measurements}\textemdash{}We use sequences of readout gadgets to fault-tolerantly measure logical Pauli $X$ or $Z$ operators. An $X$-basis ($Z$-basis) logical measurement sequence (LMS) measures a logical $X$-type ($Z$-type) Pauli operator using a sequence of $X$-basis ($Z$-basis) readout gadgets, which we refer to as logical readout gadgets. An LMS contains $d_{1}$ logical readout gadgets distributed across $d_{1}$ consecutive level-1 SE rounds, with one gadget in each round (Fig.~\ref{fig: logical pauli measurement}(\textbf{a})). Each pair of adjacent logical readout gadgets in an LMS provides an additional level-1 detector for detecting readout errors on these gadgets. Each logical readout gadget is separated by at least $\alpha_{\mathsf{c}}d_{0}$ level-0 SE rounds from the other logical readout gadgets in the same LMS and from same-basis readout gadgets used for level-1 SE.

Different LMSs can be densely packed (Fig.~\ref{fig: logical pauli measurement}(\textbf{b})). For multiple logical measurements, each LMS can be decoded independently, so correlations between different LMSs need not be suppressed by additional temporal separation. Their logical readout gadgets can therefore be separated by only one level-0 SE round. This enables multiple mutually commuting logical Pauli measurements to proceed in parallel by trading additional shuttle buses and long-range CNOT gates for increased logical-operation throughput. We prove that an LMS achieves full level-1 distance and remains fault tolerant even when different LMSs are densely packed (supplementary materials section~\ref{sec: logical pauli measurements}).

We benchmark this parallelism using an HLP based on the $[[8,6,2]]$ Iceberg code (Fig.~\ref{fig: logical pauli measurement}(\textbf{c})). We compare a circuit containing one LMS with one containing three densely packed LMSs over the same five level-1 SE rounds. The single-LMS circuit measures a high-weight logical $Z$ operator supported on six of the eight cores. The presence of the LMSs does not appreciably increase the total logical error rates, and logical measurement error rates remain well below the memory error rate. These results show that long logical Pauli operators can be measured reliably and that multiple logical measurements can be performed in parallel without degrading memory performance.

The readout-gadget framework can also be extended beyond logical $X$ and $Z$ measurements. Consider a general logical Pauli operator $\mathsf{P}_{xz}=i^{x\cdot z}\prod_{i}\mathsf{X}_{i}^{x_i}\mathsf{Z}_i^{z_i}$, where the binary vectors $x$ and $z$ specify its logical $X$ and $Z$ supports; $\mathsf{X}_i$ and $\mathsf{Z}_i$ denote the logical $X$ and $Z$ operators on the $i$th core, respectively. For $x\cdot z=0$, we construct an $H$-transformed readout gadget in which a shuttle bus first couples to the $Z$ support of $\mathsf{P}_{xz}$, undergoes a transversal Hadamard transformation, and then couples to its $X$ support (Fig.~\ref{fig: logical pauli measurement}(\textbf{d})). We prove in supplementary materials section~\ref{sec: logical pauli measurements} that LMSs constructed from these gadgets also achieve full level-1 distance and can be densely packed for parallel measurements. When $x\cdot z=1$, the same construction measures $\mathsf{P}_{xz}$ with the aid of a catalytic logical $Y$ state~\cite{chamberland_universal_2022}. We further propose an $HS$-transformed gadget that directly measures this class of operators without requiring the additional logical $Y$ state (Fig.~\ref{fig: logical pauli measurement}(\textbf{e})). Finally, extended hybrid-unit CNOT gates couple shuttle buses to external RSC patches, enabling joint logical Pauli measurements between logical qubits encoded in an HLP and those encoded in external RSC patches (Fig.~\ref{fig: logical pauli measurement}(\textbf{f})).

\textit{Outlook}\textemdash{}To summarize, we design the HLP, a concatenation-based FTQC architecture that combines high qubit efficiency with infrequent use of long-range connectivity and supports parallel logical operations. Both level-1 syndrome extraction and logical operations are supported by a common class of readout gadgets. We expect that the memory overhead can be further reduced by using high-rate level-1 codes with larger distances or alternative level-0 codes, a direction also explored in recent concurrent works~\cite{wills_concatenating_2026,low_denser_2026}. Because a hierarchical structure necessarily leads to a longer logical operation timescale, leveraging the parallelizability of logical measurements and the capability of measuring long logical operators will be important for achieving higher logical processing throughput. Alternative ways to perform logical operations on an HLP\textemdash{}such as transversal logical gates at level 1 assisted by level-1 flag qubits~\cite{chao_quantum_2018,chao_fault_tolerant_2018,he_quantum_2025,tansuwannont_construction_2026}\textemdash{}are also worth investigation. External cores, on which logical magic states can be prepared~\cite{litinski_game_2019,gidney_magic_2024,chen_efficient_magic_2026,vaknin_high_2026,sahay_fold-transversal_2026,claes_cultivating_2025}, can be combined with HLPs to implement universal logical operations. Co-designing quantum algorithms with the HLP\textendash{}external-core hybrid architecture for efficient logical computation is an interesting direction for future work. Finally, it will be important to develop detailed hardware-level implementations of the HLP. For neutral atom platforms in particular, the local-CZ-based architecture~\cite{sunami_transversal_2025,radnaev_universal_2025,rines_demonstration_2026} allows in-place implementation of entangling gates for level-0 SE without atom shuttling. It would be interesting to explore whether this architecture offers a significant speedup for our HLP compared to the zone-based architecture~\cite{bluvstein_logical_2024,bluvstein_fault_tolerant_2026}.

\textit{Note}\textemdash{}See the public repository~\cite{chen_hlp} for circuits and source codes used in our numerical experiments. \\

\bibliography{bibliography}

\clearpage  
\onecolumngrid
\begingroup
\renewcommand{\thefootnote}{*}
\begin{center} 
    {\large\bfseries Supplemental Material for\\[0.5em] ``Hierarchical Quantum Logical Processor with Amortized Long-Range Connectivity''\par} 
    \vspace{1em} 
    {\normalsize Zi-Han Chen,$^{1,2,3,}$\footnotemark $\ $Ming-Cheng Chen,$^{1,2,3}$ Chao-Yang Lu,$^{1,2,3}$ and Jian-Wei Pan$^{1,2,3}$\par} 
    \vspace{0.5em} 
    {\small $^{1}$Hefei National Research Center for Physical Sciences at the Microscale and Department of Modern Physics, New Cornerstone Science Laboratory, University of Science and Technology of China, Hefei 230026, China\\ 
    $^{2}$Shanghai Research Center for Quantum Science and CAS Center for Excellence in Quantum Information and Quantum Physics, University of Science and Technology of China, Shanghai 201315, China\\ 
    $^{3}$Hefei National Laboratory, University of Science and Technology of China, Hefei 230088, China\par} 
    \vspace{0.5em} 
\end{center} 
\vspace{1.5em}

\twocolumngrid
\footnotetext{\href{mailto:czh007@mail.ustc.edu.cn}{czh007@mail.ustc.edu.cn}}
\endgroup

\setcounter{figure}{0}
\setcounter{theorem}{0}
\renewcommand{\thefigure}{S\arabic{figure}}
\renewcommand{\theHfigure}{supplement.\arabic{figure}}
\renewcommand{\thetheorem}{S\arabic{theorem}}
\renewcommand{\theHtheorem}{supplement.\arabic{theorem}}
\setcounter{page}{1}
\renewcommand{\thepage}{S\arabic{page}}
\setcounter{algocf}{0}
\renewcommand{\thealgocf}{S\arabic{algocf}}
\renewcommand{\theHalgocf}{supplement.\arabic{algocf}}
\setcounter{equation}{0}
\renewcommand{\theequation}{S\arabic{equation}}
\renewcommand{\theHequation}{supplement.\arabic{equation}}

We present a detailed description of the hierarchical logical processor (HLP) architecture in the supplementary material. In Sec.~\ref{sec: hlp basic modules}, we describe basic modules for an HLP and its decoding procedure.  In Sec.~\ref{sec: approximate level 1 error reduction}, we analyze the induced level-1 error model from level-0 errors. We prove that the correlation between level-1 $X$ (or $Z$) errors can be effectively suppressed. In Sec.~\ref{sec: logical pauli measurements}, we prove that our logical $X$ (or $Z$) measurement protocol (with an LMS) achieves full level-1 distance. We then upper bound the probability of logical measurement errors for LMSs using the induced level-1 error model analyzed in Sec.~\ref{sec: approximate level 1 error reduction}. We describe the construction of two additional readout gadgets, the $H$-transformed readout gadget and the $HS$-transformed readout gadget. For LMSs based on $H$-transformed readout gadgets, we prove similar results to $X$-basis (or $Z$-basis) LMSs. We describe the extension of hybrid-unit CNOT gates to enable interfacing an HLP with external cores. In Sec.~\ref{sec: details on numerical simulation}, we describe details on numerical simulations, including circuit-level simulations, soft-output simulations, and logical error rate extrapolation. 

\section{Hierarchical Logical Processor: Basic Modules}\label{sec: hlp basic modules}

Our HLP essentially implements fault-tolerant logical operations (preserving logical information and measuring logical Pauli operators) on a concatenated code\textemdash{}logical information is encoded in a level-1 high-rate CSS code $\mathcal{C}_1$ with a distance of $d_1$, which is run on top of the level-0 distance-$d_0$ RSC. In this section, we describe key components for our HLP and how to perform level-1 SE and logical Pauli $X$ or $Z$ measurements. We describe extensions to these modules, including performing general logical Pauli measurements and interfacing with logical information directly encoded in external RSC patches (external cores), in later sections. 

\subsection{Cores, shuttle buses, and hybrid-unit CNOT gates}
An HLP consists of two types of working units: cores and shuttle buses, both of which are code patches whose logical qubits serve as level-1 qubits. Every level-1 data qubit of $\mathcal{C}_1$ is the logical qubit of a distance-$d_{0}$ RSC patch, referred to as a core. Every level-1 ancilla qubit for level-1 SE and logical Pauli measurements in our HLP is the logical qubit of an elongated RSC patch, referred to as a shuttle bus. We require each shuttle bus to have a width of $d_0$ and a length of $d_0d_1$. We note that our specific choice here is for a clean presentation of results, and that it is possible to use shuttle buses of different widths and lengths. A shuttle bus whose logical $X$ (or $Z$) operator is supported on the long edge is called an $X$ (or $Z$) bus. See Fig.~\ref{fig: cores and buses} for an illustration of cores, $X$ buses, and $Z$ buses. We can perform (with a depth-1 layer of physical CNOT gates) transversal logical CNOT gates with an $X$ (or $Z$) bus on the control side (or target side) and up to $d_1$ cores on the target side (or control side). We refer to a depth-1 layer of physical CNOT gates implementing logical CNOT gates between (up to $d_1$) cores and a shuttle bus as a hybrid-unit CNOT gate. More specifically, a hybrid-unit CNOT gate between an $X$ (or $Z$) bus and up to $d_1$ cores is called a bus-core (or core-bus) CNOT gate. 

\begin{figure}
    \centering
    \includegraphics[width=\columnwidth]{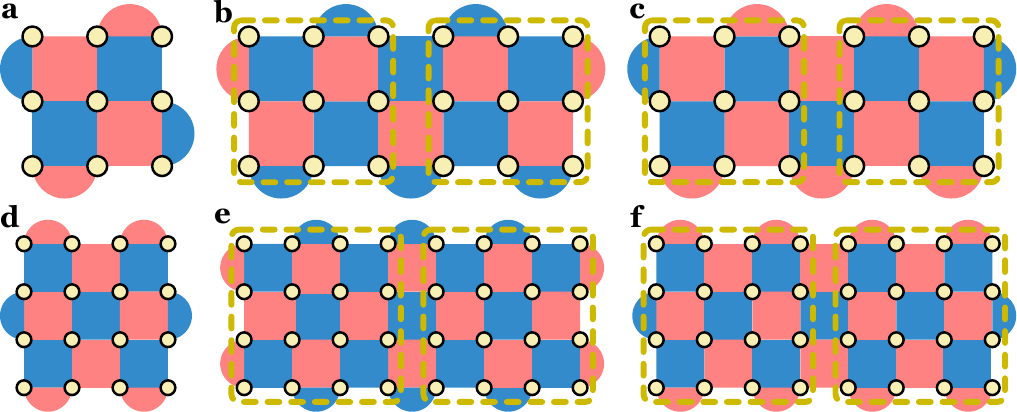}
    \caption{Cores and shuttle buses. (\textbf{a}\textendash{}\textbf{c}) A core, an $X$ bus, and a $Z$ bus, respectively, with $d_0=3$ and $d_1=2$. (\textbf{d}\textendash{}\textbf{f}) A core, an $X$ bus, and a $Z$ bus, respectively, with $d_{0}=4$ and $d_1=2$. Every dashed yellow region supports a bus-core (or core-bus) CNOT gate between an $X$ (or $Z$) bus and a core.}
    \label{fig: cores and buses}
\end{figure}

\subsection{Circuit structure and error model}\label{subsec: circuit structure and error model}
\textit{Level-0 circuit structure}\textemdash{}For conceptual simplicity and hardware friendliness, we implement level-0 SE rounds on cores and shuttle buses synchronously and require that there is at most a depth-1 layer of hybrid-unit CNOT gates between every two consecutive level-0 SE rounds. We refer to each level-0 SE round with its preceding depth-1 layer of hybrid-unit CNOT gates (if any) as a level-0 time step. We now define the (time-like) separation between two hybrid-unit CNOT gates.
\begin{definition}
    [Separation between two hybrid-unit CNOT gates]\label{def: separation between two cnot gates} If the supports of two hybrid-unit CNOT gates are not disjoint, their separation is defined as the absolute difference between their corresponding level-0 time steps. Otherwise, their separation is defined to be infinity.    
\end{definition}
We merge transversal initialization (or measurement) on a shuttle bus into the following (or preceding) level-0 SE round on the working unit. A shuttle bus is activated by the transversal initialization and deactivated by the transversal measurement. A deactivated shuttle bus will not be activated again. Instead, we reallocate all qubits of a deactivated shuttle bus to activate new shuttle buses later. We define the lifetime of a shuttle bus as the period between the level-0 time step at which the shuttle bus is initialized and the level-0 time step at which the shuttle bus is measured.  

\textit{Level-1 circuit structure}\textemdash{}Every level-0 circuit on cores and shuttle buses induces a corresponding level-1 circuit that represents operations on logical qubits of these working units. More specifically, level-0 SE rounds correspond to idling operations at level 1; transversal initialization and measurement on shuttle buses correspond to reset and measurement operations on their corresponding logical qubits at level 1; hybrid-unit CNOT gates correspond to level-1 CNOT gates. In every layer of a level-1 circuit, each level-1 qubit (corresponding to the logical qubit of a working unit) participates in one type of level-1 operation. Since each hybrid-unit CNOT gate can simultaneously couple a level-1 ancilla with up to $d_1$ level-1 data qubits, we require that in every level-1 layer, (i) every level-1 data qubit participates in at most one CNOT gate with a level-1 ancilla and (ii) every level-1 ancilla qubit on an $X$ (or $Z$) bus participates in at most $d_1$ CNOT gates on the control side (or target side). 

\textit{Perfect time boundaries}\textemdash{}Throughout this work, we consider HLPs with perfect initialization and measurement. More specifically, an HLP is initialized by (i) measuring all level-0 and level-1 stabilizers perfectly and (ii) measuring a collection of logical Pauli operators perfectly according to the logical initialization basis. Similarly, the final measurement on an HLP applies the same perfect measurement procedure used during perfect initialization. 

\textit{Uniform depolarizing circuit noise model}\textemdash{}This model is a standard circuit-level noise model used for benchmarking fault-tolerant protocols. We perform all numerical simulations on level-0 circuits using this circuit noise model, with error locations listed as follows.
\begin{itemize}
    \item Each single-qubit initialization in $X$ (or $Z$) basis is followed by a single-qubit $Z$ (or $X$) error with probability $p$.
    \item  Each single-qubit measurement in $X$ (or $Z$) basis is preceded by a single-qubit $Z$ (or $X$) error with probability $p$.
    \item Each single-qubit gate (including idling) is followed by a single-qubit depolarizing channel with noise strength $p$, consisting of three possible Pauli errors $\{X,Y,Z\}$, each with probability $p/3$.  
    \item Each two-qubit gate is followed by a two-qubit depolarizing channel with noise strength $p$, consisting of 15 possible Pauli errors $\{I,X,Y,Z\}^{\otimes2}\backslash\{II\}$, each with probability $p/15$.
\end{itemize}

\textit{Phenomenological depolarizing noise model}\textemdash{}In this noise model, errors may only occur on level-0 data qubits at the beginning of each level-0 time step and on level-0 measurement results. We regard an $X$-basis (or $Z$-basis) measurement error as a $Z$ (or $X$) error. The error locations are listed as follows.
\begin{itemize}
    \item For each working unit $\sigma$ at the beginning of every level-0 time step, if $\sigma$ is transversally initialized in the $X$ (or $Z$) basis, then every data qubit of $\sigma$ experiences a Pauli $Z$ (or $X$) error with probability $p$ following the transversal initialization. Otherwise, every data qubit experiences a single-qubit depolarizing channel with noise strength $p$.
    \item Every measurement result is flipped with probability $p$. 
\end{itemize}

\subsection{Level-1 readout gadget}\label{subsec: level 1 readout gadget}
An $X$-basis (or $Z$-basis) readout gadget performs a projective measurement on a level-1 Pauli $X$ (or $Z$) operator on level-1 data qubits with a level-1 ancilla implemented by an $X$ bus (or $Z$ bus). More specifically, for a level-1 Pauli $X$ operator $\mathsf{P}_{x}:=\mathsf{X}_1^{x_1}\otimes\cdots\otimes\mathsf{X}_{n_1}^{x_{n_1}}$ where $x:=(x_1,\cdots,x_{n_1})$ is a vector in $\mathbb{Z}_2^{n_1}$, the level-1 circuit of an $X$-basis readout gadget to measure $\mathsf{P}_{x}$ proceeds as follows:
\begin{enumerate}
    \item Initialize the level-1 ancilla in the $X$ basis.
    \item Perform a sequence of CNOT gates between the level-1 ancilla (control) and level-1 data qubits (target), such that for every level-1 data qubit $i$ with $x_i=1$, there is exactly one CNOT gate between this level-1 data qubit and the level-1 ancilla. Every level-1 layer of CNOT gates has at most $d_1$ level-1 data qubits on the target side. 
    \item  Measure the level-1 ancilla in the $X$ basis. 
\end{enumerate}
Similarly, the level-1 circuit of a $Z$-basis readout gadget to measure a level-1 Pauli $Z$ operator $\mathsf{P}_z:=\mathsf{Z}_1^{z_1}\otimes\cdots\otimes\mathsf{Z}_{n_1}^{z_{n_1}}$ with $z:=(z_1,\cdots,z_{n_1})\in\mathbb{Z}_2^{n_1}$ proceeds as follows:
\begin{enumerate}
    \item Initialize the level-1 ancilla in the $Z$ basis.
    \item Perform a sequence of CNOT gates between level-1 data qubits (control) and the level-1 ancilla qubit (target), such that for every level-1 data qubit $i$ with $z_i=1$, there is exactly one CNOT gate between this level-1 data qubit and the level-1 ancilla. Every level-1 layer of CNOT gates has at most $d_1$ level-1 data qubits on the control side. 
    \item  Measure the level-1 ancilla in the $Z$ basis. 
\end{enumerate}

We now describe the level-0 implementation of an $X$-basis (or $Z$-basis) readout gadget. The level-1 ancilla of the readout gadget is the logical qubit of the shuttle bus associated with the gadget. Initialization and measurement on the level-1 ancilla in the $X$ or $Z$ basis are implemented by transversal initialization and measurement on the corresponding shuttle bus, respectively. Every level-1 layer of CNOT gates in the readout gadget is implemented by a hybrid-unit CNOT gate. We insert padding level-0 SE rounds on the shuttle bus so that (i) every two adjacent hybrid-unit CNOT gates in the readout gadget are separated by at least $\alpha_{\mathsf{b}}d_{0}$ level-0 SE rounds and (ii) transversal reset (or measurement) of the shuttle bus is separated from the following (or preceding) hybrid-unit CNOT gate by at least $\alpha_{\mathsf{b}}d_{0}$ rounds. Note that level-0 $X$ (or $Z$) errors on $X$ (or $Z$) buses may propagate to cores via hybrid-unit CNOT gates, thereby inducing correlated errors at level 1. The parameter $\alpha_{\mathsf{b}}$ should therefore be chosen to suppress such error correlations. Moreover, the level-1 measurement result, referred to as the readout value, of a readout gadget may not be reliable, since a logical $X$ (or $Z$) error on the $Z$ (or $X$) bus of a $Z$-basis (or $X$-basis) readout gadget would flip the readout value, causing a readout error.

Readout gadgets (each with its own shuttle bus) can be juxtaposed, interleaved, or sequentially performed to implement an HLP. As mentioned in Sec.~\ref{subsec: circuit structure and error model}, hybrid-unit CNOT gates in readout gadgets should not overlap at any level-0 time step to guarantee that the level-0 implementation has at most a depth-1 layer of physical CNOT gates at the beginning of each level-0 time step. We define the separation between two readout gadgets as follows. 
\begin{definition}
    [Separation between two readout gadgets]\label{def: separation between readout gadgets} Given two readout gadgets, the separation between them is defined as the minimum separation between a hybrid-unit CNOT gate in one gadget and a hybrid-unit CNOT gate in the other. 
\end{definition}

\subsection{Basic Modules I: logical memory} 
The most fundamental functionality of an HLP is to protect quantum information by performing SE on the level-1 code $\mathcal{C}_1$. We perform each level-1 SE round by measuring the same set of level-1 stabilizers, with each stabilizer measured using a readout gadget. Note that once a readout gadget is completed, the qubits used for its shuttle bus are now available again for a new readout gadget. Thus, we can trade off between the qubit overhead for readout gadgets and time overhead for a level-1 SE round by limiting the number of level-1 stabilizers measured in parallel. We require that every two $X$-basis (or $Z$-basis) readout gadgets for level-1 SE should be separated by at least $\alpha_{\mathsf{c}}d_0$ level-0 time steps, where $\alpha_{\mathsf{c}}$ is a positive parameter. As discussed in the main text, level-0 $X$ (or $Z$) errors on cores may induce correlated $X$ (or $Z$) errors on $Z$ (or $X$) buses via hybrid-unit CNOT gates, thereby inducing correlated readout errors on $Z$-basis (or $X$-basis) readout gadgets. The parameter $\alpha_{\mathsf{c}}$ is chosen to suppress correlations among readout errors. 

We refer to the parameters $\alpha_{\mathsf{b}}$ and $\alpha_{\mathsf{c}}$ as compilation parameters. We later prove (under the phenomenological depolarizing noise model) that we can effectively suppress correlations among level-1 errors by setting $\alpha_{\mathsf{c}}=1$ and $\alpha_{\mathsf{b}}=d_1$. Larger compilation parameters imply sparser hybrid-unit CNOT gates and lower requirements on long-range connectivity, while also resulting in longer level-1 SE round time. Choosing these parameters\textemdash{}balancing level-1 round time, level-1 error correlation, and demand on long-range connectivity\textemdash{}is subtle and requires extensive numerical study. 

\subsection{Basic Modules II: logical $X$ or $Z$ measurements}
We use a sequence of $X$-basis (or $Z$-basis) readout gadgets, called a logical measurement sequence (LMS), to reliably measure a single logical $X$ (or $Z$) operator. We refer to readout gadgets in an LMS as logical readout gadgets. More specifically, to measure an $X$-basis (or $Z$-basis) logical operator $\mathsf{P}$, an LMS consists of $d_1$ $X$-basis (or $Z$-basis) logical readout gadgets $\mathsf{R}_1,\cdots,\mathsf{R}_{d_1}$, sequentially placed in $d_1$ contiguous level-1 SE rounds from $\mathsf{t}$ to $\mathsf{t}+d_{1}-1$, such that each logical readout gadget $\mathsf{R}_{j}$ ($j\in\{1,\cdots,d_1\}$) measures $\mathsf{P}$ and is embedded in the level-1 SE round $\mathsf{t}+j-1$. Moreover, we require that every $X$-basis logical readout gadget is separated from $X$-basis readout gadgets for level-1 SE and other logical readout gadgets in the same LMS by at least $\alpha_{\mathsf{c}}d_{0}$ level-0 SE rounds. When measuring multiple (mutually commuting) logical operators in parallel, $X$-basis (or $Z$-basis) logical readout gadgets from different LMSs are allowed to be separated by as little as one level-0 time step. This enables parallel logical measurements at the cost of additional shuttle buses and increased use of long-range connectivity for hybrid-unit CNOT gates.

\subsection{Detectors and decoding}\label{subsec: detectors and decoding}
There are two levels of detectors for an HLP: level-0 detectors directly generated from level-0 SE rounds in the level-0 circuit and level-1 detectors induced by the level-1 circuit. Note that since an HLP is implemented physically by the level-0 circuit, level-1 detectors are also compiled into sums over level-0 measurement results. In the following, we describe how to set up detectors and perform decoding for an HLP with basic modules described in the previous two subsections. 

\textit{Level-0 detectors}\textemdash{}The construction of level-0 detectors is essentially the same as constructing detectors on circuits composed of SE and transversal CNOT gates on the RSC. Every level-0 detector is supported on measurement results in up to two adjacent level-0 time steps. Consider two adjacent level-0 time steps $t-1$ and $t$. If there are no hybrid-unit CNOT gates at time step $t$, for each level-0 X or Z stabilizer $\mathrm{S}$ measured in both time steps $t-1$ and $t$, there is a level-0 detector $\mathrm{D}(\mathrm{S},t):=m_{\mathrm{S},t-1}\oplus m_{\mathrm{S},t}$, where $\oplus$ is the XOR operator; $m_{\mathrm{S},t-1}$ and $m_{\mathrm{S},t}$ denote measurement results of $\mathrm{S}$ at time steps $t-1$ and $t$, respectively; and $t$ also denotes the time coordinate of the detector. On the other hand, suppose there is a layer of hybrid-unit CNOT gates $\Lambda_t$ at time step $t$. Through the gate $\Lambda_{t}$, a stabilizer $\mathrm{S}$ on a working unit propagates to $\Lambda_t\mathrm{S}\Lambda_t^{\dagger}$, which we denote as $\Lambda_t(\mathrm{S})$. Since $\Lambda_t(\mathrm{S})$ is a product of stabilizer generators, we will also use it to denote the set of these stabilizer generators. Then, in this case, the level-0 detector corresponding to $\mathrm{S}$ generated at time step $t$ is $\mathrm{D}(\mathrm{S},t):=m_{\mathrm{S},t-1}\oplus\left(\oplus_{\mathrm{S'}\in\Lambda_t(\mathrm{S})}m_{\mathrm{S}',t}\right)$. 

\textit{Level-1 detectors}\textemdash{}Level-1 detectors are induced by readout gadgets for level-1 SE and level-1 logical measurements. For every two adjacent level-1 SE rounds $\mathsf{t}$ and $\mathsf{t}+1$, the two readout gadgets for each level-1 stabilizer $\mathsf{S}$ induce a level-1 detector. Every two adjacent logical readout gadgets in an LMS induce a level-1 detector. 

\textit{Level-0 decoding hypergraph}\textemdash{}Let $\mathcal{E}$ be the set of all level-0 single-qubit $X$ and $Z$ error locations; let $\mathcal{D}_{0}$ be the set of all level-0 detectors; let $\mathcal{D}_1$ be the set of all level-1 detectors. Configurations of errors (subsets of $\mathcal{E}$), configurations of level-0 syndromes (subsets of $\mathcal{D}_0$), and configurations of level-1 syndromes (subsets of $\mathcal{D}_1$) are naturally identified with vectors in $\mathbb{Z}_2^{|\mathcal{E}|}$, $\mathbb
{Z}_2^{|\mathcal{D}_{0}|}$, and $\mathbb{Z}_2^{|\mathcal{D}_1|}$, respectively. Define syndrome maps $\partial_{0}:\mathbb{Z}_2^{|\mathcal{E}|}\to\mathbb{Z}_2^{|\mathcal{D}_{0}|}$ and $\partial_{1}:\mathbb{Z}_2^{|\mathcal{E}|}\to\mathbb{Z}_2^{|\mathcal{D}_{1}|}$, representing level-0 and level-1 syndromes of errors, respectively. There are three types of errors in $\mathcal{E}$.
\begin{enumerate}
    \item Canonical graphlike errors. A canonical graphlike error refers to an error that (i) has at most two endpoints and (ii) all its endpoints are on the same core or shuttle bus.
    \item  Decomposable hyperedge errors. A decomposable hyperedge error refers to an error that (i) can be decomposed into canonical graphlike errors and (ii) has endpoints over exactly two units, a core and a shuttle bus.   
    \item Primitive hyperedge errors. A primitive hyperedge error refers to an error that cannot be decomposed into canonical graphlike errors. 
\end{enumerate}
A primitive hyperedge error is always associated with a hybrid-unit CNOT gate. Consider a primitive hyperedge $X$ error $e$ in level-0 time step $t$, then there is a hybrid-unit CNOT gate $\Lambda_t$ at time step $t$ such that the error $e$ occurs on a working unit $\sigma$ on the control side of $\Lambda_t$. Moreover, this error triggers three $Z$ detectors $\mathrm{D}_{1}$, $\mathrm{D}_2$, and $\mathrm{D}_{3}$, such that the first two detectors are on the unit $\sigma$ with time coordinates $t$ and $t+1$, respectively, and the third one is on a different working unit $\sigma'$ on the target side of $\Lambda_t$ with a time coordinate $t$. We say the detector $\mathrm{D}_{3}$ is on the pointy end of $e$. Similar statements hold for primitive hyperedge $Z$ errors.  

Let $\mathcal{E}_{0}$ be the collection of all canonical graphlike errors and primitive hyperedge errors in $\mathcal{E}$. We construct the level-0 decoding hypergraph $\mathcal{G}_0$ with the vertex set $\mathcal{D}_{0}$ and the hyperedge set $\mathcal{E}_{0}$. Every hyperedge $e\in\mathcal{E}_{0}$ is connected exactly to all vertices in $\partial_{0}e$. Note that our analysis of single-qubit $X$ and $Z$ errors above already encapsulates all two-qubit $X$ and $Z$ errors, since a two-qubit $X$ (or $Z$) error following a CNOT gate is equivalent to a single-qubit $X$ (or $Z$) error preceding the CNOT gate. 

\textit{Level-0 decoding graphs}\textemdash{}Let $\mathcal{D}_{\mathsf{B}_{\lozenge}}$ be the collection of all level-0 $Z$ detectors on $X$ buses and $X$ detectors on $Z$ buses; let $\mathcal{D}_{\mathsf{C}}$ be the collection of all level-0 detectors on cores; let $\mathcal{D}_{\mathsf{B}_{\blacklozenge}}$ be the collection of all level-0 $X$ detectors on $X$ buses and $Z$ detectors on $Z$ buses. Similarly, for errors in $\mathcal{E}_{0}$, let $\mathcal{E}_{\mathsf{B}_{\lozenge}}$ be the collection of all $X$ errors on $X$ buses and $Z$ errors on $Z$ buses; let $\mathcal{E}_{\mathsf{C}}$ be the collection of all $X$ and $Z$ errors on cores; let $\mathcal{E}_{\mathsf{B}_{\blacklozenge}}$ be the collection of all $Z$ errors on $X$ buses and $X$ errors on $Z$ buses. For a level-0 syndrome configuration $\eta\in\mathbb{Z}_2^{|\mathcal{D}_{0}|}$, denote its restriction to $\mathcal{D}_{\mathsf{B}_{\lozenge}}$, $\mathcal{D}_{\mathsf{C}}$, and $\mathcal{D}_{\mathsf{B}_{\blacklozenge}}$ as $\eta|_{\mathsf{B}_{\lozenge}}$, $\eta|_{\mathsf{C}}$, and $\eta|_{\mathsf{B}_{\blacklozenge}}$, respectively. Similarly, for a level-0 error configuration $\epsilon\in\mathbb{Z}_2^{|\mathcal{E}_{0}|}$, denote its restriction to $\mathcal{E}_{\mathsf{B}_{\lozenge}}$, $\mathcal{E}_{\mathsf{C}}$, and $\mathcal{E}_{\mathsf{B}_{\blacklozenge}}$ as $\epsilon|_{\mathsf{B}_{\lozenge}}$, $\epsilon|_{\mathsf{C}}$, and $\epsilon|_{\mathsf{B}_{\blacklozenge}}$, respectively. 

From a level-0 decoding hypergraph $\mathcal{G}_{0}$ for a level-0 circuit, we now construct three level-0 decoding graphs $\mathcal{G}_{\mathsf{B}_{\lozenge}}$, $\mathcal{G}_{\mathsf{C}}$, and $\mathcal{G}_{\mathsf{B}_{\blacklozenge}}$ with vertex sets $\mathcal{D}_{\mathsf{B}_\lozenge}$, $\mathcal{D}_{\mathsf{C}}$, and $\mathcal{D}_{\mathsf{B}_{\blacklozenge}}$, respectively. Edges for the decoding graphs $\mathcal{G}_{\mathsf{B}_{\lozenge}}$, $\mathcal{G}_{\mathsf{C}}$, and $\mathcal{G}_{\mathsf{B}_{\blacklozenge}}$ have one-to-one correspondence to level-0 errors in $\mathcal{E}_{\mathsf{B}_{\lozenge}}$, $\mathcal{E}_{\mathsf{C}}$, and $\mathcal{E}_{\mathsf{B}_{\blacklozenge}}$, respectively, such that every error $e$ in $\mathcal{E}_{\mathsf{B}_{\lozenge}}$ (or $\mathcal{E}_{\mathsf{C}}$ or $\mathcal{E}_{\mathsf{B}_{\blacklozenge}}$) corresponds to an edge in $\mathcal{G}_{\mathsf{B}_{\lozenge}}$ (or $\mathcal{G}_{\mathsf{C}}$ or $\mathcal{G}_{\mathsf{B}_{\blacklozenge}}$) with endpoints $(\partial_{0}e)|_{\mathsf{B}_\lozenge}$ (or $(\partial_0e)|_{\mathsf{C}}$ or $(\partial_{0}e)|_{\mathsf{B}_{\blacklozenge}}$). By construction, every canonical graphlike error in $\mathcal{E}_{\mathsf{B}_{\lozenge}}$ (or $\mathcal{E}_{\mathsf{C}}$ or $\mathcal{E}_{\mathsf{B}_{\blacklozenge}}$) has a level-0 syndrome contained in $\mathcal{D}_{\mathsf{B}_{\lozenge}}$ (or $\mathcal{D}_{\mathsf{C}}$ or $\mathcal{D}_{\mathsf{B}_{\blacklozenge}}$); every primitive hyperedge error in $\mathcal{E}_{\mathsf{B}_{\lozenge}}$ (or $\mathcal{E}_{\mathsf{C}}$ or $\mathcal{E}_{\mathsf{B}_{\blacklozenge}}$) triggers two detectors in $\mathcal{D}_{\mathsf{B}_{\lozenge}}$ (or $\mathcal{D}_{\mathsf{C}}$ or $\mathcal{D}_{\mathsf{B}_{\blacklozenge}}$).

Consider a working unit $\sigma$. If $\sigma$ is a core, define the $X$-basis (or $Z$-basis) decoding subgraph on $\sigma$ as the subgraph of $\mathcal{G}_{\mathsf{C}}$ induced by the vertex set of all level-0 $X$ (or $Z$) detectors on $\sigma$. If $\sigma$ is an $X$ bus, define the $X$-basis (or $Z$-basis) decoding subgraph on $\sigma$ as the subgraph of $\mathcal{G}_{\mathsf{B}_{\blacklozenge}}$ (or $\mathcal{G}_{\mathsf{B}_{\lozenge}}$) induced by the vertex set of all level-0 $X$ (or $Z$) detectors on $\sigma$. We similarly make these definitions for $Z$ buses. Define the $X$-basis (or $Z$-basis) segment on $\sigma$ between level-0 time steps $t_{0}$ and $t_{1}$ (including both $t_{0}$ and $t_{1}$) as the subgraph of the $X$-basis (or $Z$-basis) decoding subgraph on $\sigma$ induced by the vertex set of all level-0 $X$ (or $Z$) detectors with time coordinates between $t_{0}$ and $t_{1}$. For every $t_{0}\leq t\leq t_{1}$, we say the level-0 time step $t$ is contained in this segment. We refer to a segment on a core (or a shuttle bus) as a core (or bus) segment. 

\textit{Level-0 decoding subroutine}\textemdash{}Since $X$ (or $Z$) buses are always on the control (or target) side of hybrid-unit $\mathrm{CNOT}$ gates, we can see that (i) errors in $\mathcal{E}_{\mathsf{B}_{\lozenge}}$  are the only errors in $\mathcal{E}_{0}$ that trigger level-0 detectors in $\mathcal{D}_{\mathsf{B}_\lozenge}$, (ii) primitive hyperedge errors in $\mathcal{E}_{\mathsf{B}_{\lozenge}}$ may only have a pointy end in $\mathcal{D}_{\mathsf{C}}$, (iii) primitive hyperedge errors in $\mathcal{E}_{\mathsf{C}}$ may only have a pointy end in $\mathcal{D}_{\mathsf{B}_\blacklozenge}$, and (iv) errors in $\mathcal{E}_{\mathsf{B}_\blacklozenge}$ should always be canonical graphlike errors. Based on these structural properties of the three level-0 decoding graphs, we develop our level-0 decoding subroutine, Algorithm~\ref{Alg: level-0 decoding subroutine}, that takes a level-0 syndrome configuration $\eta\in\mathbb{Z}_2^{|\mathcal{D}_{0}|}$ induced by some level-0 error configuration $\epsilon\in\mathbb{Z}_2^{|\mathcal{E}_{0}|}$ and infers a level-0 error event $\kappa\in\mathbb{Z}_2^{|\mathcal{E}_{0}|}$ with $\partial_0\kappa=\partial_0\epsilon=\eta$. While the error configuration $\epsilon+\kappa$ triggers no level-0 syndrome, it may still trigger level-1 detectors and logical observables. We refer to the configuration $\epsilon+\kappa$ as the residual error configuration after level-0 decoding; we then need to perform level-1 decoding, which will be described later. 

\begin{algorithm}
        \let\oldnl\nl
        \newcommand{\nonl}{\renewcommand{\nl}{\let\nl\oldnl}}
       \caption{Level-0 decoding subroutine}\label{Alg: level-0 decoding subroutine}
       \KwIn{\justifying \small Level-0 syndrome configuration $\eta\in\mathbb{Z}_2^{|\mathcal{D}_{0}|}$.}
       \KwData{\justifying \small Level-0 decoding graphs $\mathcal{G}_{\mathsf{B}_{\lozenge}}$, $\mathcal{G}_{\mathsf{C}}$, and $\mathcal{G}_{\mathsf{B}_\blacklozenge}$.}   
       \KwOut{\justifying\small Level-0 error configuration $\kappa\in\mathbb{Z}_2^{|\mathcal{E}_{0}|}$ with $\partial_{0} \kappa=\eta$.} 
       \small \justifying
       Infer a level-0 error configuration $\kappa_{\mathsf{B}_{\lozenge}}$ (consisting only of $X$ errors on $X$ buses and $Z$ errors on $Z$ buses) by running a matching-based decoder on $\mathcal{G}_{\mathsf{B}_{\lozenge}}$ with the syndrome configuration $\eta|_{\mathsf{B}_{\lozenge}}$. 

       Infer a level-0 error configuration $\kappa_{\mathsf{C}}$ (consisting only of errors on cores) by running a matching-based decoder on $\mathcal{G}_{\mathsf{C}}$ with the corrected syndrome configuration $(\eta+\partial_{0}\kappa_{\mathsf{B}_{\lozenge}})|_{\mathsf{C}}$.   

        Infer a level-0 error configuration $\kappa_{\mathsf{B}_{\blacklozenge}}$ (consisting only of $X$ errors on $Z$ buses and $Z$ errors on $X$ buses) by running a matching-based decoder on $\mathcal{G}_{\mathsf{B}_{\blacklozenge}}$ with the corrected syndrome configuration $(\eta+\partial_{0}\kappa_{\mathsf{B}_{\lozenge}}+\partial_{0}\kappa_{\mathsf{C}})|_{\mathsf{B}_{\blacklozenge}}$. 

        \Return{$\kappa=\kappa_{\mathsf{B}_{\lozenge}}+\kappa_{\mathsf{C}}+\kappa_{\mathsf{B}_{\blacklozenge}}$.}
   \end{algorithm}

Note that we can combine Algorithm~\ref{Alg: level-0 decoding subroutine} with correlated matching~\cite{fowler_optimal_2013} to enhance the decoding performance by taking correlations between $X$ and $Z$ errors into account. More specifically, given a level-0 syndrome configuration induced by physical errors, we can first run Algorithm~\ref{Alg: level-0 decoding subroutine}, then change the weights of all decoding graphs based on the output following the approach in Ref.~\cite{fowler_optimal_2013}, and finally run Algorithm~\ref{Alg: level-0 decoding subroutine} again on reweighted decoding graphs. The output of the second run of Algorithm~\ref{Alg: level-0 decoding subroutine} is taken as the level-0 decoding result. We call the procedure above correlated level-0 decoding. We use the original level-0 decoding routine in Algorithm~\ref{Alg: level-0 decoding subroutine} for theoretical analysis and the correlated level-0 decoding for numerical benchmarking.  

\textit{Level-1 errors}\textemdash{}Consider a collection $\mathcal{M}$ of readout gadgets, such that the separation between every two $X$-basis (or $Z$-basis) readout gadgets in $\mathcal{M}$ is lower bounded by $\alpha_{\mathsf{c}} d_{0}$ level-0 time steps. We now define level-1 error locations and assign each a corresponding segment. For every core $\sigma$, readout gadgets in $\mathcal{M}$ assign a sequence of core-bus CNOT gates $\Lambda_{t_1},\cdots,\Lambda_{t_{i}}$ (at level-0 time steps $t_{1}<\cdots<t_{i}$) and a sequence of bus-core CNOT gates $\Lambda_{t_{1}'},\cdots,\Lambda_{t'_{j}}$ acting on $\sigma$ (at level-0 time steps $t_{1}'<\cdots<t_{j}'$). Between every two adjacent level-1 CNOT gates with $\sigma$ on the control side (corresponding to core-bus CNOT gates $\Lambda_{t_{c}}$ and $\Lambda_{t_{c+1}}$ with $c\in\{1,\cdots,i-1\}$), we assign a level-1 $X$ error location on the core and define the $Z$-basis segment on $\sigma$ between level-0 time steps $t_{c}+1$ and $t_{c+1}$ as the associated segment for the level-1 error location. We assign another two level-1 $X$ error locations on $\sigma$, one before the first level-1 CNOT gate corresponding to the core-bus CNOT gate $\Lambda_{t_1}$ and the other after the last level-1 CNOT gate corresponding to the core-bus CNOT gate $\Lambda_{t_{i}}$. We define $Z$-basis segments on $\sigma$ before level-0 time step $t_{1}$ and after level-0 time step $t_{i}+1$ as the associated segments for the first and the last level-1 $X$ error locations on $\sigma$, respectively. See Fig.~\ref{fig: level 1 error locations} for an illustration of all level-1 $X$ error locations set above. Similarly, for the sequence of bus-core CNOT gates, we assign level-1 $Z$ error locations on the core and define an associated segment for each level-1 $Z$ error location. For every $X$ bus $\sigma'$ of an $X$-basis readout gadget in $\mathcal{M}$, we assign a level-1 $Z$ error location on $\sigma'$, which does not propagate out of $\sigma'$ and only triggers a readout error on the readout gadget. We define the $X$-basis segment on $\sigma'$ spanning the lifetime of $\sigma'$ as the associated segment for this level-1 $Z$ error location. We similarly assign a level-1 $X$ error location for the $Z$ bus in every $Z$-basis readout gadget in $\mathcal{M}$ and define the associated bus segment for the level-1 $X$ error location. These level-1 error locations defined above allow us to perform level-1 decoding. Every level-1 $X$ (or $Z$) error on a level-1 $X$ (or $Z$) error location is called a primitive level-1 error. Every segment associated with a level-1 error location is called a canonical segment. 
\begin{figure}
    \centering
    \includegraphics[width=0.8\columnwidth]{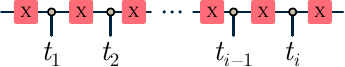}
    \caption{Level-1 $X$ error locations on a core. Only the control side of core-bus CNOT gates in $\mathcal{M}$ acting on the core is illustrated.}
    \label{fig: level 1 error locations}
\end{figure}

Every primitive level-1 $X$ (or $Z$) error with an associated canonical segment on a working unit $\sigma$ between level-0 time steps $t_1$ and $t_2$ is equivalent to a logical $X$ (or $Z$) error (as a level-0 error configuration) on $\sigma$ at the beginning of a level-0 time step $t$ with $t_1\leq t\leq t_2$. Let $\mathsf{E}$ be the collection of all primitive level-1 errors. A level-1 error configuration (as a subset of $\mathsf{E}$) is naturally identified with a vector in $\mathbb{Z}_2^{|\mathsf{E}|}$. Let $\overline{\partial}_{1}:\mathbb{Z}_2^{|\mathsf{E}|}\to \mathbb{Z}_2^{|\mathcal{D}_1|}$ be the level-1 syndrome map. We construct a level-1 decoding hypergraph $\mathsf{G}$ with the vertex set $\mathcal{D}_1$ and the hyperedge set $\mathsf{E}$, such that each hyperedge $\mathsf{e}\in\mathsf{E}$ is exactly connected to vertices in $\overline{\partial}_{1}\mathsf{e}$. 

\textit{Soft outputs and level-1 error probability}\textemdash{}For every level-1 error location, we can extract a soft output on the associated canonical segment during the level-0 decoding process to estimate the probability of this level-1 error. This error probability is then converted to the weight of the corresponding hyperedge in the level-1 decoding hypergraph. In the following, we briefly review how to extract the soft output~\cite{meister_efficient_2024} from a segment during a level-0 decoding shot with a sparse-blossom decoder~\cite{higgott_sparse_2025}. Consider a non-negatively weighted decoding graph $\mathcal{G}(w)$ ($\mathcal{G}$ for short), with $w$ the assignment of edge weights that maps each edge to a non-negative number. At the end of each decoding shot with a sparse-blossom decoder, a region $\mathcal{R}$ (combined graph fill region~\cite{higgott_sparse_2025}) with an associated radius map $r$ is constructed on $\mathcal{G}$, where $\mathcal{R}$ is a collection of vertices in $\mathcal{G}$ and $r$ assigns each vertex in $\mathcal{G}$ a non-negative radius. Based on the region, we define the postdecoding weight for each edge $e\in\mathcal{E}$ as $\tilde{w}(e):=\max(0,w(e)-\sum_{v\in\partial e}r(v))$, where $\partial e$ denotes the endpoints of $e$ (Fig.~\ref{fig: postdecoding weight}). Define a new graph $\tilde{\mathcal{G}}(\tilde{w})$ with the same vertex and edge structure as $\mathcal{G}$, and postdecoding weights for all edges. 
\begin{figure}[h]
    \centering
    \includegraphics[width=0.8\columnwidth]{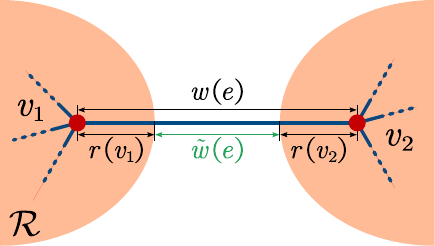}
    \caption{Postdecoding weight for an edge $e$ connecting two vertices $v_1$ and $v_2$ at the end of a decoding shot. The region $\mathcal{R}$ created during the decoding process is illustrated in orange.}
    \label{fig: postdecoding weight}
\end{figure}

We are now ready to describe the extraction of the soft output from a segment during a level-0 decoding routine. Consider an $X$-basis core segment on the core $\mathsf{c}$ containing time steps between $t_1$ and $t_2$ as shown in Fig.~\ref{fig: core segment soft output}. 
\begin{figure}[h]
    \centering
    \includegraphics[width=0.9\columnwidth]{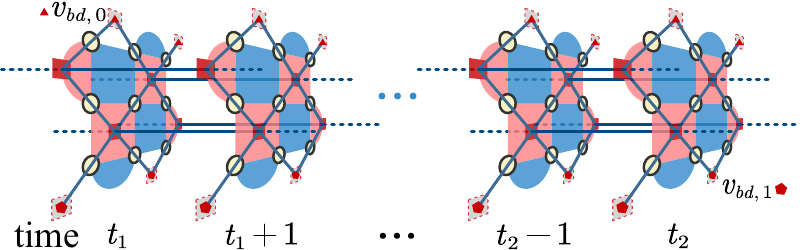}
    \caption{$X$-basis core segment (between level-0 time steps $t_{1}$ and $t_{2}$) embedded in the decoding graph $\mathcal{G}_{\mathsf{C}}$. Edges correspond to errors in a phenomenological depolarizing noise model. Vertices ($X$ detectors) in the segment are shown as small crimson squares. Each small crimson triangle (or pentagon) in a gray square denotes the virtual boundary vertex $v_{bd,0}$ (or $v_{bd,1}$) on the upper (or lower) boundary.}
    \label{fig: core segment soft output}
\end{figure}
There are two types of boundary edges (weight-1 edges) in the core segment, each corresponding to one of two opposing $Z$ boundaries. We attach a virtual boundary vertex $v_{bd,0}$ or $v_{bd,1}$ to boundary edges corresponding to either boundary. After decoding on $\mathcal{G}_{\mathsf{C}}$, we obtain a reweighted graph $\tilde{\mathcal{G}}_{\mathsf{C}}$ by applying the postdecoding reweight procedure described above. The soft output corresponding to the segment is defined as the minimum path length from $v_{bd,0}$ to $v_{bd,1}$ on $\tilde{\mathcal{G}}_{\mathsf{C}}$, and can be efficiently computed by Dijkstra's algorithm~\cite{meister_efficient_2024}. Soft outputs for other canonical segments are extracted similarly. 

\textit{Full decoding procedure}\textemdash{} Consider an HLP with $X$-basis and $Z$-basis LMSs under perfect time boundaries. Each LMS generates a logical observable representing the measurement result of the corresponding logical Pauli operator. Moreover, logical measurements at perfect time boundaries may generate additional logical observables. Let $\mathsf{L}$ be the collection of all logical observables. See Algorithm~\ref{Alg: full decoding procedure} for the full decoding procedure. Intuitively, we first perform a round of level-0 decoding. Then, at level 1, we decode logical observables supported on logical measurements at time boundaries independently from the other logical observables generated by LMSs. We then decode the logical observable induced by each LMS independently.

\begin{algorithm}
        \let\oldnl\nl
        \newcommand{\nonl}{\renewcommand{\nl}{\let\nl\oldnl}}
       \caption{Full decoding procedure for an HLP with basic modules.}\label{Alg: full decoding procedure}
       \KwIn{\justifying \small Level-0 syndrome configuration $\eta_{0}\in\mathbb{Z}_2^{|\mathcal{D}_{0}|}$, level-1 syndrome configuration $\eta_{1}\in\mathbb{Z}_2^{|\mathcal{D}_{1}|}$.}
       \KwData{\justifying \small An HLP with basic modules and perfect time boundaries.} 
       \KwOut{\justifying\small Inferred correction on logical observables $\mathbb{Z}_2^{|\mathsf{L}|}$.} 
       \small \justifying
      Construct level-0 decoding graphs and perform level-0 decoding according to Algorithm~\ref{Alg: level-0 decoding subroutine}. Let $\kappa$ be the level-0 decoding result. Perform the postdecoding reweight procedure on decoding graphs $\mathcal{G}_{\mathsf{C}}$ and $\mathcal{G}_{\mathsf{B}_{\blacklozenge}}$, and keep both reweighted graphs $\tilde{\mathcal{G}}_{\mathsf{C}}$ and $\tilde{\mathcal{G}}_{\mathsf{B}_{\blacklozenge}}$. Initialize an empty register $L\in\mathbb{Z}_2^{|\mathsf{L}|}$. Commit the correction to logical observables from $\kappa$ to $L$.  

        Consider logical observables supported on logical measurements at time boundaries. Let $\mathcal{M}_{\mathrm{SE}}$ denote the collection of all readout gadgets for level-1 SE. Assign level-1 errors and construct the level-1 decoding hypergraph for $\mathcal{M}_{\mathrm{SE}}$. For every hyperedge in the level-1 decoding hypergraph, extract the soft output from the associated canonical segment. Perform most-likely-error (MLE) decoding~\cite{gottesman_fault_tolerant_2014,cain_correlated_2024,zhou_low_overhead_2025} on the level-1 decoding hypergraph and add the correction on logical observables supported on final measurements to $L$. 

        \For{each logical measurement sequence (LMS)}{
        \nonl Let $\mathcal{M}_{1}$ denote the collection of all readout gadgets for level-1 SE and all logical readout gadgets in this LMS. Similar to the above step, construct the level-1 decoding hypergraph for $\mathcal{M}_1$, extract soft output for each hyperedge, and perform MLE decoding on the decoding hypergraph. Then, add the correction to the logical observable generated by the LMS to $L$. 
        }

        \Return{$L$.}
   \end{algorithm}

\section{Induced Level-1 Error Model}\label{sec: approximate level 1 error reduction}
In this section, we theoretically study how physical errors under the phenomenological depolarizing noise model, together with the level-0 decoding result, induce level-1 errors. We first set up the necessary notation for later discussion.

Let $\mathcal{E}_{\mathrm{dep}}$ be the collection of all primitive level-0 errors in the phenomenological depolarizing noise model. We refer to an element of $\mathbb{Z}_2^{|\mathcal{E}_{\mathrm{dep}}|}$, which specifies a subset of $\mathcal{E}_{\mathrm{dep}}$, as a general level-0 error configuration to distinguish it from vectors in $\mathbb{Z}_2^{|\mathcal{E}_{0}|}$, which are called level-0 error configurations. Since $\mathcal{E}_{0}$, containing only $X$ and $Z$ errors, is a subset of $\mathcal{E}_{\mathrm{dep}}$, there is a natural injection from $\mathbb{Z}_{2}^{|\mathcal{E}_{0}|}$ to $\mathbb{Z}_2^{|\mathcal{E}_{\mathrm{dep}}|}$. Therefore, every level-0 error configuration is also a general level-0 error configuration. We use the same notation $\partial_{0}$ and $\partial_{1}$ for the syndrome maps extended from level-0 error configurations to general level-0 error configurations. Define a linear map $\mathbb{P}_{\mathrm{X}}:\mathbb{Z}_2^{|\mathcal{E}_{\mathrm{dep}}|}\to\mathbb{Z}_2^{|\mathcal{E}_{0}|}$ that maps every primitive level-0 error in $\mathcal{E}_{\mathrm{dep}}$ to its $X$-basis component in $\mathcal{E}_{0}$. We similarly define $\mathbb{P}_{\mathrm{Z}}:\mathbb{Z}_2^{|\mathcal{E}_{\mathrm{dep}}|}\to\mathbb{Z}_2^{|\mathcal{E}_{0}|}$. For a general level-0 error configuration $g$, we define its $X$-basis (or $Z$-basis) component as $g_{X}:=\mathbb{P}_{\mathrm{X}}(g)$ (or $g_Z:=\mathbb{P}_{\mathrm{Z}}(g)$). 

Consider a configuration $\epsilon\in\mathbb{Z}_2^{|\mathcal{E}_{\mathrm{dep}}|}$ of level-0 physical errors triggering level-0 detectors $\eta=\partial_{0}\epsilon$.  We use Algorithm~\ref{Alg: level-0 decoding subroutine} to obtain an inferred error configuration $\kappa\in\mathbb{Z}_2^{|\mathcal{E}_{0}|}$ with $\partial_{0}\kappa=\eta$. Let $f:=\epsilon+\kappa$, which we call the general residual error configuration. Let $\mathcal{M}$ be the collection of readout gadgets whose measurement results are of interest. We construct level-1 error locations for $\mathcal{M}$ and let $\mathsf{E}$ be the collection of primitive level-1 errors. Let $\mathcal{D}_{\mathsf{B}_{\blacktriangledown}}\subset\mathcal{D}_{\mathsf{B}_{\blacklozenge}}$ be the subset of detectors in $\mathcal{D}_{\mathsf{B}_{\blacklozenge}}$ on buses in $\mathcal{M}$; let $\mathcal{E}_{\mathsf{B}_{\blacktriangledown}}$ be the subset of errors in $\mathcal{E}_{\mathsf{B}_{\blacklozenge}}$ on buses in $\mathcal{M}$. Let $\partial_{\blacktriangledown}$ be the restricted level-0 syndrome map, which restricts level-0 syndromes to $\mathcal{D}_{\mathsf{B}_{\lozenge}}\cup \mathcal{D}_{\mathsf{C}}\cup \mathcal{D}_{\mathsf{B}_{\blacktriangledown}}$. We note that level-0 errors in $\mathcal{E}_{\mathsf{B}_{\blacklozenge}}\backslash\mathcal{E}_{\mathsf{B}_{\blacktriangledown}}$ do not affect readout gadgets in $\mathcal{M}$. We say two level-0 error configurations $g_{1}$ and $g_{2}$ are $\mathcal{M}$-equivalent\textemdash{}denoted as $g_{1}\simeq_{\mathcal{M}}g_{2}$\textemdash{}if they differ by at most a space-time stabilizer (described later) and an error configuration in $\mathcal{E}_{\mathsf{B}_{\blacklozenge}}\backslash\mathcal{E}_{\mathsf{B}_{\blacktriangledown}}$. 

In Sec.~\ref{subsec: error propagation through cnots}, we prove that $f$ induces a level-1 error configuration in $\mathbb{Z}_2^{|\mathsf{E}|}$. This observation guarantees the existence of a solution for our level-1 decoding procedure\textemdash{}the residual error after the level-0 decoding procedure indeed corresponds to a level-1 error configuration. Without loss of generality, we assume $f\in\mathbb{Z}_2^{|\mathcal{E}_{0}|}$ by decomposing every $Y$ error in $\epsilon$ into its $X$- and $Z$-basis components. We will show that $f|_{\mathsf{B}_{\lozenge}}$, the residual error restricted to $\mathcal{E}_{\mathsf{B}_{\lozenge}}$, is equivalent to a level-0 error configuration $f^{\scriptscriptstyle\square}$ on cores. We say $f^{\scriptscriptstyle\square}$ is the propagated error from $f|_{\mathsf{B}_{\lozenge}}$. Then, the combination of the propagated error $f^{\scriptscriptstyle\square}$ and $f|_{\mathsf{C}}$ is shown to be $\mathcal{M}$-equivalent to the combination of a level-1 error configuration $\mathsf{f}_{\mathsf{C}}$ on cores and a level-0 error configuration $f^{\scriptscriptstyle\blacksquare}$ on shuttle buses in $\mathcal{M}$. Similarly, we say $f^{\scriptscriptstyle\blacksquare}$ is the propagated error configuration from $f^{\scriptscriptstyle\square}+f|_{\mathsf{C}}$, restricted to shuttle buses in $\mathcal{M}$. We will show that $f^{\scriptscriptstyle\blacksquare}+f|_{\mathsf{B}_{\blacktriangledown}}$ induces a level-1 error configuration $\mathsf{f}_{\mathsf{B}}$ on shuttle buses in $\mathcal{M}$. Finally, we see that $f$ induces the level-1 error configuration $\mathsf{f}_{\mathsf{C}}+\mathsf{f}_{\mathsf{B}}$ (or equivalently, $f\simeq_{\mathcal{M}}\mathsf{f}_{\mathsf{C}}+\mathsf{f}_{\mathsf{B}}$).

In Sec.~\ref{subsec: error reduction theorem}, we show the induced level-1 $X$ (or $Z$) errors almost satisfy a local stochastic noise model.   

\subsection{Error propagation through hybrid-unit CNOT gates}\label{subsec: error propagation through cnots}
We start with a general study of the error propagation phenomena and then analyze $f$ in detail.      

\textit{Space-time stabilizers and CNOT membranes}\textemdash{}A level-0 error configuration that triggers no (level-0 or level-1) detector or logical observable is called a space-time stabilizer~\cite{delfosse_spacetime_2023,williamson_low_overhead_2026,serra_peralta_decoding_2025}. We say two level-0 error configurations $g_1$ and $g_2$ are equivalent if they differ by a space-time stabilizer (this is stronger than $\mathcal{M}$-equivalence described above). We denote this equivalence relation as $g_1\simeq g_2$. Thus, we can deform a level-0 error configuration into another equivalent configuration by multiplying the original configuration by space-time stabilizers. For a data qubit $i$ on a working unit $\sigma$, denote the collection of $X$ (or $Z$) stabilizer generators of the working unit supported on qubit $i$ as $\delta_{X,i}$ (or $\delta_{Z,i}$). At a level-0 time step $t$, let $\Lambda_{t}$ denote the depth-1 layer of hybrid CNOT gates; let $e_{A,i,t}$ denote the $A$-basis ($A\in\{X,Z\}$) data qubit error on qubit $i$; let $e_{\mathrm{S},t}$ denote the measurement error on stabilizer $\mathrm{S}$. If qubit $i$ is not acted on by $\Lambda_{t}$, then $X$ errors $e_{X,i,t}$, $e_{X,i,t+1}$ on this qubit at the level-0 time steps $t$ and $t+1$, together with stabilizer measurement errors $\{e_{\mathrm{S},t}|\mathrm{S}\in\delta_{Z,i}\}$ at level-0 time step $t$, form a space-time stabilizer. Similarly, $e_{Z,i,t}+e_{Z,i,t+1}+\sum_{\mathrm{S}\in \delta_{X,i}}e_{\mathrm{S},t}$ is also a space-time stabilizer. Now consider the case where qubit $i$ is acted on by $\Lambda_{t}$. Denote the qubit that participated in the same physical CNOT gate as $i$ in $\Lambda_{t}$ as $j$. If qubit $i$ is on the target side of $\Lambda_t$, then $e_{Z,i,t}+e_{Z,i,t+1}+e_{Z,j,t}+\sum_{\mathrm{S}\in\delta_{X,i}}e_{\mathrm{S},t}$ is a space-time stabilizer. Notice that the level-0 error configuration $e_{Z,i,t}+e_{Z,i,t+1}+\sum_{\mathrm{S}\in\delta_{X,i}}e_{\mathrm{S},t}$ on the working unit $\sigma$ is now equivalent to the qubit error $e_{Z,j,t}$ on a different working unit. This gives rise to the error propagation phenomenon that we mentioned at the beginning of Sec.~\ref{sec: approximate level 1 error reduction}. We define a CNOT membrane associated with $\Lambda_{t}$ for the $X$-basis decoding subgraph on the working unit $\sigma$ as a plane with a time coordinate $t+\frac{1}{2}$ (only in the coordinate system of the subgraph). We say every $X$ stabilizer measurement error on $\sigma$ at level-0 time step $t$ is on the CNOT membrane. Moreover, we say a level-0 error configuration intersects with this CNOT membrane if the configuration contains at least one primitive level-0 error on the CNOT membrane. See Fig.~\ref{fig: spacetime stabilizer} for an illustration of space-time stabilizers and a CNOT membrane. 
\begin{figure}
    \centering
    \includegraphics[width=\columnwidth]{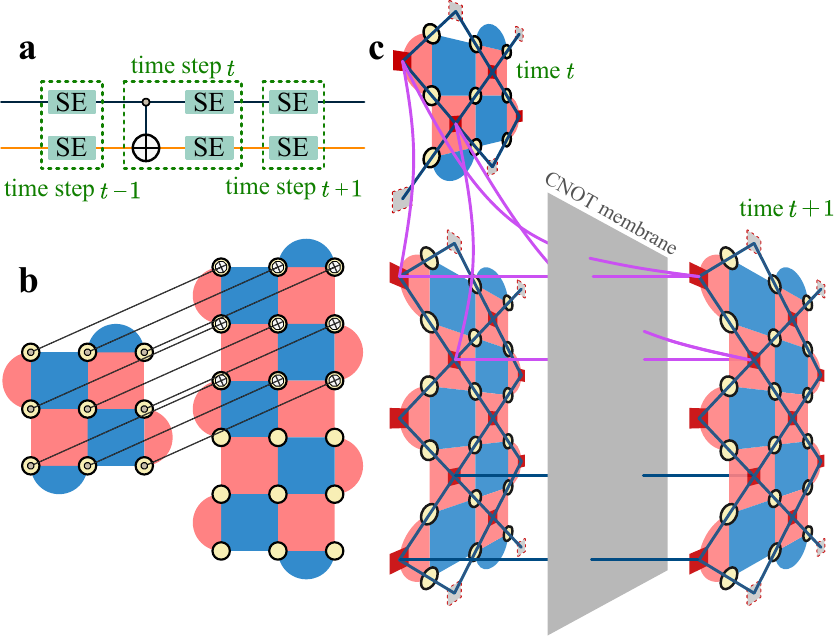}
    \caption{Space-time stabilizers and a CNOT membrane. (\textbf{a}) A level-0 circuit on a core and a $Z$ bus with a core-bus CNOT gate at level-0 time step $t$. (\textbf{b}) Illustration of the core-bus CNOT gate at level-0 time step $t$. (\textbf{c}) Illustration of primitive hyperedge errors, space-time stabilizers, and the CNOT membrane associated with the core-bus CNOT gate in (\textbf{a}). Each purple curved triangle corresponds to a primitive hyperedge $Z$ error (also an $X$ stabilizer measurement error on the CNOT membrane). The two purple curved triangles, together with three qubit errors sandwiched by these two triangles, form a space-time stabilizer. The two blue timelike edges (intersecting with the CNOT membrane), together with the two qubit errors sandwiched by these two timelike edges, also form a space-time stabilizer.}
    \label{fig: spacetime stabilizer}
\end{figure}
Similarly, if qubit $i$ is on the control side, then $e_{X,i,t}+e_{X,i,t+1}+e_{X,j,t}+\sum_{\mathrm{S}\in\delta_{Z,i}}e_{\mathrm{S},t}$ is a space-time stabilizer. In this case, we define a CNOT membrane associated with $\Lambda_t$ on $\sigma$ as a plane with a time coordinate $t+\frac{1}{2}$ in the coordinate system of the $Z$-basis decoding subgraph for $\sigma$. 

\textit{Walks, paths, and cycles}\textemdash{}Borrowing standard terminology for walks, paths, and cycles on a graph in graph theory~\cite{diestel_graph_2025}, we define these terms as sequences of primitive level-0 errors on a decoding subgraph in the following. Consider an $X$- or $Z$-basis decoding subgraph for a working unit $\sigma$. We introduce an additional virtual vertex $v_{\varnothing}$ (representing the lack of a vertex), which is considered an endpoint for every boundary edge in the decoding subgraph. A walk $\omega$ on the decoding subgraph is a sequence of primitive level-0 errors $e_1,\cdots,e_{r}$ (each corresponds to an edge in the decoding subgraph), such that there is a sequence of vertices $v_{0},\cdots,v_{r}$ on the decoding graph satisfying the following conditions: (i) only vertices $v_{0}$ and $v_{r}$ are allowed to be $v_{\varnothing}$ (in other words, only $e_{1}$ and $e_{r}$ are allowed to be boundary edges), and (ii) for every $i\in\{1,\cdots,r\}$, both $v_{i-1}$ and $v_{i}$ are endpoints of $e_{i}$. The vertex sequence $v_{0},\cdots,v_{r}$ is referred to as the vertex trail of $\omega$. A subwalk of $\omega$ is a contiguous subsequence of $\omega$. Each walk $\omega$ is associated with a level-0 error configuration, $[\omega]$, as the sum of all primitive level-0 errors in $\omega$ as vectors over $\mathbb{Z}_2$. This level-0 error configuration $[\omega]$ is called the error value of $\omega$. If $\omega$ contains no repeated primitive level-0 errors, then we can safely identify $\omega$ with $[\omega]$. We let $|\omega|$ denote the length of the sequence $\omega$ and $|[\omega]|$ denote the weight of $[\omega]$. Similarly, consider a collection of walks $\overline{\omega}$, we define its error value $[\overline{\omega}]$ as the sum of all primitive level-0 errors in $\overline{\omega}$ as vectors over $\mathbb{Z}_2$. We also identify $\overline{\omega}$ with $[\overline{\omega}]$ if $\overline{\omega}$ contains no repeated level-0 primitive errors; we let $|\overline{\omega}|$ denote the sum of lengths of all walks in $\overline{\omega}$ and $|[\overline{\omega}]|$ denote the weight of $[\overline{\omega}]$. We introduce two subscripts $||$ and $\perp$, representing the qubit-error component and the stabilizer-measurement-error component, respectively. In this way, $|\overline{\omega}_{||}|$ and $|\overline{\omega}_{\perp}|$ denote the number of occurrences of qubit errors and stabilizer measurement errors in $\overline{\omega}$ (rather than in its error value $[\overline{\omega}]$), respectively. Similarly, this definition also applies to individual walks and error configurations. We say a walk $\omega$ is closed if the starting vertex coincides with the ending vertex in the vertex trail of $\omega$. A simple walk is a walk with no repeated primitive level-0 errors. A walk $\omega$ with its vertex trail $v_{0},\cdots,v_{r}$ is a path if all these vertices are distinct; the walk $\omega$ is a cycle if the only repeated vertices in the vertex trail are $v_{0}$ and $v_{r}$. A closed walk can be decomposed as a collection of cycles; a walk that is not closed can be decomposed as a path and a collection of cycles. We note that every walk that we analyze in the following has no repeated primitive level-0 errors on any CNOT membrane.

\textit{Classification of closed walks on decoding graphs}\textemdash{}We classify closed walks in an $X$-basis or $Z$-basis decoding subgraph on a working unit into the following categories.
\begin{enumerate}
    \item Lightning walk. A lightning walk is a closed walk terminating at two opposite spatial boundaries. A canonical lightning walk is a lightning walk that does not intersect with any CNOT membrane. 
    \item Sprout walk. A sprout walk is a closed walk terminating at a time boundary. A benign sprout walk is a sprout walk that does not intersect with any CNOT membrane. A benign sprout walk does not propagate errors.  
    \item Cross-membrane walk. A cross-membrane walk is a closed walk that does not belong to the above two categories and intersects with at least one CNOT membrane.
    \item  Benign walk. A benign walk is a closed walk that does not belong to the above three categories. The error value of a benign walk is a level-0 space-time stabilizer. 
\end{enumerate}

We establish the following two lemmas. The first allows us to equivalently deform a lightning walk into a canonical lightning walk and a collection of benign and cross-membrane walks. The second shows that the error value of a sprout walk is equivalent to the combination of benign sprout cycles and cross-membrane cycles.
\begin{lemma}\label{lemma: lightning error decomposition}
    If a lightning walk $\ell$ on a working unit $\sigma$ intersects with at least one CNOT membrane and has no repeated edges on it, we can find a sequence $g$ of error paths on $\sigma$ such that (i) $g$ also forms a canonical lightning walk, (ii) each path in $g$ is immediately after the earliest CNOT membrane intersected by $\ell$, and (iii) we can rearrange edges in $\ell$ and paths in $g$ into a collection of cross-membrane walks and benign walks. We regard $g$ also as a canonical lightning walk and refer to it as the canonical retraction of $\ell$. 
\end{lemma}
\begin{proof}
     Denote the earliest CNOT membrane intersected by $\ell$ as $\mathfrak{m}$. As $\ell$ terminates at two opposing space boundaries, we introduce two virtual boundary vertices $v_{bd,0}$ and $v_{bd,1}$ for these two space boundaries, respectively. We connect each virtual boundary vertex to the boundary edges terminating on the corresponding space boundary. Without loss of generality, we suppose $\ell$ starts with a boundary edge connected to $v_{bd,0}$ and ends with a boundary edge connected to $v_{bd,1}$. By traversing $\ell$, we obtain an ordered sequence of time-like edges $e_1, e_2, \cdots, e_r$ on $\mathfrak{m}$, corresponding to stabilizer measurement errors. For each edge $e_i$, denote its endpoints before and after $\mathfrak{m}$ on $\sigma$ as $\partial_{<}e_i$ and $\partial_{>}e_i$, respectively.  Let $l_i$ ($0<i<r$) be the subwalk in $\ell$ between $e_{i}$ and $e_{i+1}$ (excluding both $e_{i}$ and $e_{i+1}$); let $l_{0}$ be the subwalk in $\ell$ from $v_{bd,0}$ to $e_1$ (excluding $e_{1}$); let $l_{r}$ be the subwalk in $\ell$ from $e_{r}$ to $v_{bd,1}$ (excluding $e_{r}$). In this way, $\ell$ is exactly the sequence $(l_0,e_{1},l_1,e_2,l_2,\cdots,e_r,l_r)$. Moreover, subwalks $l_0,\cdots,l_r$ alternate between the two sides of $\mathfrak{m}$. We introduce $g_0$ (or $g_r$) as the minimum-weight path connecting the vertex $\partial_{>}e_1$ (or $\partial_{>}e_{r}$) to the spatial boundary corresponding to $v_{bd,0}$ (or $v_{bd,1}$). For $i\in\{1,2,\cdots,r-1\}$, we introduce $g_{i}$ as the minimum-weight path connecting $\partial_{>}e_{i}$ and $\partial_{>}e_{i+1}$. Define the sequence $g:=(g_{0},\cdots,g_{r})$, which is immediately after $\mathfrak{m}$ by construction. Then, the sequence $g$ also forms a canonical lightning walk.  
     
     We now show that we can rearrange subwalks in the sequence representing $\ell$
     \begin{equation}\label{eq: subwalk sequence}
         (l_0,e_1,l_1,\cdots,e_r,l_r)
     \end{equation}
and paths in $g=(g_{0},\cdots,g_{r})$ into a collection of cross-membrane walks and benign walks. For every subwalk $l_a$ before $\mathfrak{m}$, we can pair it with its neighboring on-membrane edge(s) and a corresponding path in $g$ to form a cross-membrane walk; for every subwalk $l_b$ after $\mathfrak{m}$, we can pair it with a corresponding path in $g$ to form a benign walk. We describe the detailed construction as follows. If $l_0$ lies before $\mathfrak{m}$, we can see that $(l_0,e_{1},g_{0})$ forms a cross-membrane walk. Then in this case, for any non-negative even (or odd) integer $a\leq r$, $l_{a}$ is before (or after) $\mathfrak{m}$. Thus, for a non-negative odd index $a\leq r$, $(l_a,g_a)$ forms a benign walk; for a positive even index $a<r$, $(e_{a},l_a,e_{a+1},g_{a})$ forms a cross-membrane walk. Finally, if $r$ is even, then $(l_r,e_r,g_r)$ forms a cross-membrane walk. In this way, we rearrange subwalks in $\ell$ and paths in $(g_{0},\cdots,g_{r})$ into the collection of cross-membrane walks and benign walks above. We note that each subwalk of the sequence in Eq.~\ref{eq: subwalk sequence} and each path in $g$ is contained exactly once in these walks. A similar argument also holds if $l_0$ is after $\mathfrak{m}$. Therefore, we have successfully obtained $g$ that satisfies all three conditions in the lemma. See Fig.~\ref{fig: lightning cycle projection} for two examples illustrating the above procedure. 
\end{proof}
\begin{figure}
    \centering
    \includegraphics[width=\columnwidth]{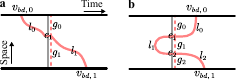}
    \caption{Two examples of lightning cycles and their canonical retractions. Solid lines represent lightning cycles; dashed lines represent canonical retractions of lightning cycles.}
    \label{fig: lightning cycle projection}
\end{figure}

\begin{lemma}
    \label{lemma: decomposition of a sprout cycle}  For a sprout walk $\ell$ that intersects with at least one CNOT membrane, we can find a level-0 error configuration $g$ on the same working unit as $\ell$, such that (i) $g$ is the combination of benign sprout cycles and (ii) $g+[\ell]$ is equivalent to a collection of cross-membrane cycles. 
\end{lemma}
\begin{proof}
    Denote the two boundary edges of $\ell$ as $e_1$ (a time-boundary edge) and $e_{-1}$, respectively. If $\ell$ only terminates at time boundaries, then define $g_1$ as the minimum-weight error path connecting the (non-virtual) endpoint of $e_1$ to a space boundary, and $g_{-1}$ as the minimum-weight path connecting the endpoint of $e_{-1}$ to the same space boundary. Define a level-0 error configuration $g:=(g_1+e_1)+(g_{-1}+e_{-1})$, which is composed of two benign sprout cycles $(g_{1}+e_{1})$ and $(g_{-1}+e_{-1})$. We can see that $g+[\ell]$ does not terminate at time boundaries and should be equivalent to a combination of cross-membrane cycles. On the other hand, if $\ell$ also terminates at a space boundary, then $e_{-1}$ is a space-boundary edge. Define a level-0 error configuration $g_1$ as the minimum-weight path connecting the vertex of $e_{1}$ to the same space boundary as $e_{-1}$. Then, we can find a benign sprout cycle $g:=g_1+e_1$ satisfying the requirements in the lemma.   
\end{proof}

\textit{Error propagation}\textemdash{}In the following, we show (i) how canonical $X$-basis (or $Z$-basis) lightning cycles on an $X$ (or $Z$) bus propagate errors to cores through CNOT membranes, and (ii) how cross-membrane walks propagate errors through CNOT membranes. The latter is used repeatedly in Sec.~\ref{subsec: error reduction theorem}.
\begin{figure}
    \centering
    \includegraphics[width=0.8\columnwidth]{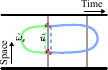}
    \caption{Cleaning a cross-membrane walk $\omega$ across its earliest intersected CNOT membrane $\mathfrak{m}$ (gray vertical line). Let $\Lambda_{t}$ be the hybrid-unit CNOT gate associated with $\mathfrak{m}$. The cross-membrane walk $\omega$ is illustrated as the loop with three colors. Edges $e_{1}$ and $e_{2}$ in $\omega$ on $\mathfrak{m}$ are illustrated in red. Let $v_{1}$ and $v_{2}$ be the endpoints of $e_1$ and $e_2$, respectively, that lie before $\mathfrak{m}$. According to Algorithm~\ref{Alg: membrane crossing}, we can find a subwalk $\tilde{\omega}_s$ (green curve) in $\omega$ before $\mathfrak{m}$ connecting $v_{1}$ to $v_{2}$. The error configuration $\tilde{u}$ (green dashed line) corresponds to a minimum-weight perfect matching of $\{v_1,v_2\}$. The propagated error to the opposite side of $\Lambda_{t}$ is obtained as $\Lambda_t(\tilde{u})+\tilde{u}$. The blue curve, together with the blue dashed line, forms a closed walk after $\mathfrak{m}$, which is the output of Algorithm~\ref{Alg: membrane crossing}.}
    \label{fig: membrane crossing}
\end{figure}
\begin{lemma}
    \label{lemma: propagation of lightning cycle}
    Consider a $Z$-basis (or $X$-basis) canonical lightning cycle $\ell$ on a $Z$ (or $X$) bus $\sigma$. Suppose CNOT membranes on this bus after $\ell$ correspond to hybrid-unit CNOT gates $\Lambda_{t_1},\cdots,\Lambda_{t_{o}}$ at level-0 time steps $t_{1}<\cdots<t_{o}$, respectively. Then, the lightning cycle $\ell$ propagates to canonical lightning cycles at level-0 time step $t_i$ on every core on the control side (or target side) of the core-bus (or bus-core) CNOT gate $\Lambda_{t_i}$ for all $i\in\{1,\cdots,o\}$.   
\end{lemma}
\begin{proof}
   Without loss of generality, we assume $\sigma$ is a $Z$ bus. Let $\ell_\sigma$ be the logical operator along a long edge of $\sigma$; let $\ell_{\sigma,t}$ be the level-0 error configuration corresponding to $\ell_{\sigma}$ at level-0 time step $t$. Then, the cycle $\ell$ is equivalent to the level-0 error configuration $\ell_{\sigma,t_1}$. We can see that $\ell_{\sigma,t_1}$ is equivalent to $\left(\Lambda_{t_1}(\ell_{\sigma,t_1})+\ell_{\sigma,t_1}\right)+\ell_{\sigma,t_1+1}$, where the first term is the product of logical $Z$ operators on every core at the control side of $\Lambda_{t_1}$ at level-0 time step $t_1$ and the second term (at level-0 time step $t_1+1$) is equivalent to $\ell_{\sigma,t_2}$. The lemma is proven by recursively applying the argument above.      
\end{proof}
\begin{lemma}
    [Error propagation of a cross-membrane walk]\label{lemma: propagation of cross-membrane cycle}
Consider a cross-membrane walk $\ell$ on a working unit $\sigma$ that intersects the CNOT membranes associated with hybrid-unit CNOT gates $\Lambda_{t_1},\cdots,\Lambda_{t_o}$ at level-0 time steps $t_{1}<\cdots<t_{o}$, respectively. Then, the level-0 error configuration $[\ell]$ is equivalent to a level-0 error configuration $u:=u_{t_1}+\cdots+u_{t_o}$, where each $u_{t_i}$ ($i\in\{1,\cdots,o\}$) as a level-0 error configuration is a collection of simple walks at level-0 time step $t_{i}$ on working units on the opposite side of $\Lambda_{t_i}$ to $\sigma$. (We also say $\ell$ propagates to $u$.) Moreover, for each $u_{t_i}$, (i) every simple walk in $u_{t_i}$ either terminates at a spatial boundary or at the pointy end of an error in $\ell$ on the CNOT membrane at $t_i+\frac{1}{2}$, and (ii) its weight, $|u_{t_i}|$, is upper bounded by $|\ell_{||}|/2$, with $|\ell_{||}|$ the number of occurrences of qubit errors in $\ell$. 
\end{lemma}

\begin{algorithm}[h]
        \let\oldnl\nl
        \newcommand{\nonl}{\renewcommand{\nl}{\let\nl\oldnl}}
       \caption{Subroutine for deforming a cross-membrane walk across the earliest CNOT membrane}\label{Alg: membrane crossing}
       \KwIn{\justifying\small A cross-membrane walk $\omega$ on a working unit $\sigma$ with no repeated edges on CNOT membranes. Suppose the earliest CNOT membrane $\mathfrak{m}$ (induced by a hybrid-unit CNOT gate $\Lambda_{t}$) intersected by $\omega$ has a time coordinate $t+\frac{1}{2}$.}
       \KwOut{\justifying\small A closed walk $\omega'$ on the same working unit supported on edges after $\mathfrak{m}$. A level-0 error configuration $u$ of qubit errors at level-0 time step $t$ on working units on the opposite side of $\Lambda_t$ to $\sigma$.}
       \small \justifying
       Denote edges in $\omega$ on $\mathfrak{m}$ as $e_1,\cdots,e_{r}$. Initialize an empty register $W$, which stores a set of walks. 

       \For{each edge $e_{i}$ ($i\in\{1,\cdots,r\}$)} {
       \nonl Denote its endpoints on the working unit $\sigma$ as $v_{i,t}$, and $v_{i,t+1}$, representing detectors with time coordinates $t$ and $t+1$, respectively. In the walk $\omega$, there exists a subwalk $\tilde{\omega}_s$ before $\mathfrak{m}$ that connects $e_i$ to another edge $e_j$ with $j\in\{1,\cdots,r\}$ or to a space boundary. 
If $\tilde{\omega}_s$ connects $e_i$ to $e_j$, then we say $e_i$ is matched to $e_{j}$. Moreover, there is a subwalk $\omega_{s}$ of $\omega$ composed exactly of $e_{i}$, $e_{j}$, and $\tilde{\omega}_{s}$. On the other hand, if $\tilde{\omega}_{s}$ connects $e_{i}$ to a space boundary, we say $e_{i}$ is matched to the space boundary. In this case, there is a subwalk $\omega_{s}$ of $\omega$ composed exactly of $e_{i}$ and $\tilde{\omega}_{s}$. If $\omega_{s}\notin W$, insert $\omega_{s}$ into $W$.}
        \nonl Now $W$ contains walks corresponding to mutually disjoint edge sequences in $\omega$, such that every occurrence of an edge in $\omega$ on or before $\mathfrak{m}$ is contained in exactly one walk in $W$.     

        In this step, we temporarily attach boundary edges on two opposing space boundaries on $\sigma$ with virtual boundary vertices $v_{bd,0}$, and $v_{bd,1}$, respectively. If either zero or two edges in $\omega$ on $\mathfrak{m}$ are matched to a space boundary, define $\tilde{u}$ as a combination of qubit errors at time step $t$ corresponding to a minimum-weight perfect matching of $\{v_{1,t},\cdots,v_{r,t}\}$. Otherwise, if an edge $e_{i}$ ($i\in\{1,\cdots,r\}$) is matched to a space boundary with a space boundary vertex $v_{bd}\in\{v_{bd,0},v_{bd,1}\}$, then define $\tilde{u}$ as a combination of qubit errors at time step $t$ corresponding to a minimum-weight perfect matching of $\{v_{1,t},\cdots,v_{r,t},v_{bd}\}$. Define a level-0 error configuration, $u:=\Lambda_{t}(\tilde{u})+\tilde{u}$, consisting of qubit errors at time step $t$ on the opposite side of $\Lambda_{t}$ to the working unit $\sigma$. 
       
    Define $\omega'$ as another closed walk on $\sigma$ by replacing every subwalk $\omega_{s}\in W$ from $\omega$ with a minimum-weight error path on time step $t+1$ connecting the endpoints of $\omega_s$.  

    \Return{$\omega'$ and $u$}
   \end{algorithm}

\begin{proof}
    The proof proceeds by recursively deforming $\ell$ across the earliest CNOT membrane intersected by it and recording the corresponding propagated error. The deformation is implemented by multiplying $\ell$ with space-time stabilizers. More specifically, in Algorithm~\ref{Alg: membrane crossing}, we show how to obtain a closed walk $\ell_1$ and a level-0 error configuration $u_{t_1}$ with $[\ell]\simeq [\ell_1]+u_{t_1}$ such that (i) the closed walk $\ell_1$ lies after the CNOT membrane at $t_{1}+\frac{1}{2}$ and only intersects with $o-1$ CNOT membranes (corresponding to hybrid-unit CNOT gates $\Lambda_{t_2},\cdots,\Lambda_{t_{o}}$), and (ii) the error configuration $u_{t_1}$ consists of qubit-error simple walks on the opposite side of $\Lambda_{t_1}$ and satisfies condition (i) in the lemma. Denote the weight of qubit errors in $[\ell]$ before (or after) the CNOT membrane $\mathfrak{m}$ as $|[\ell]_{||,<}|$ (or $|[\ell]_{||,>}|$). According to the minimum-weight constraint in the third step of Algorithm~\ref{Alg: membrane crossing}, we can see that $|u_{t_1}|\leq |[\ell]_{||,<}|$ and $|u_{t_1}|\leq|[\ell]_{||,>}|$. Thus, the weight of $u_{t_1}$ is upper bounded by $\min(|[\ell]_{||,<}|,|[\ell]_{||,>}|)\leq|[\ell]_{||}|/2\leq |\ell_{||}|/2$. See Fig.~\ref{fig: membrane crossing} for an example illustrating the procedure in Algorithm~\ref{Alg: membrane crossing}.
 
    We can then apply Algorithm~\ref{Alg: membrane crossing} again on $\ell_1$ to obtain another closed walk $\ell_2$ (intersecting with $o-2$ CNOT membranes) and a level-0 error configuration $u_{t_2}$ with $[\ell_1]\simeq [\ell_{2}]+u_{t_2}$. Thus, by recursively applying Algorithm~\ref{Alg: membrane crossing}, we can obtain a series of closed walks $\ell_{1},\cdots,\ell_{o}$ and a series of level-0 error configurations $u_{t_1},\cdots,u_{t_{o}}$, such that (i) $\ell_{i}$ intersects with $o-i$ CNOT membranes corresponding to hybrid-unit CNOT gates $\Lambda_{t_{i+1}},\cdots,\Lambda_{t_{o}}$ for every $i\in\{1,\cdots,o-1\}$, (ii) $[\ell_{o}]$ is a space-time stabilizer, (iii) every $u_{t_i}$ satisfies the first requirement in the lemma, and (iv) $[\ell_{i}]\simeq[\ell_{i+1}]+u_{t_{i+1}}$ for every $i\in\{1,\cdots,o-1\}$. Thus, we can see that $[\ell]\simeq [\ell_{1}]+u_{t_1}\simeq\cdots\simeq u_{t_1}+\cdots+u_{t_{o}}$. Notice that each $u_{t_i}$ can also be obtained by erasing all CNOT membranes before the $i$th CNOT membrane and then running Algorithm~\ref{Alg: membrane crossing} on $\ell$. Thus, the weight of $u_{t_i}$ is also upper bounded by $|\ell_{||}|/2$.  
\end{proof}
We now return to the analysis of the level-0 residual error configuration $f\in\mathbb{Z}_2^{|\mathcal{E}_{0}|}$ introduced at the beginning of this section. 
\begin{lemma}
    \label{lemma: syndromeless level 0 error induces a level 1 error} The level-0 residual error configuration $f$ is $\mathcal{M}$-equivalent to a level-1 error configuration.
\end{lemma}
\begin{proof}
    On a high level, this proof successively cleans closed walks in $\mathcal{E}_{\mathsf{B}_\lozenge}$, cores, and finally in $\mathcal{E}_{\mathsf{B}_{\blacktriangledown}}$. At each stage, closed walks either propagate to the next stage or are absorbed as level-1 errors.  
    As $(\partial_{0}f|_{\mathsf{B}_{\lozenge}})|_{\mathsf{B}_{\lozenge}}=(\partial_{0}f)|_{\mathsf{B}_\lozenge}=0$, $f|_{\mathsf{B}_{\lozenge}}$ is a collection of cycles, each of which is either an $X$-basis cycle on an $X$ bus or a $Z$-basis cycle on a $Z$ bus. According to Lemma~\ref{lemma: lightning error decomposition} and Lemma~\ref{lemma: decomposition of a sprout cycle}, $f|_{\mathsf{B}_{\lozenge}}$ can be equivalently deformed into a collection of benign sprout cycles, benign cycles, canonical lightning cycles, and cross-membrane cycles. The first two types of cycles are equivalent to no error. The last two types of cycles are equivalent to a level-0 error configuration $f^{\scriptscriptstyle \square}$ on cores according to Lemma~\ref{lemma: propagation of lightning cycle} and Lemma~\ref{lemma: propagation of cross-membrane cycle}. Thus, we see that $f|_{\mathsf{B}_{\lozenge}}\simeq f^{\scriptscriptstyle\square}$. Similarly, as $(\partial_{0}f^{\scriptscriptstyle\square}+\partial_{0}f|_{\mathsf{C}})|_{\mathsf{C}}=(\partial_{0}f|_{\mathsf{B}_{\lozenge}}+\partial_{0}f|_{\mathsf{C}})|_{\mathsf{C}}=(\partial_0f)|_{\mathsf{C}}=0$, we can see that $f^{\scriptscriptstyle\square}+f|_{\mathsf{C}}$ is equivalent to a collection of canonical lightning cycles and cross-membrane cycles on cores. Each canonical lightning cycle on a core is equivalent to the level-1 primitive core error, whose associated canonical segment contains this lightning cycle. Let $\mathsf{f}_{\mathsf{C}}$ be the level-1 error configuration on cores induced by the canonical lightning cycles in $f^{\scriptscriptstyle\square}+f|_{\mathsf{C}}$. The cross-membrane cycles on cores introduced by $f^{\scriptscriptstyle\square}+f|_{\mathsf{C}}$ propagate to a level-0 error configuration $f^{\scriptscriptstyle\blacksquare}$ of $Z$ errors on $X$ buses and $X$ errors on $Z$ buses in $\mathcal{M}$ according to Lemma~\ref{lemma: propagation of cross-membrane cycle}. Finally, as $(\partial_{0}f^{\scriptscriptstyle\blacksquare}+\partial_{0}f|_{\mathsf{B}_{\blacktriangledown}})|_{\mathsf{B}_{\blacktriangledown}}=(\partial_{0}f)|_{\mathsf{B}_{\blacktriangledown}}=0$, the error configuration $f^{\scriptscriptstyle\blacksquare}+f|_{\mathsf{B}_{\blacktriangledown}}$ is equivalent to a collection of canonical lightning cycles on shuttle buses in $\mathcal{M}$, which incurs a level-1 error configuration $\mathsf{f}_{\mathsf{B}}$ on them. Thus, according to the following equation, $f$ is $\mathcal{M}$-equivalent to the level-1 error configuration $\mathsf{f}_{\mathsf{C}}+\mathsf{f}_{\mathsf{B}}$.
    \begin{align}\large
        f=&f|_{\mathsf{B}_{\lozenge}} + f|_{\mathsf{C}} + f|_{\mathsf{B}_{\blacklozenge}} 
        \simeq f^{\scriptscriptstyle\square} + f|_{\mathsf{C}}+f|_{\mathsf{B}_{\blacklozenge}}\nonumber\\
        \simeq&_{\mathcal{M}}\ \mathsf{f}_{\mathsf{C}} + f^{\scriptscriptstyle\blacksquare}+f|_{\mathsf{B}_{\blacktriangledown}} \simeq \mathsf{f}_{\mathsf{C}}+\mathsf{f}_{\mathsf{B}}
    \end{align}
\end{proof}

\subsection{Error Reduction Theorem}\label{subsec: error reduction theorem}
Throughout this subsection, we fix the compilation parameter $\alpha_{\mathsf{b}}=d_{1}$ and ensure that every pair of $X$-basis or $Z$-basis readout gadgets in $\mathcal{M}$ is separated by at least $d_{0}$ level-0 time steps. In this case, we will show that the probability distribution of induced level-1 $X$ (or $Z$) errors satisfies a local stochastic noise model except for a rare event whose probability is exponentially suppressed in $d_{0}d_1$.

Without loss of generality, we consider level-1 $X$ errors. For now, we assume that $f$ consists only of level-0 $X$ errors. We now perform a more fine-grained analysis on $f$.  According to the proof of Lemma~\ref{lemma: syndromeless level 0 error induces a level 1 error}, we know that there are two types of errors on cores: one is the propagated errors (each is a simple walk according to Lemma~\ref{lemma: propagation of cross-membrane cycle}) from $X$-basis closed walks on $X$ buses; the other is the native errors in $f|_{\mathsf{C}}$. Similarly, there are two types of $X$-basis errors on $Z$ buses: one propagated from cores and the other in $f|_{\mathsf{B}_{\blacktriangledown}}$. We often need to consider propagated errors and native errors in $f$ as distinct objects to track their respective weights. In this section, we regard $f$ as a set of length-1 walks, referred to as native walks, each of which corresponds to a unique primitive level-0 error in $f$. Then, every level-0 error configuration contained in $f$ is regarded as a subset of $f$. Propagated simple walks are treated as distinct objects and are not identified with native walks, even when they have overlapping supports.  

We know that $f|_{\mathsf{B}_{\lozenge}}$ is a collection of error cycles $\{\tilde{\ell}_{1},\cdots,\tilde{\ell}_{r}\}$ on $X$ buses. We restrict to the case where each cycle in $f|_{\mathsf{B}_{\lozenge}}$ has a weight smaller than $d_{0}d_{1}$. (The probability for a cycle of weight at least $d_0d_1$ is so heavily suppressed that we can regard such an event as failure without analyzing it.) In this case, every cycle in $f|_{\mathsf{B}_{\lozenge}}$ is either a benign sprout cycle, a benign cycle, or a cross-membrane cycle that only intersects with a single CNOT membrane. For each $\tilde{\ell}_{i}$ ($i\in\{1,\cdots,r\}$), let $\tilde{\ell}_{i}^{\circ}$ be a set of propagated simple walks, representing the propagated errors from $\tilde{\ell}_i$. If $\tilde{\ell}_i$ is a benign cycle or a benign sprout cycle, then $\tilde{\ell}_i^{\circ}$ is an empty set. Otherwise, we obtain $\tilde{\ell}_{i}^{\circ}$ as a set of propagated simple walks on cores, following Lemma~\ref{lemma: propagation of cross-membrane cycle}. We say $\tilde{\ell}_i$ is the \textit{source} for every propagated simple walk in $\tilde{\ell}_{i}^{\circ}$. Define $f^{\circ}:=\bigcup \tilde{\ell}_i^{\circ}$, the collection of propagated simple walks on cores from $\{\tilde{\ell}_1,\cdots,\tilde{\ell}_r\}$. We note that $f^{\circ}$ may contain repeated level-0 primitive errors; thus, we need to differentiate the set from its error value $[f^{\circ}]$. Moreover, as $\tilde{\ell}_i\simeq \tilde{\ell}_i^{\circ}$, we know that the error configuration $[f^{\circ}]$ is equivalent to $f|_{\mathsf{B}_{\lozenge}}$. The propagated simple walks in $f^{\circ}$, together with the native walks in $f|_{\mathsf{C}}$, can be partitioned into a set of closed walks $\{\ell_{1},\cdots,\ell_{o}\}$ on cores, such that (i) each closed walk $\ell_{j}$ is composed of propagated simple walks in $f^{\circ}$ and native walks in $f|_{\mathsf{C}}$ and (ii) every simple walk in $f^{\circ}$, as well as every native walk in $f|_{\mathsf{C}}$, is contained in exactly one closed walk in $\{\ell_{1},\cdots,\ell_{o}\}$. We also view each closed walk $\ell_{j}$ as a set of the constituent propagated simple walks in $f^{\circ}$ and native walks in $f|_{\mathsf{C}}$. In this way, the disjoint union of two sets of walks $f^{\circ}\cup f|_{\mathsf{C}}$ is also the disjoint union of $\ell_{1},\cdots,\ell_{o}$. 

For each $\ell_j$, we will define its propagation $\ell_j^{\bullet}$ as a set of simple walks on $Z$ buses in $\mathcal{M}$; and we also say $\ell_j$ is the \textit{source} for every simple walk in $\ell_j^{\bullet}$. If $\ell_j$ is a benign walk, then $\ell_j^{\bullet}$ is an empty set. If $\ell_j$ is a cross-membrane walk, define $\ell_j^{\bullet}$ as the set of simple walks propagated from $\ell_j$ according to Lemma~\ref{lemma: propagation of cross-membrane cycle}. If $\ell_j$ is a lightning walk, according to Lemma~\ref{lemma: lightning error decomposition}, we can find the canonical retraction $g_j$ of $\ell_j$, such that we can arrange $\ell_j$ and paths in $g_{j}$ into a collection of cross-membrane walks and benign walks. Define $\ell_j^{\bullet}$ as the set of all simple walks propagated from these cross-membrane walks. Define $f^{\bullet}:=\bigcup \ell_j^{\bullet}$ as the set of all simple walks on buses in $\mathcal{M}$ propagated from $\{\ell_1,\cdots,\ell_o\}$. Similarly, the set of walks $f^{\bullet}\cup f|_{\mathsf{B}_{\blacktriangledown}}$ can be decomposed as a disjoint union of closed walks on $Z$ buses. We refer to each such closed walk as a closed walk in $f^{\bullet}\cup f|_{\mathsf{B}_{\blacktriangledown}}$. We say a closed walk in $f^{\bullet}\cup f|_{\mathsf{B}_{\blacktriangledown}}$ is native if the closed walk only consists of native walks in $f|_{\mathsf{B}_{\blacktriangledown}}$. We say two closed walks $\ell_{a}$ and $\ell_{b}$ on cores are linked if a simple walk in $\ell_a^{\bullet}$ and a simple walk in $\ell_{b}^{\bullet}$ belong to the same closed walk in $f^{\bullet}\cup f|_{\mathsf{B}_{\blacktriangledown}}$ on a $Z$ bus. We can then construct a graph $\mathcal{W}_{\mathsf{C}}$ with closed walks in $\{\ell_{1},\cdots,\ell_o\}$ as vertices, such that two vertices are connected by an edge if and only if the corresponding closed walks are linked. 
\begin{definition}
    [Downstream web]\label{def: downstream fault web}
    Consider a collection $w$ of $X$-basis walks on cores and native walks on $Z$ buses, such that (i) the subcollection of walks on cores, $w|_\mathsf{C}$, is the union of a subcollection of closed walks (each as a set of walks) in $\{\ell_1,\cdots,\ell_{o}\}$, (ii) the subcollection of native walks on $Z$ buses, $w|_{\mathsf{B}_{\blacktriangledown}}\subset f|_{\mathsf{B}_{\blacktriangledown}}$ contains no native closed walks, and (iii) $\partial_{\blacktriangledown}[w]=0$. We say $w$ is a downstream web if $w|_{\mathsf{C}}$ corresponds to a connected component in the graph $\mathcal{W}_{\mathsf{C}}$. We say a closed walk $\ell_{i}\in\{\ell_{1},\cdots,\ell_{o}\}$ is in $w$ if $\ell_i$ is a subset of $w$. Moreover, we say a downstream web $w$ is a native downstream web if every closed walk in $w|_{\mathsf{C}}$ contains no simple walks in $f^{\circ}$. (Notice that every native downstream web $w$ is composed of native walks and is contained in $f$.) 
\end{definition} 
In this way, we can see that $f^{\circ}\cup( f|_{\mathsf{C}}\cup f|_{\mathsf{B}_{\blacktriangledown}})$ is a disjoint union of downstream webs and native closed walks on $Z$ buses. 

We now set bounds on native downstream webs and then set bounds on $f$. We start with two definitions. First, we define a level-0 adjacency graph $\mathcal{A}$~\cite{gottesman_fault_tolerant_2014}, whose vertices are identified with level-0 error locations, such that every two vertices are connected by an edge if and only if their corresponding primitive level-0 errors overlap in their level-0 syndromes. This construction~\cite{gottesman_fault_tolerant_2014} is a standard technique in fault-tolerance proofs to help with combinatorial analysis of errors. Secondly, we define sets of primitive level-0 $X$ errors in $\mathcal{E}_{0}$ that may contribute to primitive level-1 $X$ errors on cores and $Z$ buses, respectively.
\begin{definition}
    [Receptive zone for a primitive level-1 error]\label{def: receptive zone for X or Z error}
    Consider a primitive level-1 $X$ error $\mathsf{e}$ associated with a canonical segment $\mathrm{G}$ on a working unit $\sigma$. The receptive zone for $\mathsf{e}$ is defined in the following as a set of level-0 $X$ error locations. If $\sigma$ is a core, the receptive zone is the collection of all level-0 $X$ error locations satisfying at least one of the following conditions.
    \begin{itemize}
    \item The level-0 error location is on $\sigma$ and triggers at least one level-0 detector contained in the canonical segment immediately before $\mathrm{G}$ or the canonical segment $\mathrm{G}$.
    \item The error location is on an $X$ bus $\sigma'$, which is on the control side of a bus-core CNOT gate $\Lambda_t$. The gate $\Lambda_{t}$ acts on $\sigma$ at level-0 time step $t$ contained in the segment immediately before $\mathrm{G}$ or the segment $\mathrm{G}$. This error location is on a data qubit and occurs at a level-0 time step no later than $t$ and later than $t-d_{0}d_1$. 
    \end{itemize}
On the other hand, if $\sigma$ is a $Z$ bus, the receptive zone is the collection of all level-0 $X$ error locations satisfying at least one of the following conditions. 
\begin{itemize}
        \item The error location is on $\sigma$ and triggers at least one level-0 $\mathrm{Z}$ detector. 
        \item  The error location is on a CNOT membrane of a core induced by a core-bus CNOT gate $\Lambda_{t}$ at level-0 time step $t$, such that $\Lambda_{t}$ acts on $\sigma$.
\end{itemize}
\end{definition}

\begin{lemma}
    [Bounds on a native downstream web]\label{lemma: bounds on a downstream web}
Consider a native downstream web $w$. Denote the collection of physical errors in $w$ as $w_{\epsilon}$ and the collection of inferred errors in $w$ as $w_{\kappa}$. We can explicitly construct a level-1 error $\mathsf{w}=\mathsf{e}_1+\mathsf{e}_2+\cdots+\mathsf{e}_{|\mathsf{w}|}$, where each $\mathsf{e}_{i}$ is a distinct primitive level-1 error, such that $w$ induces the level-1 error $\mathsf{w}$. Moreover, we can find two sets of native walks (level-0 error configurations) $\lambda\subset\mu\subset w$ such that (i) $\lambda=\{e_1,\cdots,e_{|\mathsf{w}|}\}$ with each $e_{i}$ a distinct primitive level-0 error in the receptive zone of $\mathsf{e}_{i}$, (ii) each connected component of $\mu$ (on the level-0 adjacency graph) contains at least one element in $\lambda$, and (iii) the number of physical errors in $\mu$, $|\mu\cap w_{\epsilon}|$, is lower bounded by $\max(|\mu|/4,d_0|\mathsf{w}|/2)$. 
\end{lemma}
  
\begin{proof}
We explicitly construct the level-1 error $\mathsf{w}$, and the sets $\mu$ and $\lambda$ following Algorithm~\ref{Alg: subroutine for analyzing a downstream fault web}. (As $w$ is a native downstream web, the register $\mu$ in Algorithm~\ref{Alg: subroutine for analyzing a downstream fault web} is a set of native walks.) By construction, conditions (i) and (ii) in the lemma are satisfied. 

   \begin{algorithm}
        \let\oldnl\nl
        \newcommand{\nonl}{\renewcommand{\nl}{\let\nl\oldnl}}
       \caption{Subroutine for analyzing a downstream web (Part 1)}\label{Alg: subroutine for analyzing a downstream fault web}
       \KwData{\justifying \small A downstream web $w$. Registers $\mathsf{w}$, $\mu$, $\lambda$, $\mu_{\mathrm{m}}$, and $\mu_{\mathrm{l}}$, where $\mathsf{w}$ stores a level-1 error configuration, $\mu$ stores a set of walks, $\lambda$ stores a set of primitive level-0 errors, and $\mu_{\mathrm{m}}$ and $\mu_{\mathrm{l}}$ store sets of closed walks on cores. For each closed walk $\ell_{i}$ in $w|_{\mathsf{C}}$, registers $\rho_i$, $\mathsf{g}_i$, and $\gamma_i$ store a set of walks, a sequence of primitive level-1 errors, and a sequence of primitive level-0 errors, respectively. All these registers are initialized to be empty.}   
       \small \justifying
       Denote the set of all simple walks propagated from closed core walks in $w$ as $w^{\bullet}$. Then, $w^{\bullet}\cup w|_{\mathsf{B}_{\blacktriangledown}}$ is a disjoint union of closed walks $\{l_{b_1},\cdots,l_{b_u}\}$ on buses. Let $\mathsf{w}_{\mathsf{b}}$ be the level-1 error configuration induced by these closed walks.
       
        \For{every level-1 bus error $\mathsf{e}_{b}\in\mathsf{w}_{\mathsf{b}}$}{
        \nonl Let $\sigma$ be the bus acted on by $\mathsf{e}_{b}$. Find a lightning walk $l_b\in\{l_{b_1},\cdots,l_{b_u}\}$ triggering $\mathsf{e}_{b}$.
        
        \nonl
        \If{there is a source $\ell_i$ of a simple walk in $l_b\cap w^{\bullet}$, such that the bus $\sigma$ is not on the target side of the hybrid-unit CNOT gate corresponding to the first CNOT membrane intersected by $\ell_i$} {\nonl add $\mathsf{e}_{b}$ to $\mathsf{w}$. Find a stabilizer measurement error $e_{b}\in[\ell_i]$ on the CNOT membrane, through which $\ell_i$ induces simple walks in $\ell_{i}^{\bullet}\cap l_{b}$. Append $\mathsf{e}_{b}$ to $\mathsf{g}_{i}$ and $e_b$ to $\gamma_i$.} 
        \nonl
        \ElseIf{there is a lightning walk $\ell_i$ as a source of a simple walk in $l_b\cap w^{\bullet}$}{\nonl Let $\mathsf{e}_c$ denote the primitive level-1 core error on the canonical segment immediately before the first CNOT membrane intersected by $\ell_i$. Add $\mathsf{e}_{c}$ to $\mathsf{w}$. Find a qubit error $e_{c}\in[\ell_{i}]$ before the first CNOT membrane intersected by $\ell_i$. Append $\mathsf{e}_{c}$ to $\mathsf{g}_{i}$ and $e_{c}$ to $\gamma_{i}$.}
        \nonl
        \Else{\nonl Pick one source $\ell_i$ of $l_b$. Add $\mathsf{e}_{b}$ to $\mathsf{w}$. Find a stabilizer measurement error $e_{b}\in\ell_i$ on the first CNOT membrane intersected by $\ell_i$. For every other source $\ell_{j}$ of $l_b$, update $\rho_i$ to $\rho_i\cup\ell_j$. Update $\rho_i$ to $\rho_i\cup(l_b\cap w|_{\mathsf{B}_{\blacktriangledown}})$. Append $\mathsf{e}_{b}$ to $\mathsf{g}_i$ and $e_{b}$ to $\gamma_i$.
       }}

        For each lightning walk $\ell_i$ in $w$, obtain its canonical retraction $g_i$ according to Lemma~\ref{lemma: lightning error decomposition}. Let $\mathsf{e}_{c}$ be the primitive level-1 core error induced by $g_i$. If $\ell_i$ is canonical, select a qubit error $e_c\in[\ell_i]$. Otherwise, find a qubit error $e_{c}\in[\ell_{i}]$ before the first CNOT membrane intersected by $\ell_i$. If $\mathsf{g}_i$ contains no level-1 core error, then add $\mathsf{e}_{c}$ to $\mathsf{w}$; append $\mathsf{e}_{c}$ to $\mathsf{g}_{i}$ and $e_{c}$ to $\gamma_{i}$. 
   \end{algorithm}

\addtocounter{algocf}{-1} 
\begin{algorithm}
  \setcounter{AlgoLine}{3}
  \caption{Subroutine for analyzing a downstream web (Part 2)}
\let\oldnl\nl
\newcommand{\nonl}{\renewcommand{\nl}{\let\nl\oldnl}}
  \small\justifying
  Create a buffer register $\mathsf{w}'$ as a copy of $\mathsf{w}$. \For{every closed walk $\ell_i$} {
    \nonl If $\mathsf{g}_{i}$ does not contain any primitive level-1 error in $\mathsf{w}'$, then reset $\mathsf{g}_i$ and $\gamma_{i}$ to empty and continue. Otherwise, update $\mu$ to $\mu\cup \ell_i\cup \rho_i$. If $\ell_i$ is a lightning walk, then update $\mu_{\mathrm{l}}$ to $\mu_{\mathrm{l}}\cup\ell_{i}$. If $\ell_{i}$ is a cross-membrane walk, update $\mu_{\mathrm{m}}$ to $\mu_{\mathrm{m}}\cup\ell_{i}\cup(\rho_{i}\cap w|_{\mathsf{C}})$.
    
    \nonl Construct two empty buffer registers $\mathsf{g}'_{i}$ and $\gamma'_{i}$ corresponding to $\mathsf{g}_{i}$ and $\gamma_{i}$, respectively.

    \nonl\For{every primitive level-1 error $\mathsf{e}$ in $\mathsf{w}'$ that is also in $\mathsf{g}_i$ as the $j$-th member}{
    \nonl Add the $j$-th member of $\gamma_i$ to $\lambda$. Remove $\mathsf{e}$ from $\mathsf{w}'$. Append the $j$-th member of $\mathsf{g}_{i}$ to $\mathsf{g}_{i}'$; append the $j$-th member of $\gamma_{i}$ to $\gamma_{i}'$.
    }

    \nonl Reset $\mathsf{g}_i$ and $\gamma_i$ to $\mathsf{g}'_{i}$ and $\gamma_{i}'$, respectively. Since $\mathsf{g}_{i}$ has the same length as $\gamma_{i}$, we regard every $j$-th member of $\mathsf{g}_i$ as corresponding to the $j$-th member of $\gamma_i$.    
    }
\end{algorithm}

We now bound the weight of physical errors in $\mu$ to prove condition (iii). We first note that every closed walk $\ell_{i}$ (on a core) in $w$ is a simple walk; the corresponding register $\mathsf{g}_i$ is non-empty if and only if $\ell_i$ is added to $\mu$ during the fourth step of Algorithm~\ref{Alg: subroutine for analyzing a downstream fault web}. For each lightning walk $\ell_i$ in $w$ on a core with a non-empty register $\mathsf{g}_{i}$, the number of primitive bus errors in $\mathsf{g}_i$, if non-zero, is strictly smaller than the number of membranes intersected by $\ell_i$. Moreover, $\mathsf{g}_i$ contains at most one primitive core error. Thus, the weight of stabilizer measurement errors in $\ell_i$, $|\ell_{i,\perp}|$, is lower bounded by $d_0(|\mathsf{g}_i|-1)$; and the weight of qubit errors in $\ell_i$, $|\ell_{i,||}|$, is lower bounded by $d_0$. Thus, the weight of $\ell_{i}$, $|\ell_{i}|$, is lower bounded by $d_0|\mathsf{g}_i|$. 
Let $\mathsf{w}_{\mathrm{lc}}$ be the sum of the register $\mathsf{g}_i$ for every lightning core walk $\ell_i$ in $w$. Thus, we can bound the weight of $\mu_{\mathrm{l}}$ (contributed by lightning walks on cores in step 4 of Algorithm~\ref{Alg: subroutine for analyzing a downstream fault web}) as follows:
\begin{equation}\label{eq: bound on lightning walk error}
    d_{0}|\mathsf{w}_{\mathrm{lc}}|\leq |\mu_{\mathrm{l}}|.
\end{equation}

Consider a primitive level-1 error $\mathsf{e}_{b}$ in $\mathsf{w}\backslash\mathsf{w}_{\mathrm{lc}}$. We know that $\mathsf{e}_{b}$ is a bus error contained in the level-1 error register $\mathsf{g}_{j}$ of some cross-membrane walk $\ell_j$ in $w$. We say $\mathsf{e}_{b}$ is a leading bus error if the $Z$ bus acted on by $\mathsf{e}_{b}$ is on the target side of the core-bus CNOT gate corresponding to the first CNOT membrane intersected by $\ell_{j}$. Otherwise, we say $\mathsf{e}_b$ is a trailing bus error. If the level-1 error register $\mathsf{g}_{j}$ of $\ell_j$ contains no leading bus error, then the weight of stabilizer measurement errors in $\ell_j$, $|\ell_{j,\perp}|$, is lower bounded by $d_0|\mathsf{g}_j|$. On the other hand, if $\mathsf{g}_j$ contains a leading bus error $\mathsf{e}_{b}$, then $|\ell_{j,\perp}|$ is lower bounded by $d_0(|\mathsf{g}_{j}|-1)$. According to Algorithm~\ref{Alg: subroutine for analyzing a downstream fault web}, we have picked a lightning walk $l_{b}$ on a $Z$ bus that induces $\mathsf{e}_{b}$. Let $\xi_{b}\subset\{\ell_1,\cdots,\ell_{o}\}$ be the collection of all sources of simple walks in $l_{b}\cap w^{\bullet}$. We say $\xi_b$ is the source set corresponding to $\mathsf{e}_{b}$. Then, Algorithm~\ref{Alg: subroutine for analyzing a downstream fault web} guarantees that $\mu$ contains $\xi_{b}$ and $l_{b}\cap w|_{\mathsf{B}_{\blacktriangledown}}$. According to Lemma~\ref{lemma: propagation of cross-membrane cycle}, the weight of simple walks, $|l_{b}\cap w^{\bullet}|$, is upper bounded by $|\xi_{b,||}|$/2. MWPM decoding ensures that $|l_{b}\cap w_{\kappa}|\leq |l_{b}\cap w_{\epsilon}|+|l_{b}\cap w^{\bullet}|$. From these two inequalities, we obtain the following bounds:
\begin{align}\label{eq: single error bound on qubit errors in cm core walks}
    |(l_{b}\cap w|_{\mathsf{B}_{\blacktriangledown}})|\leq& 2|l_b\cap w_{\epsilon}| + |\xi_{b,||}|/2, \nonumber\\
    d_0\leq |l_{b}|\leq& 2|l_b\cap w_{\epsilon}| + |\xi_{b,||}|.
\end{align}
Let $\mathsf{w}_{\mathrm{tb}}$ and $\mathsf{w}_{\mathrm{lb}}$ be the collection of trailing bus errors and the collection of leading bus errors in $\mathsf{w}\backslash\mathsf{w}_{\mathrm{lc}}$, respectively. Then, the total weight of stabilizer measurement errors in $\mu_{\mathrm{m}}$, $|\mu_{\mathrm{m},\perp}|$ has a lower bound: 
\begin{equation}\label{eq: lower bound stab error in cm walks}
    d_{0}|\mathsf{w}_{\mathrm{tb}}|\leq |\mu_{\mathrm{m},\perp}|.
\end{equation}
As the source sets corresponding to primitive errors in $\mathsf{w}_{\mathrm{lb}}$ are mutually disjoint, we can extend Eq.~\ref{eq: single error bound on qubit errors in cm core walks} to obtain the following inequalities.
\begin{align}\label{eq: bound on qubit errors in all cm core walks}
    |\mu\cap w|_{\mathsf{B}_{\blacktriangledown}}| \leq& 2|(\mu\cap w_{\epsilon}|_{\mathsf{B}_{\blacktriangledown}})| + |\mu_{\mathrm{m},||}|/2 ,\nonumber\\
    d_0|\mathsf{w}_{\mathrm{lb}}|\leq& 2|(\mu\cap w_{\epsilon}|_{\mathsf{B}_{\blacktriangledown}})| + |\mu_{\mathrm{m},||}|.
\end{align}
We now bound the weight of physical errors in $\mu$ in terms of the weight of level-1 errors $|\mathsf{w}|$.
\begin{align}
    |\mu\cap w_{\epsilon}|=&|\mu_{\mathrm{l}}\cap w_{\epsilon}|+ |\mu_{\mathrm{m}}\cap w_{\epsilon}| + |(\mu|_{\mathsf{B}_{\blacktriangledown}}\cap w_{\epsilon})|\nonumber\\
    \geq& \frac{1}{2}(|\mu_{\mathrm{l}}|+|\mu_{\mathrm{m},\perp}|+|\mu_{\mathrm{m},||}|+2|(\mu|_{\mathsf{B}_{\blacktriangledown}}\cap w_{\epsilon})|)\nonumber\\
    \geq& \frac{d_0}{2}(|\mathsf{w}_{\mathrm{lc}}|+|\mathsf{w}_{\mathrm{tb}}|+|\mathsf{w}_{\mathrm{lb}}|)=d_{0}|\mathsf{w}|/2,
\end{align}
where the first inequality is implied by MWPM decoding; the second inequality is derived from Eq.~\ref{eq: bound on lightning walk error}, Eq.~\ref{eq: lower bound stab error in cm walks}, and Eq.~\ref{eq: bound on qubit errors in all cm core walks}.
Similarly, we can bound $|\mu\cap w_{\epsilon}|$ by $|\mu|$ as follows.
\begin{equation}
    |\mu\cap w_{\epsilon}| \geq \frac{1}{2}(|\mu_{\mathrm{l}}|+|\mu_{\mathrm{m},\perp}|+|\mu_{\mathrm{m},||}|/2+|(\mu|_{\mathsf{B}_{\blacktriangledown}})|)\geq \frac{1}{4}|\mu|.
\end{equation}
\end{proof}
Building on Lemma~\ref{lemma: bounds on a downstream web}, we can now bound $f$.
\begin{lemma}
    [Bounds on a residual level-0 error configuration]\label{lemma: bounds on an error config}
    Consider the level-0 residual error configuration $f$ (consisting only of $X$ errors). Suppose $f|_{\mathsf{B}_{\lozenge}}$ contains no cycle of weight at least $d_{0}d_{1}$. Denote the collection of physical errors in $f$ as $f_{\epsilon}$ and the collection of inferred errors in $f$ as $f_{\kappa}$. We can explicitly construct a level-1 error $\mathsf{f}$, such that $f$ induces the level-1 error $\mathsf{f}$. For each level-1 error subconfiguration $\mathsf{h}:=\mathsf{e}_{1}+\cdots+\mathsf{e}_{|\mathsf{h}|}\subset \mathsf{f}$ (each $\mathsf{e}_{i}$ is a distinct primitive level-1 error), we can find two level-0 error configurations $\lambda_{\mathsf{h}}\subset\mu_{\mathsf{h}}\subset f$, such that (i) $\lambda_{\mathsf{h}}=e_1+\cdots+e_{|\mathsf{h}|}$ with each $e_{i}$ a distinct primitive level-0 error in the receptive zone of $\mathsf{e}_{i}$, (ii) each connected component of $\mu_{\mathsf{h}}$ (on the level-0 adjacency graph) contains at least one element in $\lambda_{\mathsf{h}}$, and (iii) the weight of physical errors in $\mu_{\mathsf{h}}$, $|\mu_{\mathsf{h}}\cap f_{\epsilon}|$, is lower bounded by $\max(|\mu_{\mathsf{h}}|/4,d_0|\mathsf{h}|/2)$. 
\end{lemma}
\begin{proof}

\begin{algorithm}
        \let\oldnl\nl
        \newcommand{\nonl}{\renewcommand{\nl}{\let\nl\oldnl}}
       \caption{Subroutine for analyzing a residual level-0 error configuration}\label{Alg: subroutine for analyzing f}
       \KwData{\justifying \small A residual level-0 error configuration $f$. Registers $\mathsf{f}$, $\mu_{\mathsf{h}}$, $\lambda_{\mathsf{h}}$, $\mu_{\mathrm{l}}$, and $\mu_{\mathrm{m}}$, where $\mathsf{f}$ stores a level-1 error configuration, $\mu_{\mathsf{h}}$ stores a level-0 error configuration, $\lambda_{\mathsf{h}}$ stores a set of primitive level-0 errors, and $\mu_{\mathrm{l}}$ and $\mu_{\mathrm{m}}$ store sets of walks on cores. All these registers are initialized to be empty.}   
       \small \justifying
        We know that $f^{\circ}\cup f|_{\mathsf{C}}\cup f|_{\mathsf{B}_{\blacktriangledown}}$ is a disjoint union of downstream webs $\{w_1,\cdots,w_{o}\}$ and native walks $\{\nu_1,\cdots,\nu_r\}$ on $Z$ buses. Run Algorithm~\ref{Alg: subroutine for analyzing a downstream fault web} on each downstream web $w_i$; let $\mathsf{w}_i$ be the corresponding register that stores a level-1 error configuration. Add every $\mathsf{w}_{i}$ with $i\in\{1,\cdots,o\}$ to $\mathsf{f}$. Add level-1 bus errors induced by the native walks $\{\nu_{1},\cdots,\nu_{r}\}$ to $\mathsf{f}$.

        \textbf{Input} a subconfiguration $\mathsf{h}\subset \mathsf{f}$. Create a buffer register $\mathsf{h}'$ as a copy of $\mathsf{h}$.

        \For{each native walk $\nu_i$ inducing a primitive level-1 bus error $\mathsf{e}_{b}$} {
        \nonl If $\mathsf{e}_{b}$ is not in $\mathsf{h}'$, then continue. Otherwise, update $\mu_{\mathsf{h}}$ to $\mu_{\mathsf{h}}\cup \nu_i$; find a qubit error $e$ in $\nu_{i}$ and add $e$ to $\lambda_{\mathsf{h}}$.
        
        \nonl Remove $\mathsf{e}_{b}$ from $\mathsf{h}'$.
        }
       
        \For{each closed walk $\ell_j$ on a core in every downstream web $w_{i}$}{
        \nonl According to Algorithm~\ref{Alg: subroutine for analyzing a downstream fault web}, there are three registers $\mathsf{g}_{j}$, $\rho_j$, and $\gamma_{j}$ associated with $\ell_j$. If the register $\mathsf{g}_j$ contains no primitive level-1 error in $\mathsf{h}'$, then continue.  

        \nonl Update $\mu_{\mathsf{h}}$ to $\mu_{\mathsf{h}}\cup (\ell_j\cap {f}|_{\mathsf{C}})\cup(\rho_{j}\cap({f}|_{\mathsf{C}}\cup{f}|_{\mathsf{B}_{\blacktriangledown}}))$. Let $\tilde{\xi}_j\subset f|_{\mathsf{B}_{\lozenge}}$ be the set of all sources for simple walks in $(\ell_{j}\cup \rho_j)\cap {f}^{\circ}$. Update $\mu_{\mathsf{h}}$ to $\mu_{\mathsf{h}}\cup\tilde{\xi}_j$. 
        
        \nonl If $\ell_j$ is a lightning walk, update $\mu_{\mathrm{l}}$ to $\mu_{\mathrm{l}}\cup\ell_j$. If $\ell_j$ is a cross-membrane walk, update $\mu_{\mathrm{m}}$ to $\mu_{\mathrm{m}}\cup\ell_j\cup(\rho_j\cap w_i|_{\mathsf{C}})$.

       \nonl \For{every primitive level-1 error $\mathsf{e}$ in $\mathsf{g}_j$ that is also in $\mathsf{h}'$}{
       \nonl Denote its corresponding primitive level-0 error in $\gamma_j$ as $e$. If $\mathsf{e}$ is (i) a bus error or (ii) a core error with $e\in\ell_j\cap f|_{\mathsf{C}}$, then add $e$ to ${\lambda}_{\mathsf{h}}$. Otherwise, if $e\notin \ell_j\cap{f}|_{\mathsf{C}}$, there exists a closed walk $\tilde{\ell}\in\tilde{\xi}_j$, such that $\tilde{\ell}$ is the source of a simple walk in $\ell_j$ containing $e$. Suppose $e$ is on the core $\sigma$; let $\Lambda_{t}$ denote the hybrid-unit CNOT gate through which $\tilde{\ell}$ propagates the simple walk containing $e$, where $t$ is the level-0 time step of the gate. In this case, pick a qubit error $\tilde{e}\notin\lambda_{\mathsf{h}}$ on $\tilde{\ell}$, such that $\tilde{e}$ is a qubit error that happens before $\Lambda_{t}$. Add $\tilde{e}$ to ${\lambda_{\mathsf{h}}}$.   

        \nonl  Remove $\mathsf{e}$ from $\mathsf{h}'$.
       }
}
   \end{algorithm}

We construct $\mathsf{f}$ in step 1 of Algorithm~\ref{Alg: subroutine for analyzing f}, then input a subconfiguration $\mathsf{h}\subset \mathsf{f}$ and construct $\mu_{\mathsf{h}}$ and $\lambda_{\mathsf{h}}$ by running the remaining steps of Algorithm~\ref{Alg: subroutine for analyzing f}. We now bound the weight of physical errors in $\mu_{\mathsf{h}}$. Following the definitions introduced in the proof of Lemma~\ref{lemma: bounds on a downstream web}, we let $\mathsf{h}_{\mathrm{lc}}$ denote the collection of level-1 errors in $\mathsf{h}$ contributed by lightning walks in $\mu_{\mathrm{l}}$; let $\mathsf{h}_{\mathrm{lb}}$ and $\mathsf{h}_{\mathrm{tb}}$ denote the collections of leading bus errors and trailing bus errors in $\mathsf{h}$, respectively; let $\mathsf{h}_{\mathrm{nb}}$ denote the collection of level-1 bus errors in $\mathsf{h}$ induced by native walks on $Z$ buses. According to Algorithm~\ref{Alg: subroutine for analyzing f}, $\mu_{\mathrm{l}}$ is a disjoint union of lightning walks on cores; $\mu_{\mathrm{m}}$ is a disjoint union of cross-membrane walks on cores. For errors in $\mu_{\mathsf{h}}|_{\mathsf{B}_{\blacktriangledown}}$, we have the following two bounds by the same argument as for Eq.~\ref{eq: bound on qubit errors in all cm core walks}. 
   \begin{align}\label{eq: bound f cm and downstream walks}
       |(\mu_{\mathsf{h}}|_{\mathsf{B}_{\blacktriangledown}})|\leq& 2|(\mu_{\mathsf{h}}|_{\mathsf{B}_{\blacktriangledown}}\cap f_{\epsilon})| + |\mu_{\mathrm{m},||}|/2, \nonumber \\
       d_0(|\mathsf{h}_{\mathrm{lb}}|+|\mathsf{h}_{\mathrm{nb}}|)\leq& 2|(\mu_{\mathsf{h}}|_{\mathsf{B}_{\blacktriangledown}}\cap f_{\epsilon})| + |\mu_{\mathrm{m},||}|.
   \end{align}
Following Eq.~\ref{eq: bound on lightning walk error} and Eq.~\ref{eq: lower bound stab error in cm walks}, we have similar bounds on $|\mu_{\mathrm{m},\perp}|$ and $|\mu_{\mathrm{l}}|$:
\begin{align}\label{eq: bound f cm walks with lv1 errors}
    d_0|\mathsf{h}_{\mathrm{tb}}|\leq& |\mu_{\mathrm{m},\perp}| \nonumber\\
    d_0|\mathsf{h}_{\mathrm{lc}}|\leq& |\mu_{\mathrm{l}}|
\end{align}
As we are eventually interested in bounding error weights for $\mu_{\mathsf{h}}$, we would like to relate bounds on $\mu_{\mathrm{l}}$ and $\mu_{\mathrm{m}}$ to bounds on $\mu_{\mathsf{h}}$. Since the total weight of propagated simple walks on cores from $\mu_{\mathsf{h}}$ is upper bounded by $|(\mu_\mathsf{h}|_{\mathsf{B}_{\lozenge}})|/2$, we have the following two inequalities:
\begin{align}
    |\mu_{\mathrm{l}}|+|\mu_{\mathrm{m}}|\leq& |(\mu_\mathsf{h}|_{\mathsf{C}})| + |(\mu_{\mathsf{h}}|_{\mathsf{B}_{\lozenge}})|/2,\nonumber \\
    |(\mu_{\mathsf{h}}|_{\mathsf{C}}\cap f_{\kappa})| \leq& |(\mu_{\mathsf{h}}|_{\mathsf{C}}\cap f_{\epsilon})| + |(\mu_{\mathsf{h}}|_{\mathsf{B}_{\lozenge}})|/2,
\end{align}
which can be reorganized into the following bounds:
\begin{align}\label{eq: bound f cm walks}
    |(\mu_{\mathsf{h}}|_{\mathsf{C}})|\leq& 2|(\mu_{\mathsf{h}}|_{\mathsf{C}}\cap f_{\epsilon})|+|(\mu_{\mathsf{h}}|_{\mathsf{B}_{\lozenge}})|/2,\nonumber\\
    |\mu_{\mathrm{l}}|+|\mu_{\mathrm{m}}|\leq& 2|(\mu_{\mathsf{h}}|_{\mathsf{C}}\cap f_{\epsilon})|+|(\mu_{\mathsf{h}}|_{\mathsf{B}_{\lozenge}})|.
\end{align}
Also, MWPM decoding on $\mathcal{G}_{\mathsf{B}_{\lozenge}}$ implies $|(\mu_{\mathsf{h}}|_{\mathsf{B}_{\lozenge}})|\leq 2|(\mu_{\mathsf{h}}|_{\mathsf{B}_{\lozenge}}\cap f_{\epsilon})|$. Thus, we can bound the weight of physical errors in $\mu_{\mathsf{h}}$ by the weight of the level-1 error configuration $\mathsf{h}$ as follows.
\begin{align}
    |\mu_{\mathsf{h}}\cap f_{\epsilon}|=&  |(\mu_{\mathsf{h}}|_{\mathsf{B}_{\lozenge}}\cap f_{\epsilon})| + |(\mu_{\mathsf{h}}|_{\mathsf{C}}\cap f_{\epsilon})| + |(\mu_{\mathsf{h}}|_{\mathsf{B}_{\blacktriangledown}}\cap f_{\epsilon})|\nonumber\\
    \geq& |(\mu_{\mathsf{h}}|_{\mathsf{B}_{\lozenge}})|/2 + |(\mu_{\mathsf{h}}|_{\mathsf{C}}\cap f_{\epsilon})| + |(\mu_{\mathsf{h}}|_{\mathsf{B}_{\blacktriangledown}}\cap f_{\epsilon})|\nonumber\\
    \geq& (|\mu_{\mathrm{l}}|+|\mu_{\mathrm{m}}|)/2+ |(\mu_{\mathsf{h}}|_{\mathsf{B}_{\blacktriangledown}}\cap f_{\epsilon})|\nonumber\\
    \geq&\frac{d_0}{2}(|\mathsf{h}_{\mathrm{lc}}|+|\mathsf{h}_{\mathrm{lb}}|+|\mathsf{h}_{\mathrm{tb}}|+|\mathsf{h}_{\mathrm{nb}}|) = d_0|\mathsf{h}|/2,
\end{align}
where we use Eq.~\ref{eq: bound f cm walks} for the second inequality, and Eq.~\ref{eq: bound f cm and downstream walks} and Eq.~\ref{eq: bound f cm walks with lv1 errors} for the last inequality. 
Similarly, we now bound $|\mu_{\mathsf{h}}\cap f_{\epsilon}|$ in terms of the weight of $\mu_{\mathsf{h}}$.
\begin{align}
    |\mu_{\mathsf{h}}\cap f_{\epsilon}|\geq& \big(|(\mu_{\mathsf{h}}|_{\mathsf{B}_{\lozenge}})|/4 + |(\mu_{\mathsf{h}}|_{\mathsf{C}})|/4+|\mu_{\mathrm{m}}|/8 \nonumber\\
    & + |(\mu_{\mathsf{h}}|_{\mathsf{B}_{\blacktriangledown}}\cap f_{\epsilon})|\nonumber\big) \\
    \geq& (|(\mu_{\mathsf{h}}|_{\mathsf{B}_{\lozenge}})|+|(\mu_{\mathsf{h}}|_{\mathsf{C}})|+|(\mu_{\mathsf{h}}|_{\mathsf{B}_{\blacktriangledown}})|)/4=|\mu_{\mathsf{h}}|/4.
\end{align}
\end{proof}

We previously considered only level-0 $X$ error configurations in this subsection. Now, we allow $f$ to be a general level-0 residual error configuration with $f=\epsilon+\kappa$, where $\epsilon\in\mathbb{Z}_2^{|\mathcal{E}_{\mathrm{dep}}|}$ is the physical error term and $\kappa\in\mathbb{Z}_2^{|\mathcal{E}_{0}|}$ is the inferred error configuration. Let $f_{\epsilon}:=f\cap \epsilon$ and $f_{\kappa}:=f\cap\kappa$, representing physical errors and inferred errors in $f$, respectively. Consider the $X$-basis projection $f_{X}=\epsilon_{X}+\kappa_{X}$. In previous arguments, $f_{X,\epsilon}:=f_{X}\cap\epsilon_{X}$ and $f_{X,\kappa}:=f_{X}\cap\kappa_{X}$ represent physical errors and inferred errors in $f_{X}$, respectively. We define similar notations for $f_Z$. To relate subconfigurations in $f_X\cup f_{Z}$ to subconfigurations in $f$, we construct a section map $\mathbb{S}_{f}$ that maps every subset (subconfiguration) in $f_{X}\cup f_{Z}$ to a subset (subconfiguration) in $f$ with the following four properties.
\begin{enumerate}
    \item For each subconfiguration $g\subset f_{X}\cup f_{Z}$, $\mathbb{S}_{f}(g)=\bigcup_{e\in g} \mathbb{S}_{f}(e)$ with each $e$ a primitive level-0 error in $g$.
    \item For each primitive error $e_{X}\in f_{X}$, $\mathbb{P}_{X}\circ \mathbb{S}_f(e_X)=e_X$. Similarly, for each primitive error $e_{Z}\in f_{Z}$, $\mathbb{P}_{Z}\circ \mathbb{S}_f(e_Z)=e_Z$.
    \item For every $g_X\subset f_{X,\epsilon}$, $\mathbb{S}_{f}(g_X)\subset f_{\epsilon}$. For every $g_X\subset f_{X,\kappa}$, $\mathbb{S}_f(g_X)\subset f_{\kappa}$.
    \item For every $g_{Z}\subset f_{Z,\epsilon}$, $\mathbb{S}_{f}(g_Z)\subset f_{\epsilon}$. For every $g_Z\subset f_{Z,\kappa}$, $\mathbb{S}_f(g_Z)\subset f_{\kappa}$.
\end{enumerate}
By construction, we know that given a subconfiguration $g\subset f_{X}\cup f_{Z}$, $|\mathbb{S}_f(g)\cap \epsilon|\geq \max(|g_X\cap \epsilon_X|,|g_{Z}\cap\epsilon_{Z}|)$. Additionally, if $g$ forms a connected cluster on the level-0 adjacency graph, then $\mathbb{S}_{f}(g)$ also forms a connected cluster. 
Building on Lemma~\ref{lemma: bounds on an error config}, we can bound a general level-0 residual error configuration.   

\begin{lemma}
    [Bounds on a general level-0 residual error configuration]\label{lemma: bounds on a general residual error config}
    Consider the general level-0 residual error configuration $f$. Suppose $f_X|_{\mathsf{B}_{\lozenge}}$ contains no cycle with weight at least $d_{0}d_{1}$. We can explicitly construct a level-1 $X$ error $\mathsf{f}_{\mathsf{X}}$, such that $f_X$ induces the level-1 error $\mathsf{f}_{\mathsf{X}}$. For any given level-1 error subconfiguration $\mathsf{h}:=\mathsf{e}_1+\cdots+\mathsf{e}_{|\mathsf{h}|}\subset \mathsf{f}_{\mathsf{X}}$ (each $\mathsf{e}_{i}$ is a distinct primitive level-1 $X$ error), we can find two general level-0 error configurations $\lambda_{\mathsf{h}}\subset\mu_{\mathsf{h}}\subset f$ such that (i) $\lambda_{\mathsf{h}}=e_1+\cdots+e_{|\mathsf{h}|}$ with each $e_{i}$ a distinct primitive level-0 error whose $X$-basis projection is in the receptive zone of $\mathsf{e}_{i}$, (ii) each connected component of $\mu_{\mathsf{h}}$ (on the level-0 adjacency graph) contains at least one element in $\lambda_{\mathsf{h}}$, and (iii) the weight of physical errors in $\mu_{\mathsf{h}}$, $|\mu_{\mathsf{h}}\cap \epsilon|$, is lower bounded by $\max(|\mu_{\mathsf{h}}|/4,d_0|\mathsf{h}|/2)$. 
\end{lemma}
\begin{proof}
    According to Lemma~\ref{lemma: bounds on an error config}, we construct $\mathsf{f}_{\mathsf{X}}$ from $f_X$ and for any level-1 subconfiguration $\mathsf{h}\subset \mathsf{f}_{\mathsf{X}}$ find $\lambda_{\mathsf{h},X}\subset\mu_{\mathsf{h},X}\subset f_{X}$ satisfying all three conditions in Lemma~\ref{lemma: bounds on an error config}. Let $\lambda_{\mathsf{h}}:=\mathbb{S}_{f}(\lambda_{\mathsf{h},X})$ and $\mu_{\mathsf{h}}:=\mathbb{S}_{f}(\mu_{\mathsf{h},X})$. Then, these two general error configurations satisfy all requirements in this lemma.   
\end{proof}

We follow previous fault-tolerance proofs~\cite{panos_threshold,gottesman_fault_tolerant_2014,he_composable_2025} (in particular, the `bad set' formalism in Ref.~\cite{he_composable_2025}) to provide a sufficient condition for upper bounding level-1 error probabilities. 
\begin{lemma}
    [Sufficient condition for upper bounding a level-1 error probability]\label{lemma: sufficient condition for bounding level 1 error}
    Consider a level-1 error $\mathsf{g}$. Suppose there exists a collection of general level-0 error configurations $\mathcal{F}_{\mathsf{g}}$, referred to as the \textit{key collection} of $\mathsf{g}$, that satisfies the following conditions
    \begin{enumerate}
        \item For each general residual level-0 error configuration $f$ (composed of two components, physical errors $f_{\epsilon}$ and decoded errors $f_{\kappa}$, respectively) that induces a level-1 error configuration containing $\mathsf{g}$, there exists $\mu\in\mathcal{F}_{\mathsf{g}}$ with $\mu\subset f$ such that the weight of the intersection of $\mu$ and the physical errors in $f$ is lower bounded by both $d_{0}|\mathsf{g}|/2$ and $|\mu|/C_1$, where $C_1>1$ is a constant independent of $f$. 
        \item The number of error configurations in $\mathcal{F}_{\mathsf{g}}$ with a weight $\mathrm{w}$ is upper bounded by $C_{2}\cdot C_{3}^{\mathrm{w}}$, where both $C_2>0$ and $C_{3}>1$ are constants. Every error configuration in $\mathcal{F}_{\mathsf{g}}$ has a weight at least $d_{0}|\mathsf{g}|/2$. 
    \end{enumerate}
    Moreover, suppose physical level-0 errors are local stochastic, such that the event of physical errors containing a general level-0 error configuration $\epsilon$ may occur with a probability upper bounded by $p^{|\epsilon|}$, with $p\leq 1/(4C_3)^{C_1}$. Then, the level-1 error $\mathsf{g}$ occurs with probability upper bounded by $C_A(p/p_{th})^{d_{0}|\mathsf{g}|/2}$, where $C_{A}:=3C_2C_3/(C_3-1)$ and $p_{th}:=(1/(2C_3)^{C_1})$.
\end{lemma}
\begin{proof}
Define event $\Upsilon_{\mathsf{g}}$ as the set of all general level-0 physical-error configurations for which the physical-error configuration and the corresponding inferred level-0 error configuration jointly induce a level-1 error configuration containing $\mathsf{g}$. Therefore, the probability that $\mathsf{g}$ occurs at level 1 is equal to the level-0 probability of $\Upsilon_{\mathsf{g}}$. 

From the key collection $\mathcal{F}_{\mathsf{g}}$, we define the set of key errors $\Xi:=\{\xi|\xi\subset \mu\ \text{for some}\  \mu\in\mathcal{F}_{\mathsf{g}}\ \text{with}\ |\xi|\geq d_0|\mathsf{g}|/2\ \text{and}\ |\xi|\geq |\mu|/C_1\}$. For a general level-0 error configuration $\xi$, define $\hat{\xi}:=\{f|\xi\subset f\}$ as the collection of all general level-0 error configurations containing $\xi$. For any $\epsilon\in\Upsilon_{\mathsf{g}}$, let $\kappa$ be the inferred level-0 error configuration from level-0 decoding with an input level-0 syndrome configuration $\partial_{0}\epsilon$. According to the first condition in the lemma, there exists $\xi\in\Xi$ such that $\xi\subset\epsilon$ (or equivalently, $\epsilon\in\hat{\xi}$). Thus, the event $\Upsilon_{\mathsf{g}}$ is contained in the event $\hat{\Xi}:=\bigcup_{\xi\in\Xi}\hat{\xi}=\bigcup_{\mu\in\mathcal{F}_{\mathsf{g}}}\left(\bigcup_{\xi\in\Xi,\ \xi\subset\mu}\hat{\xi}\right)$. Notice that for every $\mu\in\mathcal{F}_{\mathsf{g}}$, the probability
\begin{align}
    P\left(\bigcup_{\xi\in\Xi,\ \xi\subset\mu}\hat{\xi}\right)\leq& \sum_{\xi\in\Xi,\ \xi\subset\mu}p^{\max(d_0|\mathsf{g}|/2,|\mu|/C_1)} \nonumber\\
    \leq& (2^{C_1}p)^{\max(d_0|\mathsf{g}|/2,|\mu|/C_1)},
\end{align}
where the second inequality is based on the observation that the number of summations is upper bounded by $2^{|\mu|}$. Therefore, we can bound the level-1 error probability $P(\Upsilon_{\mathsf{g}})$ as follows: 
\begin{align}
    P(\Upsilon_{\mathsf{g}})\leq& P(\hat{\Xi})\leq \sum_{\mu\in\mathcal{F}_{\mathsf{g}}} (2^{C_1}p)^{\max(d_0|\mathsf{g}|/2,|\mu|/C_1)}\nonumber\\
    \leq& \sum_{\mathrm{w}=\lceil d_0|\mathsf{g}|/2\rceil}^{\infty} (2^{C_1}p)^{\max(d_0|\mathsf{g}|/2,\mathrm{w}/C_1)}\cdot C_2C_3^{\mathrm{w}}\nonumber\\
    \leq& \frac{3C_2C_3}{C_3-1}\left(\frac{p}{1/(2C_3)^{C_1}}\right)^{d_{0}|\mathsf{g}|/2}.
\end{align}
Here, the third inequality follows from the second condition in the lemma, and the last inequality uses $p\leq1/(4C_3)^{C_1}$. 
\end{proof}
We will need the following lemma from Ref.~\cite{panos_threshold} (widely used in fault-tolerance proofs to count the number of errors) to upper bound the number of elements of a certain weight in a set of error configurations.  
\begin{lemma}
    [Counting lemma~\cite{panos_threshold}]\label{lemma: counting lemma} Consider a graph $G$ in which every vertex has a degree of at most $r$. Given a non-empty set $\lambda$ of vertices in $G$, let $\Sigma_{m}(\lambda)$ be the collection of all sets of vertices satisfying the following conditions. 
    \begin{itemize}
        \item The set contains exactly $m$ vertices. 
        \item The set contains $\lambda$ as a subset. 
        \item Every connected cluster of the set (on $G$) contains at least one element in $\lambda$. 
    \end{itemize}
Then, the number of sets in $\Sigma_{m}(\lambda)$ is upper bounded by $\mathfrak{e}^{|\lambda|-1}(r\mathfrak{e})^{m-|\lambda|}$, where $\mathfrak{e}$ is Euler's number.  
\end{lemma}
We note that for our level-0 adjacency graph $\mathcal{A}$, the degree of any vertex is upper bounded by $48$, since each primitive level-0 error triggers up to $4$ level-0 detectors, and each level-0 detector can be triggered by up to $12$ distinct level-0 error locations.  

\begin{theorem}
    [Approximate error reduction theorem]\label{thm: approximate error reduction}
    Consider a level-0 circuit composed of level-0 SE rounds on cores and shuttle buses and $X$- and $Z$-basis readout gadgets. Every core has a distance of $d_{0}$; every shuttle bus has a width of $d_{0}$ and a length of $d_{0}d_{1}$. Compilation parameter $\alpha_{\mathsf{b}}$ is set as $d_1$. We further assume perfect time boundaries for cores for simplicity. Suppose we are only interested in the results of a subset $\mathcal{M}_1$ of readout gadgets, such that the separation between every two $X$-basis (or $Z$-basis) readout gadgets in $\mathcal{M}_1$ is at least $d_{0}$. We construct level-1 error locations for $\mathcal{M}_1$ according to the procedure in Sec.~\ref{subsec: detectors and decoding}. Let $\tau_{\mathsf{c}}$ denote the maximum length for a canonical segment on cores; let $\tau_{\mathsf{b}}$ denote the maximum lifetime for a shuttle bus. Let $|\mathsf{b}|$ be the number of shuttle buses used in the circuit. Then, under the phenomenological depolarizing noise model with an error rate $p<1/(4r\mathfrak{e})^{4}$ at level 0, except for a rare event with a probability upper bounded by $\mathrm{v}_{\mathsf{b}}(p/\sqrt{p_{th}})^{d_0d_1/2}$, the induced level-1 $X$ (or $Z$) errors are local stochastic with a level-1 error rate upper bounded by $z_{\mathrm{max}}(p/p_{th})^{d_{0}/2}$. Here, $\mathrm{v}_{\mathsf{b}}:=\tau_{\mathsf{b}}d_{0}^2d_{1}|\mathsf{b}|$, $p_{th}:=1/(2r\mathfrak{e})^4$, $z_{\mathrm{max}}:=16\max(\tau_{\mathsf{c}},\tau_{\mathsf{b}}) d_{0}^3d_{1}^2$, and $r=48$, which upper bounds the vertex degrees of the level-0 adjacency graph.  
\end{theorem}
\begin{proof}
Without loss of generality, we focus on level-1 $X$ errors in the proof. First, consider the event $\mathcal{B}_{d_1}$ that a general level-0 residual error configuration induces an $X$-basis cycle on an $X$ bus of weight at least $d_{0}d_{1}$. Let $\mathcal{H}_{\mathrm{w}}$ be the set of all weight-$\mathrm{w}$ general level-0 error configurations whose $X$-basis restriction is a weight-$\mathrm{w}$ cycle on an $X$ bus. Then, according to Lemma~\ref{lemma: counting lemma}, $|\mathcal{H}_{\mathrm{w}}|$ is upper bounded by $(4d_{0}^2d_{1}\cdot\tau_{\mathsf{b}}\cdot|\mathsf{b}|)\cdot (r\mathfrak{e})^{\mathrm{w}-1}$, where the first factor upper bounds the number of level-0 error locations on $X$ buses and the second factor upper bounds the number of weight-$\mathrm{w}$ clusters on the level-0 adjacency graph $\mathcal{A}$ containing a fixed level-0 error location on an $X$ bus. The probability of inducing any given general level-0 error configuration in $\mathcal{H}_{\mathrm{w}}$ is upper bounded by $p^{\mathrm{w}/2}2^{\mathrm{w}}$. Thus, the probability of the event $\mathcal{B}_{d_1}$ is upper bounded as
\begin{align}
  P(\mathcal{B}_{d_1})\leq& \frac{4d_0^2d_1\cdot\tau_{\mathsf{b}}\cdot|\mathsf{b}|}{r\mathfrak{e}} \sum_{\mathrm{w}=d_{0}d_{1}}^{\infty}(2r\mathfrak{e} \sqrt{p})^{\mathrm{w}} \nonumber \\
  =& \frac{4d_0^2d_1\cdot \tau_{\mathsf{b}}\cdot|\mathsf{b}|}{r\mathfrak{e}(1-2r\mathfrak{e}\sqrt{p})}\left((2r\mathfrak{e})^2p\right)^{d_{0}d_{1}/2}\nonumber\\
  \leq& \mathrm{v}_{\mathsf{b}}\left(p/\sqrt{p_{th}}\right)^{d_{0}d_{1}/2}.
\end{align}
We exclude this event in the following. 
  
For a level-1 $X$ error configuration $\mathsf{g}$, define $\mathcal{F}_{\mathsf{g}}$ as a collection of all general level-0 error configurations satisfying the following conditions: 
\begin{itemize}
    \item The general level-0 error configuration has a weight of at least $d_0|\mathsf{g}|$/2.
    \item The general level-0 error configuration contains a weight-$|\mathsf{g}|$ subconfiguration $\lambda$, such that each primitive level-1 error $\mathsf{e}$ in $\mathsf{g}$ corresponds to a distinct primitive level-0 error in $\lambda_X$ in the receptive zone of $\mathsf{e}$.  
    \item Each connected component (on the level-0 adjacency graph) of the general level-0 error configuration contains at least one element in $\lambda$. 
\end{itemize}
The size of the receptive zone (the number of level-0 error locations whose $X$-basis restriction is contained in the zone) of a primitive level-1 core error is upper bounded by $2\tau_{\mathsf{c}}\cdot 4d_{0}^2(1+d_{0}d_{1}^2)\leq 16\tau_{\mathsf{c}}d_{0}^3 d_{1}^{2}$. The size of the receptive zone of a primitive level-1 bus error is upper bounded by $4\tau_{\mathsf{b}}d_{0}^{2}d_{1}+\tau_{\mathsf{b}}d_{0}^2\leq 8\tau_{\mathsf{b}}d_{0}^2d_{1}$. Then, we can simply upper bound the size of any receptive zone by $z_{\mathrm{max}}$. Denote the number of weight-$\mathrm{w}$ configurations in $\mathcal{F}_{\mathsf{g}}$ as $C_{\mathsf{g}}(\mathrm{w})$. According to Lemma~\ref{lemma: counting lemma}, $C_{\mathsf{g}}(\mathrm{w})\leq (z_{\mathrm{max}}/r)^{|\mathsf{g}|}(r\mathfrak{e})^{\mathrm{w}}$. Combining this bound with Lemma~\ref{lemma: bounds on a general residual error config}, we can see that $\mathcal{F}_{\mathsf{g}}$ is a key collection of $\mathsf{g}$ in Lemma~\ref{lemma: sufficient condition for bounding level 1 error}. Then, the probability of inducing $\mathsf{g}$, $P(\mathsf{g})$, is upper bounded as follows.
\begin{equation}
    P(\mathsf{g})\leq z_{\mathrm{max}}^{|\mathsf{g}|} \cdot \left(p/p_{th}\right)^{d_{0}|\mathsf{g}|/2}
\end{equation}
Thus, the distribution of induced level-1 errors is local stochastic with a level-1 error probability upper bounded by $z_{\mathrm{max}} \left(p/p_{th}\right)^{d_{0}/2}$.

\end{proof}

\section{Logical Pauli measurements}\label{sec: logical pauli measurements}
In this section, we analyze the performance of the LMSs constructed for measuring logical $X$ and $Z$ operators on an HLP. We then describe how to measure general logical Pauli operators with $H$-transformed and $HS$-transformed readout gadgets. We extend our analysis of the induced level-1 error model as well as LMSs to incorporate $H$-transformed gadgets. Finally, we extend hybrid-unit CNOT gates to enable transversal CNOT gates between shuttle buses and external cores. We tentatively propose a hybrid architecture where an HLP or possibly multiple HLPs work closely with external cores to run logical computation. 

\subsection{Performance analysis of LMSs for measuring logical $X$ and $Z$ operators}
We first consider a single LMS consisting of $d_{1}$ $X$-basis (or $Z$-basis) readout gadgets. Let $\mathcal{M}$ be the collection of all logical readout gadgets in the LMS and all readout gadgets for level-1 SE. We construct a level-1 error model for $\mathcal{M}$. We can lower bound the number of primitive level-1 errors required to trigger a logical measurement error in this LMS.   
\begin{lemma}[Lower bound on the level-1 distance for an LMS]\label{lemma: level 1 distance for LMS}
    Given the level-1 error model for $\mathcal{M}$, for a level-1 error configuration $\mathsf{f}$ with $\overline{\partial}_{1}\mathsf{f}=0$, if $|\mathsf{f}|$ is smaller than $d_1$, then $\mathsf{f}$ does not induce a logical measurement error on the LMS.   
\end{lemma}
\begin{proof}
    Define an adjacency graph $\mathcal{A}_{1}$, with a vertex set of all primitive level-1 errors, such that two vertices are connected by an edge if and only if the corresponding level-1 errors have overlapping level-1 syndromes. Without loss of generality, we assume $\mathsf{f}$ is a connected cluster on $\mathcal{A}_1$. Let $\mathsf{D}_{\mathsf{S}}$ be the set of all level-1 detectors generated by level-1 SE. Let $\overline{\partial}_{1\mathsf{S}}$ be the projection of the level-1 syndrome map onto $\mathsf{D}_{\mathsf{S}}$. Define another adjacency graph $\mathcal{A}_{\mathsf{S}}$, with a vertex set of all primitive level-1 errors on cores and buses for level-1 SE, such that two vertices $\mathsf{e}_{1}$ and $\mathsf{e}_{2}$ (corresponding to two primitive level-1 errors) are connected by an edge if and only if $\overline{\partial}_{1\mathsf{S}}\mathsf{e}_1\cap\overline{\partial}_{1\mathsf{S}}\mathsf{e}_{2}\neq \emptyset$. The level-1 error configuration $\mathsf{f}$ is composed of two components: a configuration $\mathsf{f}_{\mathsf{S}}$ of level-1 errors on cores and buses for level-1 SE and a configuration $\mathsf{f}_{\mathsf{R}}$ of level-1 errors on buses in the LMS. We can see that $\overline{\partial}_{1\mathsf{S}}\mathsf{f}_{\mathsf{S}}=\overline{\partial}_{1\mathsf{S}}\mathsf{f}=0$. Thus, we can decompose $\mathsf{f}_{\mathsf{S}}$ into a series of level-1 error configurations $\mathsf{f}_{\mathsf{S},1},\cdots,\mathsf{f}_{\mathsf{S},j}$, such that each $\mathsf{f}_{\mathsf{S},i}$ ($i\in\{1,\cdots,j\}$) is a connected cluster on $\mathcal{A}_{\mathsf{S}}$ with $\overline{\partial}_{1\mathsf{S}}\mathsf{f}_{\mathsf{S},i}=0$. 

    For each $\mathsf{f}_{\mathsf{S},i}$ ($i\in\{1,\cdots,j\}$), denote the smallest and the largest level-1 time coordinates of detectors triggered by primitive errors in $\mathsf{f}_{\mathsf{S},i}$ as $\mathsf{t}_{i,0}$ and $\mathsf{t}_{i,1}$, respectively. We know that for every integer $\mathsf{t}$ with $\mathsf{t}_{i,0}\leq \mathsf{t}\leq \mathsf{t}_{i,1}-1$, there is at least one primitive error in $\mathsf{f}_{\mathsf{S},i}$ such that the level-1 detectors triggered by the error in $\mathsf{D}_{\mathsf{S}}$ cover both time coordinates $\mathsf{t}$ and $\mathsf{t}+1$. Thus, the weight of primitive errors (in $\mathsf{f}_{\mathsf{S},i}$) whose syndromes have two different time coordinates is lower bounded by $\mathsf{t}_{i,1}-\mathsf{t}_{i,0}$. We then back propagate all primitive level-1 core errors in $\mathsf{f}_{\mathsf{S},i}$ back to the beginning of the level-1 SE round $\mathsf{t}_{i,0}$ and denote the resulting level-1 error configuration as $\mathsf{h}_{i}$ (equivalent to $\mathsf{f}_{\mathsf{S},i}$). We can see that all core errors in $\mathsf{h}_{i}$ (at the beginning of the level-1 SE round $\mathsf{t}_{i,0}$) form a level-1 stabilizer of the level-1 code. Moreover, $\mathsf{h}_{i}$ should contain no bus errors for level-1 SE but may contain bus errors for logical readout gadgets in level-1 SE rounds from $\mathsf{t}_{i,0}-1$ to $\mathsf{t}_{i,1}$ (at most $\mathsf{t}_{i,1}-\mathsf{t}_{i,0}+2$ errors). Let $\mathsf{h}_i|_{\mathsf{C}}$ and $\mathsf{h}_{i}|_{\mathsf{B}}$ denote the subconfigurations of core errors and bus errors in $\mathsf{h}_{i}$, respectively. If $\mathsf{h}_{i}|_{\mathsf{B}}$ contains a bus error for the logical readout gadget at level-1 SE round $\mathsf{t}_{i,1}$ (or $\mathsf{t}_{i,0}-1$), then $\mathsf{f}_{\mathsf{S},i}$ contains at least one core error only triggering detectors in $\mathsf{D}_{\mathsf{S}}$ with a time coordinate $\mathsf{t}_{i,1}$ (or $\mathsf{t}_{i,0}$). Thus, we can see that $|\mathsf{h}_{i}|_{\mathsf{B}}|\leq |\mathsf{f}_{\mathsf{S},i}|$. Then, the level-1 error configuration $\mathsf{h}_{\mathsf{R}}:=\mathsf{f}_{\mathsf{R}}+\sum_{i=1}^{j}\mathsf{h}_{i}|_{\mathsf{B}}$, as an equivalent configuration to $\mathsf{f}$, is composed of only bus errors for readout gadgets and has a weight upper bounded by $|\mathsf{f}|<d_1$. Thus, the level-1 error configuration $\mathsf{f}$ cannot induce a logical measurement error.  
\end{proof}

Now, we show that the probability of any LMS being faulty remains exponentially suppressed even when LMSs are densely packed. This theoretical guarantee enables parallel logical $X$ or $Z$ measurements. We first introduce definitions and notations to describe the level-1 structure of an HLP. 
Let $n_1$ and $s_{1}$ be the number of level-1 data qubits and the number of level-1 stabilizers measured in each level-1 SE round. Let $\mathcal{M}_{\mathsf{SE}}$ denote the set of readout gadgets for level-1 SE. Define the degree of a level-1 data qubit as the number of level-1 stabilizers supported on this qubit that are measured in a level-1 SE round. Let $\delta_1$ be the maximum degree of a level-1 data qubit in our HLP; let $\varpi_1$ be the maximum weight of measured level-1 stabilizers.

\begin{theorem}
    [Fault probability for LMSs]\label{thm: fault probability for LMSs}
    Consider an HLP run for $\mathsf{T}$ level-1 SE rounds with a collection $\mathcal{R}$ of $X$-basis and $Z$-basis LMSs. We further assume perfect time boundaries for the HLP for simplicity. Let $\varpi_{L}$ be the maximum weight of measured logical Pauli operators in $\mathcal{R}$. Compilation parameters $\alpha_{\mathsf{b}}$ and $\alpha_{\mathsf{c}}$ are set as $d_1$ and $1$, respectively. We require that every $X$-basis (or $Z$-basis) logical readout gadget in LMS is separated from other logical readout gadgets in the same LMS, as well as $X$-basis (or $Z$-basis) readout gadgets for level-1 SE, by at least $d_{0}$ level-0 time steps. We follow the notation in Theorem~\ref{thm: approximate error reduction}. Then, under the phenomenological depolarizing noise model with an error rate $p<1/(4r\mathfrak{e})^{4}$ at level 0 and with a core distance $d_{0}\geq d_{th}$, the probability of any logical measurement error in $\mathcal{R}$ is upper bounded by $C_{\mathcal{R}}(p/p_{th})^{d_{0}d_{1}/4}+\mathrm{v}_{\mathsf{b}}(p/\sqrt{p_{th}})^{d_{0}d_1/2}$. Here, $d_{th}=2\log\left((4r_1\mathfrak{e})^2z_{\mathrm{max}}\right)/\log(p_{th}/p)$, $C_{\mathcal{R}}=|\mathcal{R}|N_{X}(2r_1\mathfrak{e}\sqrt{z_{\mathrm{max}}})^{d_1}$, $N_X=n_1(\delta_1+1)\mathsf{T}+s_1 \mathsf{T}+d_1$, and $r_1=12\delta_1^{2}\max(\varpi_1,\varpi_{L})$.   
\end{theorem}
\begin{proof}
    For an LMS $\mathcal{M}_{\mathsf{P}}$ in $\mathcal{R}$ for measuring a logical $X$ operator $\mathsf{P}$, we construct a level-1 error model for readout gadgets in $\mathcal{M}_{\mathsf{P}}\cup\mathcal{M}_{\mathsf{SE}}$. We then construct a level-1 adjacency graph with a vertex set of all primitive level-1 errors, such that two vertices are connected by an edge if and only if the corresponding level-1 errors have overlapping level-1 syndromes. Every primitive level-1 error on a core triggers at most $\delta_{1}+1$ level-1 detectors; every primitive level-1 error on a bus triggers at most two level-1 detectors. Every level-1 detector can be triggered by up to $2(\delta_1+1)\max(\varpi_{1},\varpi_{L})+2$ primitive level-1 errors. Thus, the degree of every vertex in the level-1 adjacency graph is upper bounded by $r_1=12\delta_{1}^{2}\max(\varpi_1,\varpi_{L})$. 

According to Theorem~\ref{thm: approximate error reduction}, except for a rare error event $\mathcal{B}_{d_1}$ whose probability is upper bounded by $\mathrm{v}_{\mathsf{b}}(p/\sqrt{p_{th}})^{d_{0}d_{1}/2}$, the induced level-1 $X$ errors are local stochastic, and each such error has a level-1 error probability upper bounded by $z_{\mathrm{max}}(p/p_{th})^{d_{0}/2}$. Here, $\mathrm{v}_{\mathsf{b}}=\tau_{\mathsf{b}}d_{0}^2d_{1}|\mathsf{b}|$ with $|\mathsf{b}|=s_1\mathsf{T}+d_{1}|\mathcal{R}|$, the number of shuttle buses used in the HLP. Let $\mathfrak{E}$ be the collection of all residual level-1 $X$ error configurations that (i) correspond to a connected cluster on the level-1 adjacency graph and (ii) trigger a logical measurement error for $\mathcal{M}_{\mathsf{P}}$. According to Lemma~\ref{lemma: level 1 distance for LMS}, each element in $\mathfrak{E}$ should have weight at least $d_{1}$. Note that the number of weight-$\mathrm{w}$ clusters on the adjacency graph containing a fixed vertex is upper bounded by $\left(r_1\mathfrak{e}\right)^{\mathrm{w}-1}$; the number of level-1 $X$ error locations is upper bounded by $N_X=n_1(\delta_1+1)\mathsf{T}+s_1 \mathsf{T}+d_1$. Thus, the number of weight-$\mathrm{w}$ elements in $\mathfrak{E}$ is upper bounded by $N_X (r_1\mathfrak{e})^{\mathrm{w}-1}$. The probability of inducing a weight-$\mathrm{w}$ element in $\mathfrak{E}$ is upper bounded by $2^{\mathrm{w}}p_1^{\mathrm{w}/2}$. Thus, excluding the rare event $\mathcal{B}_{d_1}$, the probability of a logical measurement error for $\mathcal{M}_{\mathsf{P}}$, $P_{E}(\mathcal{M}_{\mathsf{P}})$, satisfies
\begin{align}
    P_{E}(\mathcal{M}_{\mathsf{P}})\leq& N_X\sum_{\mathrm{w}\geq d_{1}}(r_1\mathfrak{e})^{\mathrm{w}-1}p_{1}^{\mathrm{w}/2}2^{\mathrm{w}} \nonumber\\
    \leq& \frac{N_X\left((2r_{1}\mathfrak{e})^{2}p_1\right)^{d_{1}/2}}{r_1\mathfrak{e}(1-2r_{1}\mathfrak{e}\sqrt{p_1})}\nonumber\\
    \leq&  N_X\left((2r_1\mathfrak{e})^2p_1\right)^{d_1/2},
\end{align}
where we used $d_{0}\geq d_{th}$ to guarantee $(2r_1\mathfrak{e})^2p_{1}\leq1/4$. Since the occurrence of a logical measurement error on $\mathcal{M}_{\mathsf{P}}$ corresponds to an event (a set) of level-0 physical error configurations, we then use the union bound to upper bound the probability of logical measurement errors in $\mathcal{R}$, $P_E(\mathcal{R})$, as follows.
\begin{align}
    P_{E}(\mathcal{R})\leq& |\mathcal{R}|N_X\left((2r_1\mathfrak{e})^2 p_1\right)^{d_1/2}+\mathrm{v}_{\mathsf{b}}(p/\sqrt{p_{th}})^{d_{0}d_{1}/2}\nonumber\\
    \leq& C_{\mathcal{R}} (p/p_{th})^{d_{0}d_{1}/4} + \mathrm{v}_{\mathsf{b}}(p/\sqrt{p_{th}})^{d_{0}d_{1}/2}.
\end{align}
\end{proof}
We note that for a low enough error rate $p$, Theorem~\ref{thm: fault probability for LMSs} upper bounds the probability of logical measurement errors by $\sim p^{d_{0}d_{1}/4}$, while for an HLP achieving full distance, the scaling of logical measurement errors should ideally be $\sim p^{d_0d_1/2}$. We attribute the discrepancy here to our proof technique, which separately analyzes level-0 and level-1 errors. More specifically, by analyzing level-0 errors, Theorem~\ref{thm: approximate error reduction} proves that the effective level-1 error probability $p_1$ scales as $\sim p^{d_{0}/2}$, indicating that full level-0 distance is achieved. Then, with this level-1 error rate, the probability of logical measurement errors at best scales as $\sim p_1^{d_{1}/2}\sim p^{d_{0}d_{1}/4}$ even when full level-1 distance is achieved. In contrast, our decoding procedure (Algorithm~\ref{Alg: full decoding procedure}) extracts soft information from level-0 decoding and applies this information to level-1 decoding, thereby effectively coupling the level-0 and level-1 decoding procedures. As a result, our circuit-level simulations in Fig.~\ref{fig: hlp performance}(\textbf{b}) and Fig.~\ref{fig: logical pauli measurement}(\textbf{c}) indicate HLPs therein are operating close to their full distance.
 
\subsection{Extension module I: measuring general level-1 logical Pauli operators}\label{subsec: extension model measure general logical operators}
Consider a logical operator $\mathsf{P}_{xz}:=i^{x\cdot z}\mathsf{X}_{1}^{x_1}\mathsf{Z}_1^{z_1}\otimes \cdots\otimes \mathsf{X}_{n_1}^{x_{n_1}}\mathsf{Z}_{n_1}^{z_{n_1}}$, where $x=(x_1,\cdots,x_{n_1})$ and $z=(z_{1},\cdots,z_{n_1})$ are vectors in $\mathbb{Z}_2^{n_1}$. When $x\cdot z$ is even, we construct the $H$-transformed readout gadget which proceeds as follows (see also Fig.~\ref{fig: logical pauli measurement}(\textbf{d})):
\begin{enumerate}
    \item Initialize the level-1 ancilla in the $Z$ basis. 
    \item Perform a sequence of CNOT gates between level-1 data qubits (controls) and the level-1 ancilla qubit (target), such that for every level-1 data qubit $i$ with $z_i=1$, exactly one CNOT gate is applied between that qubit and the ancilla. Every level-1 layer has at most $d_1$ control qubits. 
    \item   Perform a level-1 $H$ gate on the ancilla.
    \item Perform a sequence of CNOT gates between the level-1 ancilla (control) and level-1 data qubits (targets), such that for every level-1 data qubit $i$ with $x_i=1$, exactly one CNOT gate is applied between that qubit and the ancilla. Every level-1 layer has at most $d_1$ target qubits. 
    \item Measure the level-1 ancilla in the $X$ basis. 
\end{enumerate}
\begin{figure}
    \centering
    \includegraphics[width=\columnwidth]{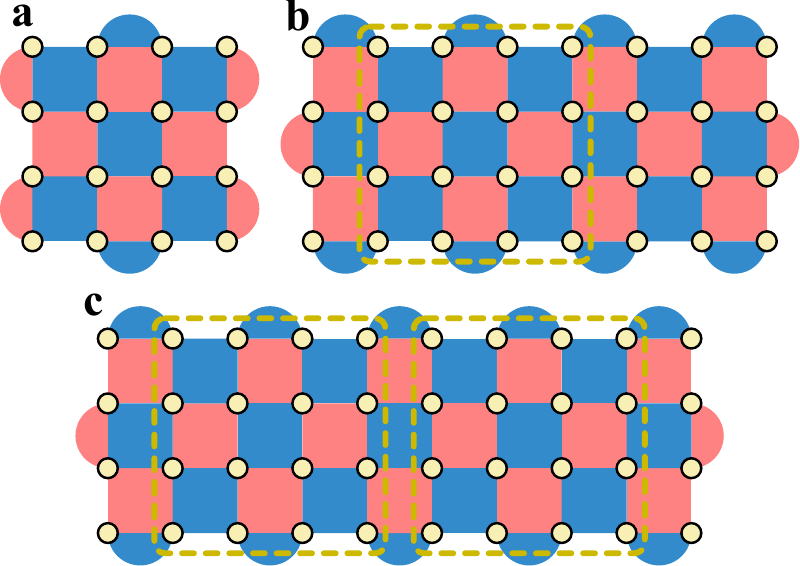}
    \caption{Bus-core CNOT gates during an $H$-transformed readout gadget. (\textbf{a}) Distance-4 rotated surface code. (\textbf{b}) $H$ bus after the transversal $H$ gate. Only one region (yellow dashed box) on the bus remains available for a bus-core CNOT gate. (\textbf{c}) A slightly elongated $H$ bus. Two regions are now available, allowing bus-core CNOT gates with up to two cores.}
    \label{fig: h bus corebus cnot gate}
\end{figure}
We now describe the level-0 implementation of an $H$-transformed readout gadget. The level-1 ancilla of the gadget is implemented by a shuttle bus, referred to as an $H$ bus. In the first step, the $H$ bus is transversally initialized in the $Z$ basis as a $Z$ bus. Then, the sequence of level-1 CNOT gates between level-1 data qubits (controls) and the level-1 ancilla (target) is implemented by core-bus CNOT gates between the corresponding cores and the $H$ bus. In the third step, we apply a layer of transversal $H$ gates on all data qubits of the $H$ bus, thereby implementing a logical $H$ gate on the level-1 ancilla and transforming the bus into an $X$ bus. We group this layer of transversal $H$ gates with the following level-0 SE round as a level-0 time step. Similarly, level-1 CNOT gates in the fourth step are implemented by bus-core CNOT gates between the $H$ bus (control) and cores (targets). Finally, the $H$ bus is transversally measured in the $X$ basis. Similar to $X$-basis and $Z$-basis readout gadgets, we insert padding level-0 SE rounds on the $H$ bus, so that (i) every two core-bus (or bus-core) CNOT gates in this gadget are separated by at least $\alpha_{\mathsf{b}}d_{0}$ level-0 SE rounds and (ii) transversal initialization and measurement are separated from the following and the preceding hybrid-unit CNOT gates, respectively, by at least $\alpha_{\mathsf{b}}d_{0}$ level-0 SE rounds. In addition, we require the layer of transversal $H$ gates to be separated from any core-bus or bus-core CNOT gate by at least $\alpha_{H}d_{0}$ level-0 SE rounds. Here, $\alpha_{H}$ is a new compilation parameter. For an even distance $d_{0}$, since the core has asymmetrical $X$ and $Z$ boundaries (Fig.~\ref{fig: cores and buses}(\textbf{d})), the $H$ bus after the transversal $H$ gate is mismatched with cores along $X$ boundaries (Fig.~\ref{fig: h bus corebus cnot gate}(\textbf{a}\textendash{}\textbf{b})). Therefore, now the $H$ bus only supports bus-core CNOT gates with up to $d_{1}-1$ cores at a time. This minor issue can be fixed by using an $H$ bus of length $d_{0}d_{1}+2$ (Fig.~\ref{fig: h bus corebus cnot gate}). For simplicity, we will not make such a modification to the bus length in what follows.  

\begin{figure}
    \centering
    \includegraphics[width=\columnwidth]{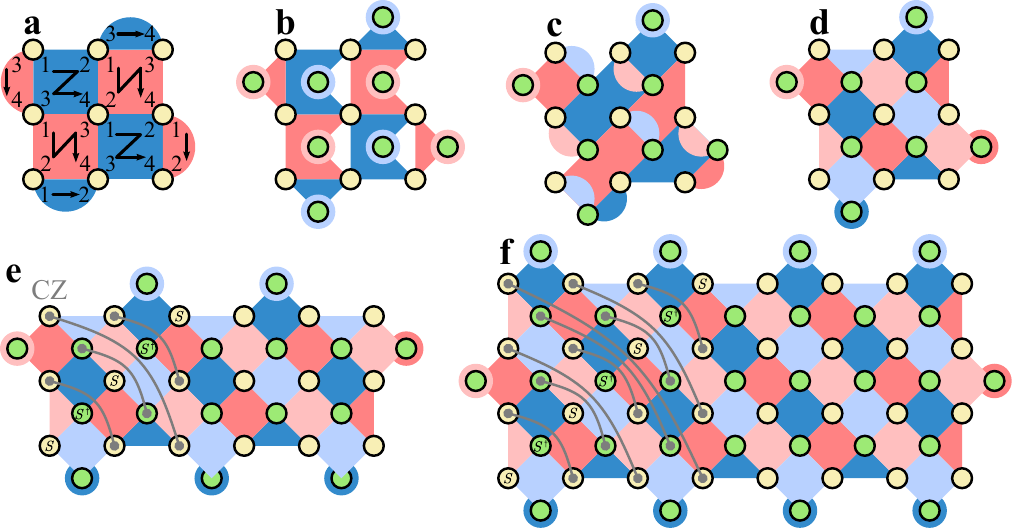}
    \caption{Mid-cycle transversal $S$ gate. (\textbf{a}) CNOT gate schedule for a level-0 SE circuit on the rotated surface code. The same schedule is applied to $H$ buses (after the transversal $H$ gate). (\textbf{b}\textendash{}\textbf{d}) Mid-cycle states after ancilla reset, the first layer of CNOT gates, and the second layer of CNOT gates, respectively, during a level-0 SE round (adapted from Refs.~\cite{mcewen_relaxing_2023,chen_transversal_2026}). The last state is also called a half-cycle state. (\textbf{e}) Fold-transversal $S$ gate applied on the half-cycle state of an $H$ bus with $d_{0}=3$ and $d_{1}=2$. (\textbf{f}) Fold-transversal $S$ gate for an $H$ bus with $d_{0}=4$ and $d_{1}=2$. Both $H$ buses in (\textbf{e}\textendash{}\textbf{f}) have long logical $X$ operators. Each mid-cycle fold-transversal gate in (\textbf{e}\textendash{}\textbf{f}) is embedded in a single level-0 SE round and implements a logical $S$ gate on the $H$ bus.}
    \label{fig: mid cycle transversal S gate}
\end{figure}

If $x\cdot z$ is odd, there are two ways to measure $\mathsf{P}_{xz}$. First, we can use a logical $Y$ state as a catalyst to convert the logical measurement task back to the case above, where $x\cdot z$ is even~\cite{chamberland_universal_2022}. Suppose we reserve a single logical qubit in an HLP to host a logical $Y$ state and we would like to perform the logical Pauli measurement $\mathsf{P}_{xz}$ supported on other logical qubits. Then, we can instead measure $\mathsf{P}_{xz}\cdot \overline{\mathsf{Y}}$ using the $H$-transformed logical readout gadget; here, $\overline{\mathsf{Y}}$ is the logical $Y$ operator corresponding to the logical $Y$ state. In this way, we effectively measure $\mathsf{P}_{xz}$ while preserving the logical $Y$ state, which can be reused to assist future logical Pauli measurements. Alternatively, we can construct another readout gadget, called an $HS$-transformed readout gadget, to directly measure $\mathsf{P}_{xz}$ without using ancillary logical states. The $HS$-transformed readout gadget works as follows.
\begin{itemize}
    \item[1\textendash{}4.] Same as steps 1\textendash{}4 for the $H$-transformed readout gadget. 
    \item[5.] Perform a level-1 $S$ gate on the ancilla.
    \item[6.] Measure the ancilla in $X$ basis. 
\end{itemize}
The level-0 implementation of an $HS$-transformed readout gadget mostly follows from that of an $H$-transformed readout gadget. The only new step for the former is to perform a logical $S$ gate on the $H$ bus. Following the approach in Ref.~\cite{chen_transversal_2026}, we can perform a mid-cycle fold-transversal $S$ gate by leveraging the dynamics of the level-0 SE circuit on the bus~\cite{mcewen_relaxing_2023,chen_transversal_2026} (see Fig.~\ref{fig: mid cycle transversal S gate}). We require the mid-cycle transversal $S$ gate to be separated from both the preceding bus-core CNOT gate and the final transversal $X$-basis measurement by at least $\alpha_{\mathsf{b}}d_{0}$ level-0 SE rounds. 

We note that the level-0 implementation of both $H$-transformed and $HS$-transformed readout gadgets requires numerical benchmarking to optimize the number of padding level-0 SE rounds required between two adjacent logical operations on an $H$ bus. We leave this to future work. In the remainder of this subsection, we theoretically analyze circuits with $H$-transformed readout gadgets (in addition to $X$-basis and $Z$-basis readout gadgets) under the phenomenological depolarizing noise model. We will extend Theorem~\ref{thm: approximate error reduction} to prove that level-1 $X$ (or $Z$) errors are almost uncorrelated even in the presence of $H$-transformed gadgets. Based on this result, we then formulate and analyze an LMS based on $H$-transformed gadgets to reliably measure a logical operator $\mathsf{P}_{xz}$ with $x\cdot z=0$. (The case $x\cdot z=1$ can also be handled by such an LMS, given a logical $Y$ state.) 

Consider a level-0 circuit consisting of $X$-basis, $Z$-basis, and $H$-transformed readout gadgets, together with level-0 SE rounds on all working units. As described in Sec.~\ref{subsec: detectors and decoding}, we partition level-0 detectors on cores and shuttle buses in $X$-basis and $Z$-basis readout gadgets into $\mathcal{D}_{\mathsf{B}_{\lozenge}}$, $\mathcal{D}_{\mathsf{C}}$, and $\mathcal{D}_{\mathsf{B}_{\blacklozenge}}$ and primitive level-0 $X$ and $Z$ errors on these working units into $\mathcal{E}_{\mathsf{B}_{\lozenge}}$, $\mathcal{E}_{\mathsf{C}}$, and $\mathcal{E}_{\mathsf{B}_{\blacklozenge}}$. Since an $H$ bus behaves similarly to an $X$ or $Z$ bus, we likewise partition the level-0 detectors on $H$ buses into $\mathcal{D}_{\mathsf{B}_{\lozenge}}$ and $\mathcal{D}_{\mathsf{B}_{\blacklozenge}}$, and the primitive level-0 $X$ and $Z$ errors on these buses into $\mathcal{E}_{\mathsf{B}_{\lozenge}}$ and $\mathcal{E}_{\mathsf{B}_{\blacklozenge}}$. 

For an $H$ bus $\sigma$ in an $H$-transformed gadget, let $t_{H}$ be the level-0 time step containing the transversal $H$ gate on $\sigma$. Construction of level-0 detectors on $\sigma$ before (or after) time step $t_{H}$ is the same as that on a $Z$ bus (or an $X$ bus). Every level-0 detector generated at time step $t_{H}$ consists of the measurement result on a $Z$ or $X$ stabilizer in time step $t_{H}-1$ and the measurement result on the corresponding $X$ or $Z$ stabilizer in time step $t_{H}$. We refer to a level-0 detector on the bus as upstream (or downstream) if it is supported on $X$ (or $Z$) stabilizer measurements before time step $t_H$ and $Z$ (or $X$) stabilizer measurements at or after time step $t_{H}$. We refer to primitive qubit $Z$ (or $X$) errors at or before time step $t_{H}$, $X$ (or $Z$) stabilizer measurement errors before time step $t_{H}$, qubit $X$ (or $Z$) errors after time step $t_{H}$, and $Z$ (or $X$) stabilizer measurement errors at or after time step $t_{H}$ as upstream (or downstream) errors. Then, we add all upstream and downstream level-0 detectors on $H$ buses to $\mathcal{D}_{\mathsf{B}_{\lozenge}}$ and $\mathcal{D}_{\mathsf{B}_{\blacklozenge}}$, respectively; we add all upstream and downstream level-0 errors on these buses to $\mathcal{E}_{\mathsf{B}_{\lozenge}}$ and $\mathcal{E}_{\mathsf{B}_{\blacklozenge}}$, respectively. We can then construct all three level-0 decoding graphs and perform level-0 decoding according to Sec.~\ref{subsec: detectors and decoding}. Define the upstream (or downstream) decoding subgraph on an $H$ bus $\sigma$ as the subgraph of $\mathcal{G}_{\mathsf{B}_{\lozenge}}$ (or $\mathcal{G}_{\mathsf{B}_{\blacklozenge}}$) induced by the set of all upstream (or downstream) level-0 detectors on $\sigma$. 

We define the $Z$-basis (or $X$-basis) separation of an $H$-transformed readout gadget from another readout gadget as the minimum separation between a core-bus (or bus-core) CNOT gate in the $H$-transformed gadget and a core-bus (or bus-core) CNOT gate in the other. Suppose we are interested in a subset $\mathcal{M}$ of readout gadgets, such that (i) $X$-basis (or $Z$-basis) readout gadgets in $\mathcal{M}$ are separated by at least $d_{0}$ level-0 time steps, and (ii) the $X$-basis and $Z$-basis separations of each $H$-transformed readout gadget in $\mathcal{M}$ from any other readout gadget in $\mathcal{M}$ are at least $d_{0}$ level-0 time steps. Following the notation in Sec.~\ref{sec: approximate level 1 error reduction}, we let $\mathcal{E}_{\mathsf{B}_{\blacktriangledown}}$ be the subcollection of all level-0 errors in $\mathcal{E}_{\mathsf{B}_{\blacklozenge}}$ on buses in $\mathcal{M}$. 

Consider a general level-0 residual error configuration $f$. We know that $f|_{\mathsf{B}_{\lozenge}}:=f_X|_{\mathsf{B}_{\lozenge}}\cup f_{Z}|_{\mathsf{B}_{\lozenge}}$ is a collection of $X$-basis cycles on $X$ buses, $Z$-basis cycles on $Z$ buses, and cycles on the upstream decoding subgraphs on $H$ buses. We restrict all cycles in $f|_{\mathsf{B}_{\lozenge}}$ to have a weight smaller than $d_{0}d_{1}$ as we did in Sec.~\ref{subsec: error reduction theorem}. Then, every cross-membrane cycle in $f|_{\mathsf{B}_{\lozenge}}$ on an $H$ bus is either an $X$-basis or $Z$-basis cycle. Let $f^{\circ}$ denote the collection of propagated simple walks (on cores) from cycles in $f|_{\mathsf{B}_{\lozenge}}$; let $f_X^{\circ}$ (or $f_{Z}^{\circ}$) denote the collection of propagated simple walks (on cores) from $X$-basis (or $Z$-basis) cycles in $f|_{\mathsf{B}_{\lozenge}}$. Then, $f^{\circ}=f_{X}^{\circ}\cup f_{Z}^{\circ}$. Following Sec.~\ref{subsec: error reduction theorem}, $f_{X}^{\circ}\cup f_{X}|_{\mathsf{C}}$ is a disjoint union of $X$-basis closed walks on cores, such that each closed walk consists of propagated simple walks in $f_{X}^{\circ}$ and length-1 walks (primitive level-0 errors) in $f_{X}|_{\mathsf{C}}$. Let $f_{X}^{\bullet}$ denote the propagated simple walks on buses in $\mathcal{M}$ from these closed walks; we similarly define $f_{Z}^{\bullet}$ as the propagated simple walks on buses in $\mathcal{M}$ from $f_{Z}^{\circ}\cup f_{Z}|_{\mathsf{C}}$.  

Restricted to $X$ (or $Z$) buses, $f_Z^{\bullet}\cup f_{Z}|_{\mathsf{B}_{\blacktriangledown}}$ (or $f_{X}^{\bullet}\cup f_{X}|_{\mathsf{B}_{\blacktriangledown}}$) is a disjoint union of $Z$-basis (or $X$-basis) closed walks, potentially triggering readout errors on $X$-basis (or $Z$-basis) readout gadgets. However, restricted to the downstream decoding subgraph of an $H$ bus, a closed walk may contain both $X$-basis walks in $f_{X}^{\bullet}\cup f_{X}|_{\mathsf{B}_{\blacktriangledown}}$  and $Z$-basis walks in $f_{Z}^{\bullet}\cup f_{Z}|_{\mathsf{B}_{\blacktriangledown}}$. In other words, restricted to $H$ buses, $f_{X}^{\bullet}\cup f_{X}|_{\mathsf{B}_{\blacktriangledown}}\cup f_{Z}^{\bullet}\cup f_{Z}|_{\mathsf{B}_{\blacktriangledown}}$ is a disjoint union of closed walks on downstream decoding subgraphs. We will choose the compilation parameter $\alpha_{H}$ sufficiently large that closed walks containing walks from both $f_{Z}^{\bullet}$ and $f_X|_{\mathsf{B}_{\blacktriangledown}}$, or from both $f_X^{\bullet}$ and $f_{Z}|_{\mathsf{B}_{\blacktriangledown}}$, are strongly suppressed.  

We assign two level-1 error locations during the lifetime of an $H$ bus: a level-1 $X$ error location (before the transversal $H$ gate) and a level-1 $Z$ error location (after the transversal $H$ gate). Both error locations trigger the same readout error on the corresponding $H$-transformed readout gadget. Similar to Def.~\ref{def: receptive zone for X or Z error}, we define the receptive zones for both error locations.
\begin{definition}[Receptive zones for level-1 error locations on an $H$ bus]
    Consider an $H$ bus $\sigma$ assigned with two primitive level-1 errors $\mathsf{e}_{\mathsf{X}}$ and $\mathsf{e}_{\mathsf{Z}}$ corresponding to the $X$ error location and the $Z$ error location, respectively. The receptive zone for $\mathsf{e}_{\mathsf{X}}$ is the collection of all level-0 $X$ error locations satisfying at least one of the following conditions. 
\begin{itemize}
        \item The error location corresponds to a downstream error on $\sigma$. 
        \item  The error location is on a CNOT membrane of a core induced by a core-bus CNOT gate $\Lambda_{t}$, which acts on $\sigma$ at level-0 time step $t$.  
\end{itemize}
The receptive zone for $\mathsf{e}_{\mathsf{Z}}$ is the collection of all level-0 $Z$ error locations satisfying at least one of the following conditions.
\begin{itemize}
    \item The error location corresponds to a downstream error on $\sigma$.
    \item The error location is on a CNOT membrane of a core induced by a bus-core CNOT gate $\Lambda_t$, which acts on $\sigma$ at level-0 time step $t$. 
\end{itemize}
\end{definition}

We generalize the definition of downstream webs (Def.~\ref{def: downstream fault web}). Let $\{l_{1},\cdots,l_u\}$ be the collection of disjoint closed walks such that their union is $f_{X}^{\bullet}\cup f_{X}|_{\mathsf{B}_{\blacktriangledown}}\cup f_{Z}^{\bullet}\cup f_{Z}|_{\mathsf{B}_{\blacktriangledown}}$. Consider two $X$-basis closed walks $\ell_{1}$ and $\ell_{2}$ as subsets of $f_X^{\circ}\cup f_{X}|_{\mathsf{C}}$ on cores. Let $\ell_{1}^{\bullet}$ and $\ell_{2}^{\bullet}$ be the collections of propagated simple walks from $\ell_{1}$ and $\ell_{2}$, respectively. We say  $\ell_1$ and $\ell_2$ are linked if there is a propagated simple walk in $\ell_1^{\bullet}$ and another propagated simple walk in $\ell_{2}^{\bullet}$ that belong to the same closed walk in $\{l_1,\cdots,l_u\}$. Consider a set $w$ consisting of closed walks on cores and native walks in $f_{X}|_{\mathsf{B}_{\blacktriangledown}}\cup f_{Z}|_{\mathsf{B}_{\blacktriangledown}}$, such that (i) $w|_{\mathsf{C}}\subset f_X^{\circ}\cup f_{X}|_{\mathsf{C}}$, (ii) $w|_{\mathsf{B}_{\blacktriangledown}}\subset f_X|_{\mathsf{B}_{\blacktriangledown}}\cup f_Z|_{\mathsf{B}_{\blacktriangledown}}$ and contains no native closed walk, and (iii) $\partial_{\blacktriangledown}[w]=0$. We say $w$ is a generalized downstream fault web if $w|_{\mathsf{C}}$ is a linked cluster of closed walks on cores.   

We now extend Lemma~\ref{lemma: bounds on a general residual error config} to bound $f$ under certain additional restrictions. We follow the notation in Sec.~\ref{subsec: error reduction theorem}, where the section map $\mathbb{S}_{f}$ is defined, and the subscript $\epsilon$ for a general level-0 error configuration denotes its physical-error component. 
\begin{lemma}
    [Extension of Lemma~\ref{lemma: bounds on a general residual error config}]\label{lemma: extension of error bounding lemma}
    Consider the general level-0 residual error configuration $f$. Suppose $f|_{\mathsf{B}_{\lozenge}}$ contains no cycle with weight at least $d_{0}d_{1}$. We require that, for any closed walk $l_{b}$ ($b\in\{1,\cdots,u\}$) on a downstream decoding subgraph of an $H$ bus, if $l_{b}\cap f_{X}|_{\mathsf{B}_{\blacktriangledown}}\neq \emptyset$, then (i) $l_{b}\cap(f_{Z}^{\bullet})=\emptyset$, and (ii) $\mathbb{S}_{f}(l_{b}\cap f_{Z}|_{\mathsf{B}_{\blacktriangledown}})\cap \mathbb{S}_{f}(\tilde{\ell})=\emptyset$ for every cross-membrane cycle $\tilde{\ell}$ in $f_X|_{\mathsf{B}_{\lozenge}}$. We can explicitly construct a level-1 $X$ error $\mathsf{f}_{\mathsf{X}}$, such that $f$ induces a level-1 error configuration whose $X$-basis component is exactly $\mathsf{f}_{\mathsf{X}}$. For any given level-1 error subconfiguration $\mathsf{h}:=\mathsf{e}_{1}+\cdots+\mathsf{e}_{|\mathsf{h}|}\subset \mathsf{f}_{\mathsf{X}}$ (each $\mathsf{e}_i$ is a distinct primitive level-1 error), we can find two general level-0 error configurations $\lambda_{\mathsf{h}}\subset\mu_{\mathsf{h}}\subset f$ such that (i) $\lambda_{\mathsf{h}}=e_1+\cdots+e_{|\mathsf{h}|}$ with each $e_{i}$ a distinct primitive level-0 error whose $X$-basis projection is in the receptive zone of $\mathsf{e}_{i}$, (ii) each connected component of $\mu_{\mathsf{h}}$ (on the level-0 adjacency graph) contains at least one element in $\lambda_{\mathsf{h}}$, and (iii) the weight of physical errors in $\mu_{\mathsf{h}}$, $|\mu_{\mathsf{h}}\cap \epsilon|$, is lower bounded by $\max(|\mu_{\mathsf{h}}|/4,d_0|\mathsf{h}|/2)$.
\end{lemma}
\begin{proof}
   For a lightning walk $l_{b}$ in $\{l_{1},\cdots,l_{u}\}$ on an $H$ bus $\sigma$, we let $l_{b}$ induce the primitive level-1 $X$ error on $\sigma$ if $l_{b}\cap (f_{X}^{\bullet}\cup f_{X}|_{\mathsf{B}_{\blacktriangledown}})\neq \emptyset$; otherwise, we let $l_{b}$ induce the primitive level-1 $Z$ error on $\sigma$. The first requirement on $f$ guarantees that we can find a set $\{w_1,\cdots,w_o\}$ of generalized downstream webs and a set of native closed walks $\{l_{b_1},\cdots,l_{b_v}\}\subset \{l_1,\cdots,l_u\}$, such that (i) these webs and closed walks are all mutually disjoint, (ii) their union contains $f_X^{\circ}\cup f_{X}|_{\mathsf{C}}\cup f_X|_{\mathsf{B}_{\blacktriangledown}}$ and is contained in $f_{X}^{\circ}\cup f_{X}|_{\mathsf{C}}\cup f_X|_{\mathsf{B}_{\blacktriangledown}}\cup f_{Z}|_{\mathsf{B}_{\blacktriangledown}}$, and (iii) each $l_{b_i}$ ($i\in\{1,\cdots,v\}$) intersects with $f_{X}^{\bullet}\cup f_{X}|_{\mathsf{B}_{\blacktriangledown}}$. Let $g=\left(\bigcup_{i=1}^{o}w_i\right)\cup\left(\bigcup_{i=1}^{v}l_{b_i}\right)$. Following Lemma~\ref{lemma: bounds on an error config}, we can construct a level-1 $X$ error configuration $\mathsf{f}_{\mathsf{X}}$, such that $g$ induces $\mathsf{f}_{\mathsf{X}}$. Furthermore, for $\mathsf{h}\subset \mathsf{f}_{\mathsf{X}}$, we can find two level-0 error configurations $\lambda_{\mathsf{h},X}$ and $\mu_{\mathsf{h},XZ}$ with $\lambda_{\mathsf{h},X}\subset \mu_{\mathsf{h},XZ}\subset (f_X\cup f_{Z})$, such that (i) $\lambda_{\mathsf{h},X}=e_{X,1}+\cdots+e_{X,|\mathsf{h}|}$ with each $e_{X,i}$ a distinct level-0 $X$ error in the receptive zone of $\mathsf{e}_{i}$, (ii) each connected component of $\mu_{\mathsf{h},XZ}$ on the level-0 adjacency graph contains at least one element in $\lambda_{\mathsf{h},X}$, (iii) $\mu_{\mathsf{h},XZ}|_{\mathsf{B}_{\lozenge}}\subset f_{X}|_{\mathsf{B}_\lozenge}$, $\mu_{\mathsf{h},XZ}|_{\mathsf{C}}\subset f_X|_{\mathsf{C}}$, and $\mu_{\mathsf{h},XZ}|_{\mathsf{B}_{\blacktriangledown}}\subset g\cap (f_X|_{\mathsf{B}_{\blacktriangledown}}\cup f_Z|_{\mathsf{B}_{\blacktriangledown}})$, and (iv) $|\mu_{\mathsf{h},XZ}\cap (f_{X,\epsilon}\cup f_{Z,\epsilon})|\geq \max(|\mu_{\mathsf{h},XZ}|/4,d_{0}|\mathsf{h}|/2)$. The second requirement on $f$ guarantees that $\mathbb{S}_f(\mu_{\mathsf{h},XZ}\cap f_{X,\epsilon})\cap \mathbb{S}_f(\mu_{\mathsf{h},XZ}\cap f_{Z,\epsilon})=\emptyset$, indicating that $|\mathbb{S}_f(\mu_{\mathsf{h},XZ})\cap f_{\epsilon}|\geq \max(|\mathbb{S}_f(\mu_{\mathsf{h},XZ})|/4,d_{0}|\mathsf{h}|/2)$. We then prove the lemma by setting $\lambda_{\mathsf{h}}=\mathbb{S}_{f}(\lambda_{\mathsf{h},X})$ and $\mu_{\mathsf{h}}=\mathbb{S}_f(\mu_{\mathsf{h},XZ})$.
\end{proof}

We now state the following corollary, which extends Theorem~\ref{thm: approximate error reduction} to incorporate $H$-transformed gadgets. 
\begin{corollary}
    [Extension of Theorem~\ref{thm: approximate error reduction}]\label{corollary: extension of error reduction thm}
    Consider a level-0 circuit composed of level-0 SE rounds on cores and shuttle buses and X-basis, $Z$-basis, and $H$-transformed readout gadgets.  Suppose we are only interested in the results of a subset $\mathcal{M}_1$ of readout gadgets, such that (i) the separation between every two $X$-basis (or $Z$-basis) readout gadgets in $\mathcal{M}_1$ is at least $d_{0}$; (ii) the $X$-basis and $Z$-basis separations between an $H$-transformed gadget and other gadgets are at least $d_{0}$. The compilation parameters $\alpha_{\mathsf{c}}$, $\alpha_{\mathsf{b}}$, and $\alpha_{H}$ are set as $1$, $d_{1}$, and $2d_{1}$, respectively. We follow the notation in Theorem~\ref{thm: approximate error reduction}. Then, under the phenomenological depolarizing noise model with a noise rate $p<1/(4r\mathfrak{e})^{10}$, except for a rare event with a probability upper bounded by $2\mathrm{v}_{\mathsf{b}}(p/p_{th}')^{d_{0}d_{1}/2}$, the induced level-1 $X$ (or $Z$) errors are local stochastic, each with a level-1 error rate upper bounded by $z_{\mathrm{max}}(p/p_{th})^{d_{0}/2}$. Here, $p_{th}'=1/(2r\mathfrak{e})^{10}$.
\end{corollary}
\begin{proof}
    The proof idea is to exclude enough rare events so that the remaining physical error configurations induce general residual error configurations that would always meet the requirements in Lemma~\ref{lemma: extension of error bounding lemma}. In this way, the proof of Theorem~\ref{thm: approximate error reduction} works. 
    
    Similar to the proof of Theorem~\ref{thm: approximate error reduction}, we first exclude the event $\mathcal{B}$ that a general level-0 residual error configuration induces an $X$-basis cycle on an $X$ bus, a $Z$-basis cycle on a $Z$ bus, or an upstream cycle on an $H$ bus, of weight at least $d_{0}d_{1}$. According to Theorem~\ref{thm: approximate error reduction}, $P(\mathcal{B})\leq \mathrm{v}_{\mathsf{b}}(p/\sqrt{p_{th}})^{d_{0}d_1/2}$.

    Consider a general level-0 residual error configuration $f$ corresponding to a connected component on the level-0 adjacency graph. Following the notation in Lemma~\ref{lemma: extension of error bounding lemma}, we now analyze when $f$ would violate the two requirements therein. Consider a closed walk $l_{b}\in\{l_1,\cdots,l_{u}\}$ on an $H$ bus $\sigma$ with $l_{b}\cap f_{X}|_{\mathsf{B}_{\blacktriangledown}}\neq \emptyset$. Among the (mutually disjoint) closed walks on cores formed by $f^{\circ}\cup f_{X}|_{\mathsf{C}}\cup f_{Z}|_{\mathsf{C}}$, let $\{\ell_{1},\cdots,\ell_{a}\}$ be the subcollection of all sources of simple walks in $l_{b}\cap (f_{X}^{\bullet}\cup f_{Z}^{\bullet})$. Among the cross-membrane cycles in $f_X|_{\mathsf{B}_{\lozenge}}\cup f_Z|_{\mathsf{B}_{\lozenge}}$, let $\{\tilde{\ell}_1,\cdots,\tilde{\ell}_{r}\}$ be the subcollection of all sources of simple walks in $\bigcup_{i=1}^{a}(\ell_i\cap f^{\circ})$, and let $\{\tilde{\ell}_{r+1},\cdots,\tilde{\ell}_{r+s}\}$ be the subcollection of all cycles that act on qubits in the support of $l_{b}\cap f_{Z}|_{\mathsf{B}_{\blacktriangledown}}$. Define $\ell=\bigcup_{i=1}^{a} \ell_{i}$, $\tilde{\ell}=\bigcup_{i=1}^{r}\tilde{\ell}_i$, and $\tilde{\ell}'=\bigcup_{i=1}^{s}\tilde{\ell}_{r+i}$. Define
\begin{align}\label{eq: downward pyramid}
    \mu=&\big[\left(l_{b}\cap(f_{X}|_{\mathsf{B}_{\blacktriangledown}}\cup f_{Z}|_{\mathsf{B}_{\blacktriangledown}})\right) \nonumber \\
    &\cup \left(\ell\cap \left(f_X|_{\mathsf{C}}\cup f_Z|_{\mathsf{C}}\right)\right)\cup \tilde{\ell}\cup \tilde{\ell}'\big].
\end{align}
By construction, $\mu\subset f_{X}\cup f_{Z}$ is a connected cluster on the level-0 adjacency graph. 

By MWPM decoding, we have the following bounds on $\tilde{\ell}$ and $\tilde{\ell}'$.
\begin{align}\label{eq: cloud cycles of the downward pyramid}
    |\tilde{\ell}|\leq& 2|\tilde{\ell}\cap (f_{X,\epsilon}|_{\mathsf{B}_{\lozenge}}\cup f_{Z,\epsilon}|_{\mathsf{B}_{\lozenge}})| \nonumber\\
    |\tilde{\ell}'|\leq& 2|\tilde{\ell}'\cap (f_{X,\epsilon}|_{\mathsf{B}_{\lozenge}}\cup f_{Z,\epsilon}|_{\mathsf{B}_{\lozenge}})|
\end{align}
As the total weight of propagated simple walks from $\tilde{\ell}$ is upper bounded by $|\tilde{\ell}|/2$ (Lemma~\ref{lemma: propagation of cross-membrane cycle}), we see that
\begin{align}\label{eq: core cycle in the downward pyramid}
    |\ell\cap (f_X|_{\mathsf{C}}\cup f_{Z}|_{\mathsf{C}})|\leq& 2|\ell\cap(f_{X,\epsilon}|_{\mathsf{C}}\cup f_{Z,\epsilon}|_{\mathsf{C}})| + |\tilde{\ell}_{||}|/2 \nonumber\\
    |\ell|\leq& 2|\ell\cap(f_{X,\epsilon}|_{\mathsf{C}}\cup f_{Z,\epsilon}|_{\mathsf{C}})| + |\tilde{\ell}_{||}|
\end{align}
As $\ell$ may contain both lightning and cross-membrane walks, the total weight of propagated simple walks from $\ell$ is upper bounded by $|\ell|$. Then, we have
\begin{align}\label{eq: root of downward pyramid}
    |l_{b}\cap (f_{X}|_{\mathsf{B}_{\blacktriangledown}}\cup f_{Z}|_{\mathsf{B}_{\blacktriangledown}})|\leq& 2|l_{b}\cap (f_{X,\epsilon}|_{\mathsf{B}_{\blacktriangledown}}\cup f_{Z,\epsilon}|_{\mathsf{B}_{\blacktriangledown}})| + |\ell| 
\end{align}
Combining Eqs.~\ref{eq: cloud cycles of the downward pyramid},~\ref{eq: core cycle in the downward pyramid}, and~\ref{eq: root of downward pyramid}, we have
\begin{align}\label{eq: bounds on downward pyramid}
    |\mu|/5\leq& |\mu\cap (f_{X,\epsilon}\cup f_{Z,\epsilon})|\nonumber\\
    |l_{b}\cap (f_X|_{\mathsf{B}_{\blacktriangledown}}\cup f_Z|_{\mathsf{B}_{\blacktriangledown}})|+|\tilde{\ell}'|\leq& 2|\mu\cap (f_{X,\epsilon}\cup f_{Z,\epsilon})|
\end{align}
If $l_{b}$ does not satisfy the first requirement in Lemma~\ref{lemma: extension of error bounding lemma}, then $l_{b}\cap f_{Z}^{\bullet}\neq \emptyset$, which implies
\begin{equation}\label{eq: H bus cycle intersect w z membrane}
    2d_{0}d_{1}\leq|l_{b,\perp}|\leq |l_{b}\cap (f_X|_{\mathsf{B}_{\blacktriangledown}}\cup f_Z|_{\mathsf{B}_{\blacktriangledown}})|.
\end{equation}
If $l_{b}$ does not satisfy the second requirement in Lemma~\ref{lemma: extension of error bounding lemma}, then
\begin{equation}\label{eq: H bus cycle down up mix}
    2d_{0}d_{1}\leq |l_{b,\perp}| + |\tilde{\ell}_{\perp}'|\leq |l_{b}\cap (f_X|_{\mathsf{B}_{\blacktriangledown}}\cup f_Z|_{\mathsf{B}_{\blacktriangledown}})|+|\tilde{\ell}'|.
\end{equation}
From Eq.~\ref{eq: bounds on downward pyramid}, we see that both Eq.~\ref{eq: H bus cycle intersect w z membrane} and Eq.~\ref{eq: H bus cycle down up mix} imply
\begin{equation}\label{eq: physical errors in a downward pyramid}
    d_{0}d_{1}\leq |\mu \cap (f_{X,\epsilon}\cup f_{Z,\epsilon})|.
\end{equation}
Let $\mu_{f}=\mathbb{S}_{f}(\mu)$. Then, the weight of physical errors in $\mu_f$ is lower bounded by $|\mu\cap (f_{X,\epsilon}\cup f_{Z,\epsilon})|/2$. From Eq.~\ref{eq: bounds on downward pyramid} and Eq.~\ref{eq: physical errors in a downward pyramid}, we see that when $f$ violates either requirement in Lemma~\ref{lemma: extension of error bounding lemma}, we can find a subconfiguration $\mu_{f}\subset f$, such that (i) $\mu_f$ is a connected cluster on the level-0 adjacency graph, (ii) $\mu_{f}$ contains at least one downstream level-0 error on an $H$ bus, and (iii) the weight of physical errors in $\mu_f$ is lower bounded by $\max(|\mu_{f}|/10,d_{0}d_{1}/2)$. 

Consider the event $\mathcal{B}'$ that the induced general level-0 residual error configuration violates at least one requirement in Lemma~\ref{lemma: extension of error bounding lemma}. Then, following the proof of Lemma~\ref{lemma: sufficient condition for bounding level 1 error}, we have
\begin{align}
    P(\mathcal{B}')\leq \mathrm{v}_{\mathsf{b}}\left(p/p_{th}'\right)^{d_{0}d_{1}/2} 
\end{align}
We note that the right-hand side of the above inequality also upper bounds $P(\mathcal{B})$. After discarding both $\mathcal{B}$ and $\mathcal{B}'$ (with a joint probability upper bounded by $2\mathrm{v}_{\mathsf{b}}\left(p/p_{th}'\right)^{d_{0}d_1/2}$), all induced general level-0 residual error configurations satisfy Lemma~\ref{lemma: extension of error bounding lemma}. Therefore, the proof of Theorem~\ref{thm: approximate error reduction} also holds and provides the same upper bound for the level-1 error rate.    
\end{proof}
We note that we have not optimized for a larger threshold value in the above corollary. We only aimed to prove that for a small enough error rate $p$, except for a rare event with probability $\sim p^{d_{0}d_1/2}$, the induced level-1 $X$ (or $Z$) errors are effectively local stochastic with a level-1 error rate $\sim p^{d_{0}/2}$. This result then aligns with Theorem~\ref{thm: approximate error reduction}. 

We now construct an LMS, referred to as an $H$-LMS, from $H$-transformed logical readout gadgets to reliably measure a logical operator $\mathsf{P}_{xz}$ with $x\cdot z=0$. We note that the level-1 decoding process for a single $H$-LMS simultaneously involves level-1 $X$ and $Z$ errors. However, Corollary~\ref{corollary: extension of error reduction thm} only proves that the level-1 $X$ (or $Z$) errors are almost uncorrelated but does not rule out correlation between level-1 $X$ and $Z$ errors. In practice, it is empirically observed that logical $X$ and $Z$ errors on an idling RSC are almost uncorrelated~\cite{gidney_yoked_2025}. Thus, we may expect that we can effectively regard level-1 $X$ and $Z$ errors in an HLP as uncorrelated. Under this assumption, the $H$-LMS only requires $d_{1}$ logical readout gadgets to achieve full level-1 distance. We note that future numerical benchmarking is required to substantiate this. 

To guarantee the worst-case performance, we show in the following that an $H$-LMS, now consisting of $4d_{1}$ logical readout gadgets across $4d_1$ consecutive level-1 SE rounds, can reach similar logical performance even without the assumption above. We will use a modified level-1 decoding procedure (Algorithm~\ref{Alg: alternative level 1 decoding}). Based on this decoding procedure, we analyze the performance of this LMS without the above assumption. 

\begin{algorithm}
    \let\oldnl\nl
    \newcommand{\nonl}{\renewcommand{\nl}{\let\nl\oldnl}}
    \caption{Alternative level-1 decoding procedure for an $H$-LMS}\label{Alg: alternative level 1 decoding}
    \KwData{\justifying\small A level-1 detector error model configured for readout gadgets in an $H$-LMS and level-1 SE.}
    \KwIn{\justifying \small A level-1 syndrome configuration $\eta_{1}\in\mathbb{Z}_2^{|\mathsf{D}_{\mathsf{S}}\cup \mathsf{D}_{\mathsf{H}}|}$. Here, $\mathsf{D}_{\mathsf{S}}$ denotes the collection of level-1 detectors induced by level-1 SE; $\mathsf{D}_{\mathsf{H}}$ denotes the collection of level-1 detectors induced by the $H$-LMS.} 
    \KwOut{\justifying\small Inferred level-1 error configuration.}  
    \justifying\small 
    Let $\eta_{1,\mathsf{X}}$ and $\eta_{1,\mathsf{Z}}$ be the subconfigurations of $X$ syndromes and $Z$ syndromes in $\eta_{1}|_{\mathsf{D}_{\mathsf{S}}}$, respectively. Let $\mathsf{E}_{\mathsf{S,X}}$ and $\mathsf{E}_{\mathsf{S,Z}}$ be the collections of primitive level-1 $X$ and $Z$ errors, respectively, on cores and shuttle buses for level-1 SE.
    
   Restricted to syndromes in $\mathsf{D}_{\mathsf{S}}$, perform most-likely-error (MLE)~\cite{gottesman_fault_tolerant_2014,cain_correlated_2024,zhou_low_overhead_2025} decoding on $\mathsf{E}_{\mathsf{S,X}}$ with syndrome $\eta_{1,\mathsf{X}}$, yielding $\mathsf{g}_{\mathsf{X}}$. Similarly, perform MLE decoding on $\mathsf{E}_{\mathsf{S},\mathsf{Z}}$ with syndrome $\eta_{1,\mathsf{Z}}$, yielding $\mathsf{g}_{\mathsf{Z}}$. 

    Perform MWPM decoding on level-1 errors on $H$-buses with the level-1 syndrome configuration $\eta_{1}+\overline{\partial}_{1}\mathsf{g}_{\mathsf{X}}+\overline{\partial}_{1}\mathsf{g}_{\mathsf{Z}}$. Let $\mathsf{h}$ be the decoded result. (This is equivalent to decoding a repetition code.)

    \Return{$\mathsf{g}_{\mathsf{X}}+\mathsf{g}_{\mathsf{Z}}+\mathsf{h}$.}
\end{algorithm}

\begin{corollary}
    [Fault probability for an $H$-LMS]\label{corollary: fault probability for h lms}
    Consider an HLP run for $\mathsf{T}$ level-1 SE rounds with a collection of $X$-basis and $Z$-basis LMSs and $H$-LMSs. Every $H$-LMS uses $4d_1$ $H$-transformed logical readout gadgets. We assume perfect time boundaries for the HLP. We focus on a fixed $H$-LMS executed on the HLP. We set compilation parameters $\alpha_{\mathsf{c}}=1$, $\alpha_{\mathsf{b}}=d_{1}$, and $\alpha_{H}=2d_1$. We require that every $H$-transformed logical readout gadget is separated from all other logical gadgets in the same LMS, as well as from readout gadgets for level-1 SE, by at least $d_{0}$ level-0 time steps in the $X$-basis and $Z$-basis separations. We follow the notation in Theorem~\ref{thm: fault probability for LMSs}. Then, under the phenomenological depolarizing noise model with an error rate $p<1/(4r\mathfrak{e})^{10}$ at level 0 and with $d_{0}\geq d_{H}$, the probability of a logical measurement error for the $H$-LMS is upper bounded by $C_{H}(p/p_{th})^{d_{0}d_{1}/4}+2\mathrm{v}_{\mathsf{b}}(p/p_{th}')^{d_{0}d_{1}/2}$. Here, $d_H=2\log\left((4r_H\mathfrak{e})^8z_{\mathrm{max}}\right)/\log(p_{th}/p)$, $C_{H}=N_H\left((2r_H\mathfrak{e})^4\sqrt{z_{\mathrm{max}}}\right)^{d_1}$, $N_H=(n_1(\delta_1+2)+s_1)\mathsf{T}+8d_1$, $r_{H}=16\delta_1^2\max(\varpi_1,\varpi_L)$, and $p_{th}'=1/(2r\mathfrak{e})^{10}$.   
\end{corollary}
\begin{proof}
   Throughout this proof, `$H$ buses' refers to $H$-buses in the $H$-LMS, and `logical measurement error' refers to the logical measurement error for the $H$-LMS. Subscripts `$\epsilon$' and `$\kappa$' for a level-1 error configuration denote the induced-error component and the inferred-error component, respectively. We construct the level-1 error model and level-1 adjacency graph for $H$ buses and buses for level-1 SE. Let $\mathsf{D}_{\mathsf{S}}$ be the collection of level-1 detectors generated by level-1 SE. Similar to the proof of Theorem~\ref{thm: fault probability for LMSs}, each level-1 $X$ or $Z$ error triggers up to $2\delta_{1}$ level-1 detectors. On the other hand, each level-1 detector can be triggered by up to $2(\delta_1+1)\max(\varpi_1,\varpi_L)+4\leq 8\delta_{1}\max(\varpi_1,\varpi_L)$ primitive level-1 errors. Therefore, the vertex degree of each vertex in the level-1 adjacency graph is upper bounded by $r_H=16\delta_1^2\max(\varpi_1,\varpi_L)$. Moreover, the number of level-1 errors is upper bounded by $N_{H}=(n_1(\delta_1+2) +s_1)\mathsf{T}+8d_1$. 
   
   Consider a level-1 residual error configuration $\mathsf{f}:=\mathsf{f}_{\mathsf{X}}+\mathsf{f}_{\mathsf{Z}}+\mathsf{h}$ obtained from Algorithm~\ref{Alg: alternative level 1 decoding}. Here, $\mathsf{f}_{\mathsf{X}}$ and $\mathsf{f}_{\mathsf{Z}}$ consist of level-1 $X$ and $Z$ errors, respectively, on cores and shuttle buses for level-1 SE; $\mathsf{h}$ consists of level-1 errors on $H$ buses. Without loss of generality, we assume $\mathsf{f}$ is a connected cluster on the level-1 adjacency graph. Suppose $\mathsf{f}$ results in a logical measurement error. If $\mathsf{f}_{\mathsf{X}}$ (or $\mathsf{f}_{\mathsf{Z}}$) contains a weight-$\mathrm{w}$ connected cluster with no syndrome in $\mathsf{D}_{\mathsf{S}}$ and $\mathrm{w}\geq d_{1}$, then this cluster contains at least $\max(d_{1}/2,\mathrm{w}/2)$ induced level-1 $X$ (or $Z$) errors. On the other hand, we now consider the opposite case that all weight-$\mathrm{w}$ clusters in $\mathsf{f}_{\mathsf{X}}$ or $\mathsf{f}_{\mathsf{Z}}$ with no syndrome in $\mathsf{D}_{\mathsf{S}}$ must satisfy $\mathrm{w}<d_1$. According to Lemma~\ref{lemma: level 1 distance for LMS}, $\mathsf{f}_{\mathsf{X}}$ (or $\mathsf{f}_{\mathsf{Z}}$) is equivalent to a level-1 error configuration $\mathsf{h}_{\mathsf{X}}$ (or $\mathsf{h}_{\mathsf{Z}}$) consisting only of level-1 $X$ (or $Z$) errors on $H$ buses with $|\mathsf{h}_{\mathsf{X}}|\leq |\mathsf{f}_{\mathsf{X}}|$ (or $|\mathsf{h}_{\mathsf{Z}}|\leq|\mathsf{f}_{\mathsf{Z}}|$). According to the MWPM decoding at the third step of Algorithm~\ref{Alg: alternative level 1 decoding}, we obtain
    \begin{equation}
        |\mathsf{h}_{\kappa}|\leq |\mathsf{h}_{\epsilon}| + |\mathsf{h}_{\mathsf{X}}|  +|\mathsf{h}_{\mathsf{Z}}|
    \end{equation}
    Moreover, as $\mathsf{f}$ induces a logical measurement error, we have 
    \begin{equation}
        4d_1\leq |\mathsf{h}|+|\mathsf{h}_{\mathsf{X}}| + |\mathsf{h}_{\mathsf{Z}}|
    \end{equation}
    Note that $|\mathsf{h}_{\mathsf{X}}|\leq |\mathsf{f}_{\mathsf{X}}|\leq 2|\mathsf{f}_{\mathsf{X},\epsilon}|$ and similarly, $|\mathsf{h}_{\mathsf{Z}}|\leq 2|\mathsf{f}_{\mathsf{Z},\epsilon}|$. We can then bound the total weight of level-1 physical errors.
    \begin{align}
        d_1\leq& |\mathsf{h}_{\epsilon}|+|\mathsf{f}_{\mathsf{X},\epsilon}| + |\mathsf{f}_{\mathsf{Z},\epsilon}|=|\mathsf{f}_{\epsilon}|\nonumber\\
        |\mathsf{f}|=&|\mathsf{h}|+|\mathsf{f}_{\mathsf{X}}|+|\mathsf{f}_{\mathsf{Z}}|\leq 4|\mathsf{f}_{\epsilon}|
    \end{align}

    Let $\mathfrak{E}$ be a set of all level-1 residual error configurations that (i) correspond to a connected cluster on the level-1 adjacency graph and (ii) contain at least $\max(d_{1}/2,\mathrm{w}/8)$ level-1 $X$ induced errors or at least $\max(d_{1}/2,\mathrm{w}/8)$ level-1 $Z$ induced errors with $\mathrm{w}$ the weight of this cluster. Then, according to Lemma~\ref{lemma: counting lemma}, the number of weight-$\mathrm{w}$ elements in $\mathfrak{E}$ is upper bounded by $N_{H}(r_H\mathfrak{e})^{\mathrm{w}-1}$. According to our analysis above, the event of a logical measurement error is contained in the event of inducing a residual configuration in $\mathfrak{E}$. We now discard the rare event in Corollary~\ref{corollary: extension of error reduction thm}, so that the induced level-1 $X$ (or $Z$) errors are local stochastic with a level-1 error rate upper bounded by $z_{\mathrm{max}}(p/p_{th})^{d_{0}/2}$.  By Lemma~\ref{lemma: sufficient condition for bounding level 1 error}, we upper bound the probability of inducing an element in $\mathfrak{E}$ as follows.
    \begin{align}
        P(\mathfrak{E})\leq N_H((2r_H\mathfrak{e})^8p_1)^{d_1/2},
    \end{align}
    where we used $d_0\geq d_{H}$ to guarantee $p_{1}\leq 1/(4r_H\mathfrak{e})^8$. 
    Therefore, accounting for the discarded event, the probability for a logical measurement error is upper bounded by $C_{H}(p/p_{th})^{d_{0}d_{1}/4} + 2\mathrm{v}_{\mathsf{b}}(p/p_{th}')^{d_{0}d_{1}/2}$. 
\end{proof}

From the above corollary that upper bounds the probability of logical measurement error for a single $H$-LMS, we can follow Theorem~\ref{thm: fault probability for LMSs} to bound the probability of logical measurement errors for multiple LMSs (possibly including $X$-basis LMSs, $Z$-basis LMSs, and $H$-LMSs) on an HLP. We omit the exact error bound computation for brevity.  

\subsection{Extension module II: interfacing with external cores }
As shown in Fig.~\ref{fig: logical pauli measurement}(\textbf{f}), we can perform a transversal CNOT gate with an external core (a distance-$d_{ext}$ RSC with $d_{0}\leq d_{ext}\leq d_{0}d_{1}$) as the control and a $Z$ bus as the target. We call such a gate an extended core-bus CNOT gate. Similarly, we can perform an extended bus-core CNOT gate between an $X$ bus and an external core. Equipped with these extended hybrid-unit CNOT gates, the previously developed readout gadgets can jointly measure logical qubits on an HLP and on external cores. Note that we only allow external cores to directly interact with shuttle buses but not internal cores; therefore, errors on external cores do not propagate to internal cores. Consider a simple setting where external cores are only allowed to interact with an HLP via readout gadgets but not with other external cores. Then, level-0 circuits on external cores have the same structure as those on internal cores. In this way, we can apply the arguments in Sec.~\ref{sec: approximate level 1 error reduction} and Sec.~\ref{subsec: extension model measure general logical operators} to bound level-1 error rates and the probability of logical measurement errors for LMSs in the presence of external cores. By analogy with Theorem~\ref{thm: approximate error reduction} and Corollary~\ref{corollary: extension of error reduction thm}, we expect that (i) level-1 $X$ (or $Z$) errors are still local stochastic, (ii) the level-1 error rate on internal cores and shuttle buses scales as $\sim p^{d_{0}/2}$, and (iii) the level-1 error rate on external cores scales as $\sim p^{d_{ext}/2}$. Similarly, by analogy with Theorem~\ref{thm: fault probability for LMSs} and Corollary~\ref{corollary: fault probability for h lms}, we expect the probability of logical measurement errors to be upper bounded by both $\sim p^{d_{ext}/2}$ and $\sim p^{d_{0}d_{1}/4}$. More generally, consider a hybrid architecture consisting of HLPs and external cores. Each HLP can perform its own logical Clifford operations with previously introduced LMSs and simultaneously interface with a subcollection of external cores. External cores that are currently not interfacing with any HLP can prepare special logical states (such as magic states) and participate in transversal logical Clifford gates. In this way, HLPs and external cores can collaborate to execute logical quantum computation (Fig.~\ref{fig: hybrid arch scheduling}). We leave a careful theoretical and numerical analysis of this hybrid architecture for future work.  

\begin{figure}
    \centering
    \includegraphics[width=\columnwidth]{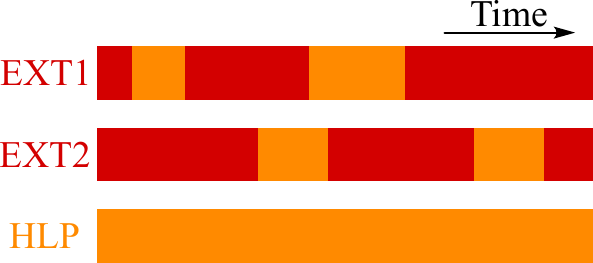}
    \caption{Scheduling of a hybrid architecture consisting of hierarchical logical processors (HLPs) and external cores. One HLP and two external cores are shown here. Red blocks denote transversal logical operations and state preparations on external cores only. Orange blocks for an external core indicate interfacing the core with an HLP via joint logical Pauli measurements.}
    \label{fig: hybrid arch scheduling}
\end{figure}

\section{Numerical Simulation}\label{sec: details on numerical simulation}
In this section, we present concrete HLP designs and provide details of numerical simulations and performance extrapolation. 

\subsection{Details on HLP constructions and circuit-level simulation}
\begin{figure}
    \centering
    \includegraphics[width=\columnwidth]{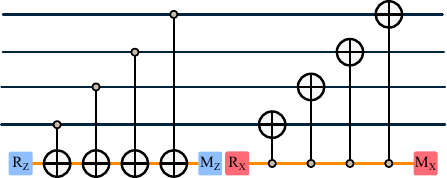}
    \caption{Level-1 SE circuit for the $[[4,2,2]]$ Iceberg code. The level-1 $Z$ and $X$ stabilizers are extracted sequentially with one shuttle bus at a time.}
    \label{fig: iceberg se}
\end{figure}

As mentioned in the main text, we can either use one shuttle bus at a time (Fig.~\ref{fig: iceberg se}) or two shuttle buses concurrently to implement the level-1 SE for an HLP based on the Iceberg code. 
For a $[[n^2,n^2-4n+2,4]]$ Square Berg code ($n$ is divisible by 4)~\cite{gidney_yoked_2025}, we measure $4n$ stabilizers in each level-1 SE round. As shown in Fig.~\ref{fig: hlp performance}(\textbf{a}), among these stabilizers, there is one $X$ stabilizer and one $Z$ stabilizer associated with each row and each column. To perform a round of level-1 SE on an HLP based on a Square Berg code, we can use $n$ shuttle buses at a time to measure the $Z$ stabilizers of all columns, followed by $X$ stabilizers of all columns, $Z$ stabilizers of all rows, and finally $X$ stabilizers of all rows. We use this level-1 SE design for the two HLPs based on the $[[64,34,4]]$ and $[[256,194,4]]$ Square Berg codes in Fig.~\ref{fig: hlp performance}. 

\begin{figure}
    \centering
    \includegraphics[width=\columnwidth]{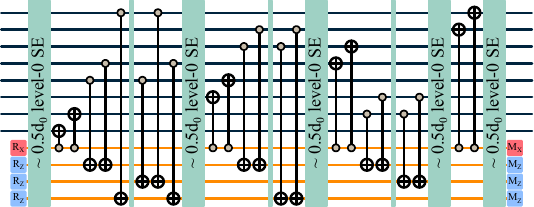}
    \caption{Scheduling of three $Z$-basis logical readout gadgets inside a level-1 SE round corresponding to the numerical results in Fig.~\ref{fig: logical pauli measurement}(\textbf{c}).  Each thin teal strip represents a single synchronized level-0 SE round on all working units.}
    \label{fig: ppm scheduling}
\end{figure}

For simulating logical memories with perfect time boundaries, we follow the approach in Ref.~\cite{gidney_yoked_2025} to simultaneously detect logical $X$ and $Z$ errors. More specifically, for every logical qubit, we associate it with a perfect physical qubit, referred to as the register qubit. For each logical qubit $i$, let $\overline{\mathsf{X}}_{i}$ and $X_i$ denote the logical $X$ operator and the Pauli $X$ operator on the associated register qubit, respectively. Similarly, let $\overline{\mathsf{Z}}_{i}$ and $Z_i$ be the logical $Z$ operator and the Pauli $Z$ operator on the register qubit, respectively. To perfectly initialize the logical memory, we first measure all (level-0 and level-1) stabilizers perfectly, then we measure $\overline{\mathsf{X}}_{i}\otimes X_{i}$ and $\overline{\mathsf{Z}}_{i}\otimes Z_{i}$ perfectly for every logical qubit $i$, thereby effectively creating a Bell pair between every logical qubit and its register qubit. After initialization, register qubits are left idling without noise, while qubits and gates in logical memory are subject to circuit-level noise. Finally, we repeat the perfect measurement procedure used in initialization. In this way, both logical $X$ and $Z$ errors can be detected simultaneously. If the logical qubit $i$ is initialized or measured by measuring both $\overline{\mathsf{X}}_i\otimes X_i$ and $\overline{\mathsf{Z}}_i\otimes Z_i$, then we say this logical qubit is initialized or measured in the Bell basis. On the other hand, if a logical qubit $i$ is initialized or measured by measuring $\overline{\mathsf{X}}_i$ (or $\overline{\mathsf{Z}}_i$), then we say this logical qubit is initialized or measured in the $X$ (or $Z$) basis. 

In Fig.~\ref{fig: logical pauli measurement}(\textbf{c}), we present circuit-level simulations for benchmarking logical Pauli $Z$ measurements (implemented with LMSs) on an HLP based on the $[[8,6,2]]$ Iceberg code with compilation parameters $\alpha_{\mathsf{c}}=1$ and $\alpha_{\mathsf{b}}=0.5$. We simulated three circuits, all having perfect time boundaries and running the HLP for five level-1 SE rounds, where the level-1 $X$ and $Z$ stabilizers are measured sequentially, as in Fig.~\ref{fig: iceberg se}. The first is a memory baseline circuit with six logical qubits initialized and measured in the Bell basis. Here, each logical qubit $i\in\{1,\cdots, 6\}$ is assigned a logical $X$ operator $\overline{\mathsf{X}}_i=\mathsf{X}_{1}\otimes \mathsf{X}_{i+1}$ and a logical $Z$ operator $\overline{\mathsf{Z}}_i=\mathsf{Z}_{i+1}\otimes \mathsf{Z}_{8}$. The second and the third circuits contain one and three LMSs, respectively. Since the LMS in the second circuit exactly corresponds to one LMS in the third circuit, we elaborate on the latter in the following. As described in the main text, the three LMSs measure three logical $Z$ operators: $\overline{\mathsf{Z}}_1\otimes\cdots \otimes\overline{\mathsf{Z}}_6$, $\overline{\mathsf{Z}}_{1}\otimes \overline{\mathsf{Z}}_2 \otimes \overline{\mathsf{Z}}_3$, and $\overline{\mathsf{Z}}_4\otimes \overline{\mathsf{Z}}_5\otimes\overline{\mathsf{Z}}_6$, respectively. We define a new set of logical qubits with the following logical $X$ and $Z$ operators.
\begin{equation}
    \begin{array}{ll}
        \overline{\mathsf{Z}}'_1:=\overline{\mathsf{Z}}_1\otimes\cdots\otimes \overline{\mathsf{Z}}_{6}, & \overline{\mathsf{X}}_1':=\overline{\mathsf{X}}_{6}\\
        \overline{\mathsf{Z}}'_2:=\overline{\mathsf{Z}}_1\otimes\overline{\mathsf{Z}}_2\otimes \overline{\mathsf{Z}}_{3}, & \overline{\mathsf{X}}_2':=\overline{\mathsf{X}}_1\otimes\overline{\mathsf{X}}_{6}\\
        \overline{\mathsf{Z}}'_3:=\overline{\mathsf{Z}}_2, & \overline{\mathsf{X}}_3':=\overline{\mathsf{X}}_1\otimes\overline{\mathsf{X}}_{2}\\
        \overline{\mathsf{Z}}'_4:=\overline{\mathsf{Z}}_3, & \overline{\mathsf{X}}_4':=\overline{\mathsf{X}}_1\otimes\overline{\mathsf{X}}_{3}\\
         \overline{\mathsf{Z}}'_5:=\overline{\mathsf{Z}}_4, & \overline{\mathsf{X}}_5':=\overline{\mathsf{X}}_4\otimes\overline{\mathsf{X}}_{6}\\
         \overline{\mathsf{Z}}'_6:=\overline{\mathsf{Z}}_5, & \overline{\mathsf{X}}_6':=\overline{\mathsf{X}}_5\otimes\overline{\mathsf{X}}_{6}
    \end{array}
\end{equation}
We initialize and finally measure the first two logical qubits in the $Z$ basis and the remaining logical qubits in the Bell basis, such that (i) the LMS outcomes are predetermined by initialization and (ii) errors on logical degrees of freedom not measured by the LMSs can be detected. The two logical readout gadgets in each LMS are embedded in the third and fourth level-1 SE rounds, respectively. More specifically, each logical readout gadget is contained within the lifetime of an $X$ bus for measuring the level-1 $X$ stabilizer. The scheduling of three logical readout gadgets corresponding to the three LMSs inside a single level-1 SE round is shown in Fig.~\ref{fig: ppm scheduling}. We note one anomaly in the scheduling: a core-bus CNOT gate coupling the $Z$ bus on the bottom of Fig.~\ref{fig: ppm scheduling} to two cores is now split into two adjacent level-0 time steps to avoid multiple core-bus CNOT gates operating on the same core at the same time. This scheduling choice does not comply with the restrictions imposed in Sec.~\ref{sec: hlp basic modules}. However, these restrictions are primarily introduced to simplify the analysis and presentation rather than to represent hard architectural constraints; practical implementations generally provide additional scheduling degrees of freedom beyond those assumed in our analysis. The simulation results presented in Fig.~\ref{fig: logical pauli measurement}(\textbf{c}) further validate this point.

As mentioned in Sec.~\ref{subsec: detectors and decoding}, we use correlated level-0 decoding in our circuit-level simulations. We extract soft outputs at the end of this procedure. We note that correlated level-0 decoding reduces to correlated matching~\cite{fowler_optimal_2013} for an idling core.  

\subsection{Soft-output Simulation}\label{subsec: behaviors of soft outputs}
In this subsection, we describe a high-speed heuristic sampling method, soft-output simulation, that directly simulates a level-1 circuit by sampling level-1 errors and then performing level-1 decoding. This method is essentially the gap simulation method developed in Ref.~\cite{gidney_yoked_2025} with gap values replaced by soft outputs~\cite{meister_efficient_2024}. The level-1 circuit error model we use is more refined than the one in Ref.~\cite{gidney_yoked_2025}. 

\begin{algorithm}
    \let\oldnl\nl
    \newcommand{\nonl}{\renewcommand{\nl}{\let\nl\oldnl}}
    \caption{Soft-output simulation}\label{Alg: soft output simulation}
    \KwIn{\justifying \small A level-1 detector error model. Every level-1 error is associated with a canonical segment and is tagged with the following information: length (in terms of level-0 time steps), Pauli basis, and working-unit type of the associated segment. } 
    \KwData{\justifying\small A probability distribution for the soft output for every canonical segment in the level-1 detector error model; such a probability distribution depends only on the tagged information of a canonical segment. A level-1 error rate ansatz $\mathrm{LER}:\mathbb{R}_{\geq0}\to\mathbb{R}_{\geq0}$ that maps the soft output on a canonical segment to an expected level-1 error rate.}
    \KwOut{\justifying\small A sampled level-1 error configuration and a decoded level-1 error configuration.}  
    \justifying\small Initialize two empty registers $\mathsf{f}$ and $\mathsf{g}$, each of which stores a level-1 error configuration. Initialize a level-1 decoding hypergraph $\mathsf{G}$ with the vertex set as the collection of all level-1 detectors and the hyperedge set as the collection of all level-1 primitive errors in the level-1 detector error model, such that the endpoints of each hyperedge $\mathsf{e}$ are exactly the level-1 syndrome of $\mathsf{e}$. 

    \For{each primitive level-1 error $\mathsf{e}$ in the error model}{
        \nonl Sample a soft output $\phi$ according to the soft-output distribution provided for the canonical segment associated with $\mathsf{e}$.
        
        \nonl Assign the weight of the hyperedge in $\mathsf{G}$ corresponding to $\mathsf{e}$ to $\phi$.

        \nonl Sample the occurrence of $\mathsf{e}$ based on the error probability $\mathrm{LER}(\phi)$. If the sampled result is positive, add $\mathsf{e}$ to $\mathsf{f}$.
    }

    Perform most-likely-error (MLE)~\cite{gottesman_fault_tolerant_2014,cain_correlated_2024,zhou_low_overhead_2025} decoding on $\mathsf{G}$ with the level-1 syndrome configuration $\overline{\partial}_1\mathsf{f}$. The decoding result is added to $\mathsf{g}$. \Return{$\mathsf{f}$ and $\mathsf{g}$.}
\end{algorithm}

The soft-output simulation procedure is described in Algorithm~\ref{Alg: soft output simulation}, which requires two key ingredients. First, we need to efficiently construct the soft-output distribution for canonical segments with various lengths, Pauli basis, and working-unit types under different values of $d_{0}$. Secondly, we need to construct the function $\mathrm{LER}$ that infers the level-1 error probability for a canonical segment based on its soft output. We elaborate on these ingredients in the following. 
\begin{figure}
    \centering
    \includegraphics[width=\columnwidth]{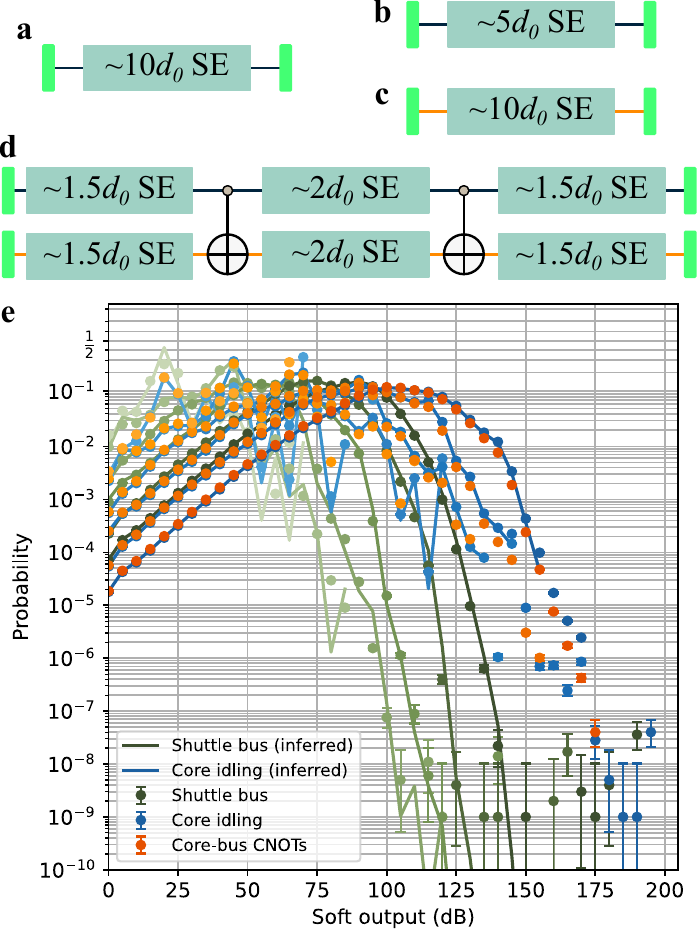}
    \caption{Soft-output reference and its validation. (\textbf{a}\textendash{}\textbf{d}) Circuits used to generate the reference soft-output distribution and to validate the soft-output inference method. Each circuit has perfect time boundaries. For each simulation shot of the circuits in (\textbf{a}), (\textbf{b}), and (\textbf{d}), we extract a pair of soft outputs, one from the $X$-basis segment and one from the $Z$-basis segment, with both segments spanning all level-0 SE rounds on the core. For the circuit in (\textbf{c}), we extract the soft output on the $X$-basis segment on the $X$ bus. (\textbf{e}) $X$-basis soft outputs for segments in (\textbf{b}\textendash{}\textbf{d}). Results are shown for distance $d_{0}$ ranging from $3$ to $8$, with darker shades corresponding to larger $d_{0}$. Soft outputs sampled from circuit-level simulations are shown as dots; inferred soft outputs are shown as curves.}
    \label{fig: so distribution testing}
\end{figure}

\textit{Soft-output distribution}\textemdash{}We first construct a reference soft-output distribution on core segments with a fixed length of $10d_{0}$ for each core size $d_{0}\in\{3,4,\cdots,12\}$. To achieve this, we perform circuit-level simulation of an idling circuit on a single core with $10d_{0}$ level-0 SE rounds (Fig.~\ref{fig: so distribution testing}(\textbf{a})). For each simulation shot, we extract a pair of soft outputs~\cite{meister_efficient_2024}, one on the $X$-basis and one on the $Z$-basis segment, where both segments span all level-0 SE rounds. From these reference distributions, we can heuristically infer soft-output distributions for canonical segments on cores and shuttle buses~\cite{gidney_yoked_2025}. Denote the soft-output distribution on an $X$-basis (or $Z$-basis) core segment with a length of $t$ level-0 time steps and a core size $d_{0}$ as $P_{X,t,d_0}$ (or $P_{Z,t,d_{0}}$). Define the corresponding complementary cumulative distribution function (CCDF) $\overline{F}_{A,t,d_{0}}$ ($A\in\{X,Z\}$) by $\overline{F}_{A,t,d_0}(\phi):=\sum_{\phi'\geq\phi}P_{A,t,d_{0}}(\phi')$. Following Ref.~\cite{gidney_yoked_2025}, from our reference CCDF $\overline{F}_{A,10d_{0},d_{0}}$ (obtained from the reference distribution above), we heuristically infer another CCDF $\overline{F}_{A,t',d_{0}}$ corresponding to a segment length $t'$ as 
\begin{equation}\label{eq: CCDF inference}
    \overline{F}_{A,t',d_{0}}(\phi):=\left(\overline{F}_{A,10d_{0},d_{0}}(\phi)\right)^{\frac{t'}{10d_{0}}}. 
\end{equation}
As a CCDF uniquely determines a distribution, we can obtain from Eq.~\ref{eq: CCDF inference} the soft-output distribution $P_{A,t',d_{0}}$. Similarly, consider an $X$ (or $Z$) bus with a width of $d_{0}$ and a length of $d_{0}d_{1}$. Denote the soft-output distribution on an $X$-basis (or $Z$-basis) segment with a length of $t$ as $P_{A,t,d_{0},d_{1}}$ with $A=X$ (or $A=Z$); denote the corresponding CCDF as $\overline{F}_{A,t,d_{0},d_{1}}$, which is inferred as  
\begin{equation}
    \overline{F}_{A,t,d_{0},d_{1}}(\phi) := \left(\overline{F}_{A,10d_{0},d_{0}}(\phi)\right)^{\frac{td_1}{10d_{0}}}.
\end{equation}
We test the validity of this heuristic inference method by comparing the inferred soft-output distributions with the numerically sampled ones for the circuits in Fig.~\ref{fig: so distribution testing}(\textbf{b}\textendash{}\textbf{d}) and observe close agreement between the two (Fig.~\ref{fig: so distribution testing}(\textbf{e})). 

\begin{figure}
    \centering
    \includegraphics[width=\columnwidth]{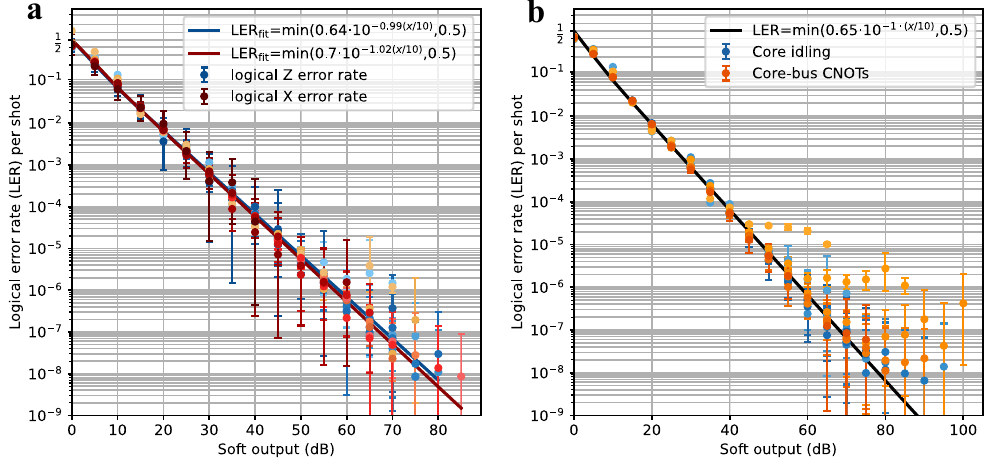}
    \caption{Fitting the level-1 error probability ansatz. (\textbf{a}) $X$-basis (or $Z$-basis) soft outputs and their corresponding logical $Z$ (or $X$) error rates for the circuit in Fig.~\ref{fig: so distribution testing}(\textbf{a}) are used to fit the level-1 error probability ansatz in Eq.~\ref{eq: ler ansatz for soft outputs}. Results are shown for $d_{0}$ ranging from $3$ to $12$, with darker shades corresponding to larger $d_{0}$. (\textbf{b}) $X$-basis soft outputs and their corresponding logical $Z$ error rates on a core for circuits in Fig.~\ref{fig: so distribution testing}(\textbf{b}) and Fig.~\ref{fig: so distribution testing}(\textbf{d}) are shown as blue and orange dots, respectively. The ansatz function $\mathrm{LER}_{a,b}$ with $a=0.65$ and $b=1$ is also shown for comparison. Here, results are shown for $d_{0}$ ranging from $3$ to $8$.}
    \label{fig: ler ansatz from soft output}
\end{figure}
\begin{figure}
    \centering
    \includegraphics[width=\columnwidth]{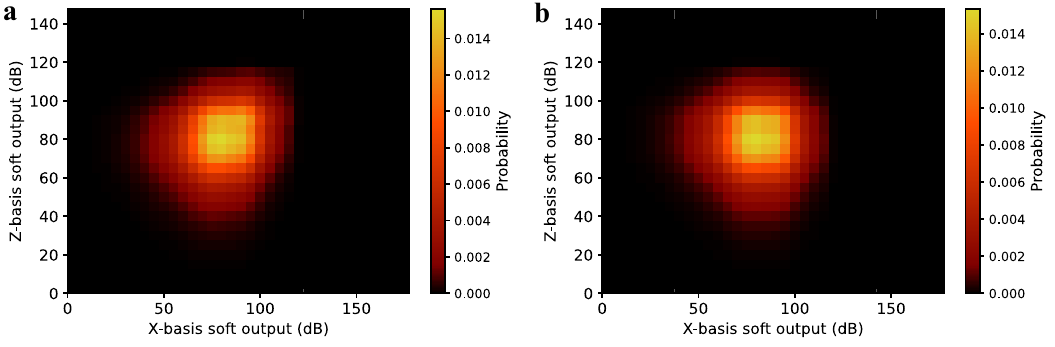}
    \caption{Correlation between $X$-basis and $Z$-basis soft outputs for the idling core in Fig.~\ref{fig: so distribution testing}(\textbf{a}) with distance $d_{0}=7$. (\textbf{a}) Joint probability distribution of the $X$-basis and $Z$-basis soft outputs. (\textbf{b}) Product of the marginal probability distributions for $X$-basis and $Z$-basis soft outputs. }
    \label{fig: xz correlation}
\end{figure}
\begin{figure}
    \centering
    \includegraphics[width=\columnwidth]{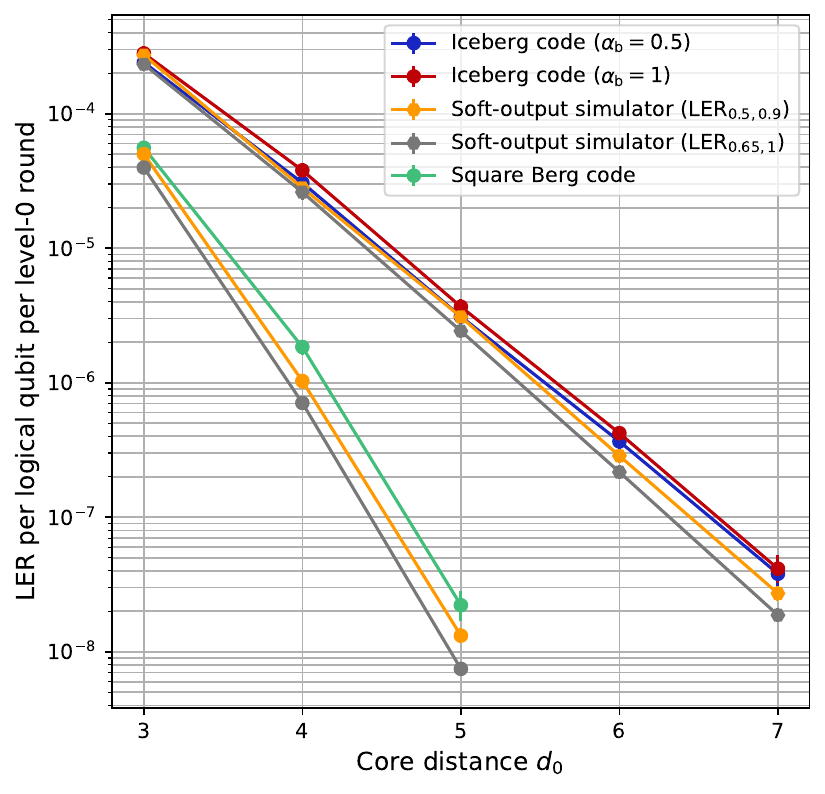}
    \caption{Calibration of soft-output simulation. See also Fig.~\ref{fig: hlp performance}(\textbf{b}) for an illustration of a subset of data points here. By default, compilation parameters $\alpha_{\mathsf{c}}$ and $\alpha_{\mathsf{b}}$ are both set to be $1$. Additionally, we simulate the HLP based on the $[[4,2,2]]$ Iceberg code at the circuit level with $\alpha_{\mathsf{b}}=0.5$.}
    \label{fig: so calibration}
\end{figure}

\textit{Level-1 error rate estimation from soft outputs}\textemdash{}We use all sampling results for the core-idling circuit in Fig.~\ref{fig: so distribution testing}(\textbf{a}) with $d_{0}\in\{3,\cdots,12\}$ to fit the following parametrized ansatz for estimating logical $X$ (or $Z$) error rates from $Z$-basis (or $X$-basis) soft outputs:
\begin{equation}\label{eq: ler ansatz for soft outputs}
    \mathrm{LER}_{a,b}(\phi):=\min(a\cdot 10^{-b\cdot\phi/10},0.5),
\end{equation}
where $a$ and $b$ are non-negative fitting parameters. We find that $a=0.65$ and $b=1$ provides a close fit for both logical $X$ and $Z$ error rates (Fig.~\ref{fig: ler ansatz from soft output}(\textbf{a})). We compare this fitted ansatz with data pairs consisting of $X$-basis soft outputs and logical $Z$ error rates on a core for circuits in Fig.~\ref{fig: so distribution testing}(\textbf{b}) and Fig.~\ref{fig: so distribution testing}(\textbf{d}). The fitted ansatz matches the data points well for the idling circuit in Fig.~\ref{fig: so distribution testing}(\textbf{b}), while data points for the circuit in Fig.~\ref{fig: so distribution testing}(\textbf{d}) containing core-bus CNOT gates exhibit some deviation at large soft output values for small $d_{0}$\textemdash{}more specifically, logical $Z$ error rates plateau at large soft output values. We attribute the deviation, at least in part, to logical $Z$ errors on the $Z$ bus between the two core-bus CNOT gates.

\begin{figure*}
    \centering
    \includegraphics[width=\textwidth]{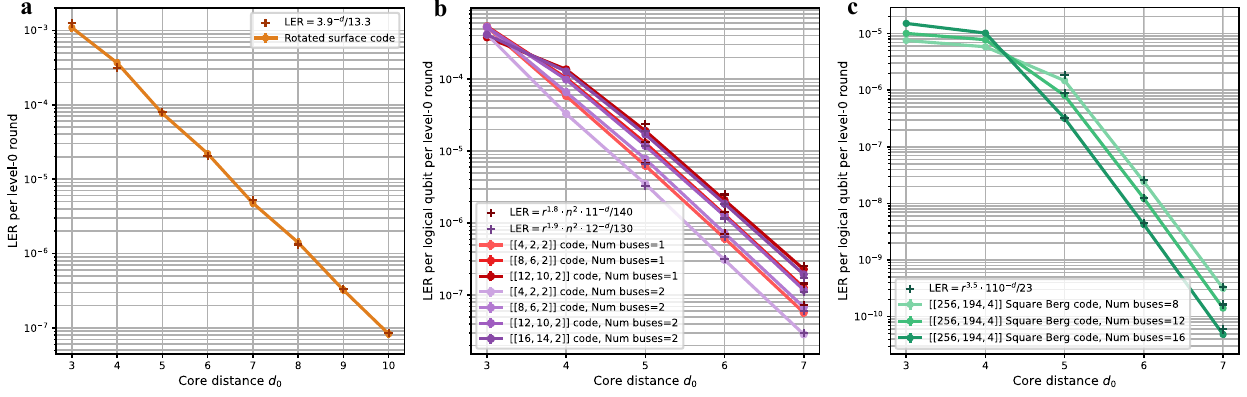}
    \caption{Extrapolation of logical error rates per level-0 SE round for the rotated surface code (RSC) and hierarchical logical processors (HLPs). (\textbf{a}) Logical error rates of an idling RSC with distances ranging from $3$ to $10$. Data points are obtained from circuit-level simulations with a correlated-matching decoder~\cite{fowler_optimal_2013}. (\textbf{b}) Logical error rates of HLPs based on Iceberg codes with one or two shuttle buses operating concurrently. (\textbf{c}) Logical error rates of HLPs based on the $[[256,194,4]]$ Square Berg code with different limits on the number of concurrently operating shuttle buses during level-1 SE. Data points in (\textbf{b}\textendash{}\textbf{c}) are obtained from soft-output simulations. To account for the gap between the soft-output and the circuit-level results observed in Fig.~\ref{fig: so calibration}, we multiply the soft-output logical error rates by a factor of two when estimating the logical error rates of HLPs. For the HLPs, logical error rates are fitted per level-1 SE round and converted to logical error rates per logical qubit per level-0 SE round for plotting.}
    \label{fig: hlp ler extrapolation}
\end{figure*}

Soft outputs collected for the idling core (Fig.~\ref{fig: so distribution testing}(\textbf{a})) also allow a fine-grained study of correlations between logical $X$ and $Z$ errors on the core. From these data, we obtain, for each core distance $d_{0}\in\{3,\cdots,12\}$, both the joint distribution of $X$-basis and $Z$-basis soft outputs and the product of their marginal distributions. We observe close agreement between the joint distribution and the product distribution for all $d_{0}\in\{4,\cdots,12\}$, indicating that the $X$-basis and $Z$-basis soft outputs are approximately independent. See Fig.~\ref{fig: xz correlation}(\textbf{a}\textendash{}\textbf{b}) for an illustration of the $d_{0}=7$ case. Since the $X$-basis and $Z$-basis soft outputs accurately predict probabilities for logical $Z$ and $X$ errors, respectively, this independence suggests that the logical error probability in one basis is insensitive to the component of physical-error configuration in the other basis, at the resolution captured by soft outputs. More specifically, we group physical-error configurations according to the $Z$-basis soft outputs determined by the $X$-basis components of those configurations. Across these groups, the logical $Z$ error probability remains approximately unchanged. Conversely, we group physical-error configurations according to the $X$-basis soft outputs determined by the $Z$-basis components of those configurations. Across these groups, the logical $X$ error probability remains approximately unchanged. Ref.~\cite{gidney_yoked_2025} observed that the overall logical $X$ and $Z$ errors are approximately uncorrelated on an idling core. The fine-grained study above extends this observation to soft-output-conditioned logical error probabilities. 
As discussed in Sec.~\ref{subsec: extension model measure general logical operators}, if level-1 $X$ and $Z$ errors are uncorrelated, we expect an $H$-LMS with $d_1$ $H$-transformed logical readout gadgets to achieve full level-1 distance. Without this assumption, Corollary~\ref{corollary: fault probability for h lms} provides a worst-case guarantee for an $H$-LMS with $4d_1$ such gadgets. The observation that level-1 $X$ and $Z$ errors on an idling core are approximately uncorrelated provides a starting point for future studies of correlations among level-1 $X$ and $Z$ errors on an HLP, as well as the optimization of $H$-transformed readout gadgets and $H$-LMSs.

As described in the main text (Fig.~\ref{fig: hlp performance}(\textbf{b})), we test the accuracy of soft-output simulation on the following two HLPs:
\begin{enumerate}
    \item An HLP with the $[[4,2,2]]$ Iceberg code as the level-1 code. The HLP is run as a memory for ten level-1 SE rounds. 
    \item An HLP with the $[[64,34,4]]$ Square Berg code as the level-1 code. We simulate only five level-1 SE rounds on the HLP, since circuit-level simulation is costly at this scale.    
\end{enumerate}
As shown in Fig.~\ref{fig: so calibration}, the soft-output simulations using the fitted level-1 error rate ansatz $\mathrm{LER}_{0.65,1}$ produce results close to those obtained from circuit-level simulations. However, they consistently underestimate the logical error rates, and the gap grows slightly with increasing core distance $d_{0}$. To reduce this dependence on $d_{0}$, we use a slightly tuned ansatz $\mathrm{LER}_{0.5,0.9}$, which yields slightly higher estimates of level-1 error rates than $\mathrm{LER}_{0.65,1}$ at large soft-output values. With this ansatz, the gap between the soft-output and circuit-level simulations remains almost constant for both HLPs as $d_{0}$ increases. We therefore use $\mathrm{LER}_{0.5,0.9}$ for the soft-output simulations of the other HLPs in the following subsection. Moreover, to account for the remaining gap, we multiply the logical error rates obtained from the soft-output simulations by a factor of two when estimating the logical error rates of HLPs.    

\subsection{Extrapolation for logical error rates}

In this subsection, we describe how we extrapolate the performance of the RSC and HLPs to obtain the qubit and time overhead estimation in Fig.~\ref{fig: hlp performance}(\textbf{c}\textendash{}\textbf{d}). For the RSC, we perform circuit-level simulation of the idling circuit in Fig.~\ref{fig: so distribution testing}(\textbf{a}) with a correlated-matching decoder~\cite{fowler_optimal_2013} and distance $d_{0}$ ranging from $3$ to $10$. We fit the logical error rates (per SE round) as a function of the code distance to $\mathrm{LER}_{\mathrm{RSC}}^{a,b}(d_{0})=a^{-d_{0}}/b$ (Fig.~\ref{fig: hlp ler extrapolation}(\textbf{a})). For HLPs based on $[[n_1,n_1-2,2]]$ Iceberg codes with a single shuttle bus operating at a time, we use soft-output simulation (of $10$ level-1 SE rounds) to estimate their logical error rates for code sizes $n_1\in\{4,8,12\}$ and core distance $d_{0}$ ranging from $3$ to $7$ (Fig.~\ref{fig: hlp ler extrapolation}(\textbf{b})). Using data with $d_{0}\in\{5,6,7\}$, we fit the logical error rates per level-1 SE round to $\mathrm{LER}_{\mathrm{ICE}}^{a,b,c}(d_{0},n_1,r)=r^{c}\cdot n_1^2\cdot a^{-d_{0}}/b$, where $r$ is the length of each level-1 SE round~\cite{gidney_yoked_2025}. We expect the fitted parameter $c$ to be close to the Iceberg-code distance~\cite{gidney_yoked_2025}. Similarly, we use soft-output simulation to study HLPs based on Iceberg codes with two shuttle buses operating concurrently, and perform the same procedure described above (Fig.~\ref{fig: hlp ler extrapolation}(\textbf{b})). For HLPs based on the $[[256,194,4]]$ Square Berg code with different numbers of concurrent shuttle buses, we perform soft-output simulation of 10 level-1 SE rounds and use the data with $d_{0}\in\{5,6,7\}$ to fit the logical error rates per level-1 SE round to $\mathrm{LER}_{\mathrm{SQR}}^{a,b,c}(d_{0},r)=r^{c}\cdot a^{-d_{0}}/b$ (Fig.~\ref{fig: hlp ler extrapolation}(\textbf{c})). Based on these fitted logical error rate ansatzes, we can determine the core distance $d_{0}$ required to attain a target logical error rate. This then allows us to estimate the qubit overhead per logical qubit and time overhead per level-1 SE round.

\end{document}